\documentclass[11pt, a4paper]{article}

\usepackage[margin=3cm]{geometry}
\usepackage{amsmath,amsthm,amssymb}
\usepackage{xcolor}
\usepackage{tikz}
\usetikzlibrary{arrows, arrows.meta, positioning, calc, fit, shapes.geometric}
\usepackage{adjustbox}
\usepackage{enumitem}
\usepackage{xspace}

\usepackage{thmtools}
\makeatletter
\@ifundefined{newcounteralias}{}{%
  \renewcommand\thmt@autorefsetup{%
    \@xa\def\csname\thmt@envname autorefname\@xa\endcsname\@xa{\thmt@thmname}}%
}
\makeatother

\usepackage{hyperref} 
\usepackage[capitalise]{cleveref}

\newcommand{\mygreen}{green!70!black}
\newcommand{\myred}{red!80!black}
\def\redb{\textcolor{\myred}{\circ}}
\def\greenb{\textcolor{\mygreen}{\bullet}}

\definecolor{myblue}{HTML}{0072B2}
\definecolor{myorange}{HTML}{E69F00}

\definecolor{defblue}{rgb}{0, 0.4, 0.796}
\newcommand{\defi}[1]{\textcolor{defblue}{\emph{#1}}}

\newtheorem{theorem}{Theorem}
\newtheorem{lemma}[theorem]{Lemma}
\newtheorem{observation}[theorem]{Observation}
\newtheorem{corollary}[theorem]{Corollary} 

\newtheorem{definition}[theorem]{Definition}
\newtheorem{remark}[theorem]{Remark}

\crefname{observation}{Observation}{Observations}
\Crefname{observation}{Observation}{Observations}

\newcommand{\opt}{\text{opt}\xspace}
\newcommand{\kopt}{\text{$k$-\opt}\xspace}

\tikzset{
  red ribbon/.style={draw=red,fill=red!7,line width=.9pt,
                     line join=round},
  blue ribbon/.style={draw=blue,fill=blue!7,line width=.9pt,
                      line join=round},
  tour edge/.style={draw=black!85,line width=1pt},
  endpoint/.style={font=\normalsize,inner sep=1pt,anchor=south},
  transition/.style={-{Stealth[length=2.6mm,width=1.6mm]},
                      draw=black!80,line width=.85pt},
}
\newcommand{\pathblock}[4]{%
  \path[#4]
    (#1,{#3+.16})
    .. controls ({#1+.25},{#3+.21}) and ({#2-.34},{#3+.12}) .. (#2,{#3+.16})
    -- (#2,{#3-.16})
    .. controls ({#2-.3},{#3-.20}) and ({#1+.26},{#3-.12}) .. (#1,{#3-.16})
    -- cycle;
}
\newcommand{\leftstub}[2]{%
  \path[red ribbon]
    (.12,{#2+.16})
    .. controls (.33,{#2+.13}) and ({#1-.28},{#2+.19}) .. (#1,{#2+.16})
    -- (#1,{#2-.16})
    .. controls ({#1-.25},{#2-.13}) and (.33,{#2-.19}) .. (.12,{#2-.16})
    .. controls (-.10,{#2-.105}) and (.30,{#2-.02}) .. (.12,{#2+.005})
    .. controls (-.07,{#2+.065}) and (.30,{#2+.115}) .. (.12,{#2+.16})
    -- cycle;
}
\newcommand{\rightstub}[2]{%
  \path[red ribbon]
    (#1,{#2+.16})
    .. controls ({#1+.27},{#2+.19}) and (6.00,{#2+.13}) .. (6.18,{#2+.16})
    .. controls (6.40,{#2+.105}) and (6.00,{#2+.04}) .. (6.18,{#2+.005})
    .. controls (6.40,{#2-.07}) and (6.00,{#2-.115}) .. (6.18,{#2-.16})
    .. controls (6.00,{#2-.13}) and ({#1+.27},{#2-.19}) .. (#1,{#2-.16})
    -- cycle;
}
\newcommand{\directionmarks}[1]{%
  \foreach \x in {1.48,1.92,2.36,2.80,3.24,3.68,4.12,4.56,5.00}{%
    \ifnum#1=1
      \draw[draw=blue,-{Stealth[length=1.7mm,width=1.25mm]},
            line width=.65pt] (\x-.12,0)--(\x+.12,0);
    \else
      \draw[draw=blue,-{Stealth[length=1.7mm,width=1.25mm]},
            line width=.65pt] (\x+.12,0)--(\x-.12,0);
    \fi
  }%
}

\newcommand{\email}[1]{\href{mailto:#1}{\texttt{#1}}}

\title{Closing All-Exp Gaps in Local Search for the Traveling Salesman Problem and Maximum Constraint Assignment}
\author{Sophia Heimann\thanks{Hertz Chair for Algorithms and Optimization, University of Bonn, Germany (\email{sheimann@uni-bonn.de})} 
\and Hung P. Hoang\thanks{Algorithms and Complexity Group, Faculty of Informatics, TU Wien, Austria, (\email{phoang@ac.tuwien.ac.at}) funded by the Austrian Science Fund (FWF, project 10.55776/ESP1136425)}
    \and Stefan Hougardy\thanks{Research Institute for Discrete Mathematics and Hausdorff Center for Mathematics, University of Bonn, Germany (\email{hougardy@or.uni-bonn.de}) funded by the Deutsche Forschungsgemeinschaft (DFG, German Research Foundation) under Germany's Excellence Strategy -- EXC-2047/1 -- 390685813}}
\date{November 2, 2026}

\begin{document}

\maketitle

\begin{abstract}
We prove the all-exp property for the Traveling Salesman Problem with the 2-\opt neighborhood, and for $(3,2,2)$- and $(2,3,2)$-Maximum Constraint Assignment with the flip neighborhood. The all-exp property states that infinitely many instances have an initial solution from which \emph{every} sequence of improving moves to a local optimum, and hence every run of local search, has exponential length. For the Traveling Salesman Problem, the all-exp property was previously known for every \kopt neighborhood with $k \ge 3$, so our result for 2-\opt completes the picture for all \kopt neighborhoods. In $(p,q,r)$-Maximum Constraint Assignment, each constraint involves at most $p$ variables, each variable occurs in at most $q$ constraints, and each domain has size $r$. Decreasing any of the parameters of $(3,2,2)$ or $(2,3,2)$ yields a problem in which every sequence of improving flips has polynomial length, so both Maximum Constraint Assignment results are tight. All three results settle long-standing open problems. Our proofs build on a common construction idea, which reveals a close connection between the local-search dynamics of the three problems.
\end{abstract}

\newpage
\tableofcontents

\section{Introduction}

Local search is a fundamental algorithmic approach for solving combinatorial optimization problems. Starting from an initial solution, a local-search algorithm repeatedly replaces the current solution with a better solution from a well-defined neighborhood, until a local optimum is reached. Despite the simplicity of this approach and its excellent practical performance for many problems, the worst-case behavior of local search can be surprisingly complicated. Johnson, Papadimitriou, and Yannakakis introduced the complexity class PLS to capture the difficulty of finding locally optimal solutions~\cite{JPY1988}. Schäffer and Yannakakis later defined tight PLS-reductions and showed that problems proven PLS-complete via such reductions exhibit particularly strong forms of hardness~\cite{SY1991}.

One such form of hardness concerns not the complexity of finding a local optimum but the number of iterations performed by local search. A first indication of such hardness is the \emph{is-exp property}: there are infinitely many instances, each with an initial solution, from which at least one sequence of improving moves to a local optimum has exponential length. This alone, however, does not imply that local search necessarily takes exponentially many iterations. Whenever several improving neighbors are available, a \emph{pivot rule} determines which one is chosen, and a suitable pivot rule may avoid the long sequence and reach a local optimum much faster. The strongest possible negative result is that a problem has the \emph{all-exp property}: there are infinitely many instances, each with an initial solution, such that \emph{every} sequence of improving moves from this initial solution to a local optimum has exponential length. Equivalently, even under an optimal pivot rule, local search started from these initial solutions requires exponentially many iterations. Every problem shown to be PLS-complete via a tight PLS-reduction has the all-exp property~\cite{SY1991}. However, the all-exp property can also be established by direct constructions, which work even for problems not known to be PLS-complete.

One of the most prominent local-search heuristics is \kopt for the Traveling Salesman Problem (TSP). 
The case $k=2$, called 2-\opt, was already studied by Croes in~1958~\cite{Cro1958} and is the oldest and simplest local-search method for the TSP. In each iteration, 2-\opt replaces two edges of the current tour with two other edges such that the resulting tour is shorter. The \kopt heuristic generalizes 2-\opt by replacing up to $k$ edges in each iteration. 
Both 2-\opt and \kopt for general $k$ have been studied extensively, theoretically as well as experimentally; see, for example,~\cite{Ben1992,JM1997,CKT1999,JM2002,HZZ2020,Zho2020}.

The TSP with the 2-\opt neighborhood has long been known to have the is-exp property; the first construction establishing this property is due to Lueker~\cite{Lue1975} (see also~\cite[page 1999]{CKT1999}). However, the is-exp property does not rule out a pivot rule that always reaches a local optimum after polynomially many iterations, and whether the TSP with the 2-\opt neighborhood also has the all-exp property has remained open.
The first all-exp result for the TSP with the \kopt neighborhood is due to Krentel~\cite{Kre1989} and holds for all sufficiently large constants $k$.
Hoang and Hougardy~\cite{HH2023} later showed that Krentel's construction requires $k\ge 14{,}208$ and, 
by refining the construction, reduced this bound to $k\ge 704$. In~2024, Heimann, Hoang, and Hougardy~\cite{HHH2024} substantially improved the bound and proved that the TSP with the \kopt neighborhood has the all-exp property for all $k\ge 5$. This result 
is based on the fact that Max-Cut with the flip neighborhood, in which a single vertex moves to the other side of the cut, has the all-exp property on graphs of bounded degree. 
For general graphs, Schäffer and Yannakakis~\cite{SY1991} already showed in~1991 that Max-Cut with this neighborhood is PLS-complete
via a tight reduction and hence has the all-exp property. In~2010, Monien and Tscheuschner~\cite{MT2010} established the all-exp property for Max-Cut with the flip neighborhood even on graphs of maximum degree four; a much simpler proof of this result was later given by Michel and Scott~\cite{MS2023}.

Very recently, Heimann, Hoang, and Hougardy~\cite{HHH2026} proved the all-exp property for the TSP with the 3-\opt and 4-\opt neighborhoods, and also with the 2.5-\opt neighborhood, which lies strictly between the 2-\opt and 3-\opt neighborhoods and was introduced in~\cite{Ben1992}. Consequently, the TSP with the \kopt neighborhood was known to have the all-exp property for every $k\ge 3$, 
and only the most basic and practically important case, 2-\opt, remained open. 
In this paper, we finally settle this case.

\begin{restatable}{theorem}{thmtwoopt}
\label{thm:2opt}
   The Traveling Salesman Problem with the 2-\opt neighborhood has the all-exp property. 
\end{restatable}

\Cref{thm:2opt} also holds for the metric TSP: adding a sufficiently large constant to all edge weights enforces the triangle inequality, while the gain of every 2-\opt move, and hence every improving sequence, remains unchanged. This is a standard observation; see, for example, Krentel~\cite{Kre1989}.

The results of~\cite{HHH2026} for 2.5-\opt, 3-\opt, and 4-\opt are based on reductions from Max-Cut with the flip neighborhood on graphs of bounded degree. 
However, extending these reductions to 2-\opt remains challenging. 
One of the main difficulties lies in the restricted structure of 2-\opt moves, as illustrated in \cref{fig:compare}.
For $k\geq 3$, the previous constructions based on reductions from Max-Cut with the flip neighborhood exploit moves that relocate an entire subpath while preserving its internal order; see \cref{fig:compare} (top).
A 2-\opt move is more constrained: it can only reverse the order of vertices along a subpath; see  \cref{fig:compare} (bottom).
The additional flexibility of 3-\opt makes it possible to encode states through the presence or absence of particular edges (e.g., the edge $ab$ in \cref{fig:compare} (top)) and to change these states by relatively independent local modifications. 
For 2-\opt, however, changing a particular adjacency may require reversing an entire subpath, potentially affecting the representation of other states. 
This makes it more difficult to implement the desired state transitions independently.

This structural restriction prevents a straightforward adaptation of the previous constructions, and it is a priori unclear whether a suitable choice of improving 2-\opt moves could always lead to a local optimum in polynomially many iterations.
Instead of extending the previous reductions, we therefore give a direct construction that exploits the restrictions imposed by 2-\opt to enforce exponentially long improving sequences under every pivot rule.

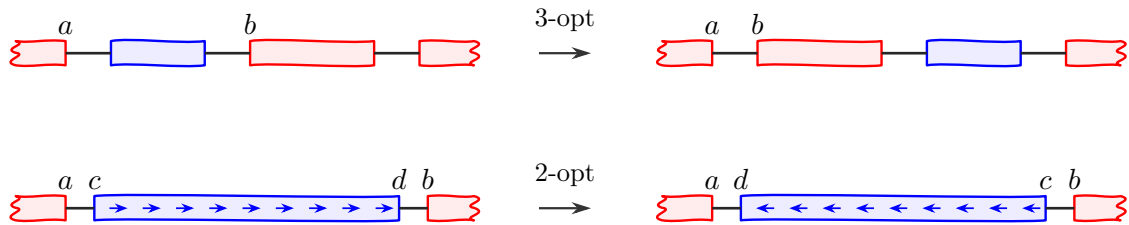
\begin{figure}
\centering
\begin{tikzpicture}[x=1cm,y=1cm]

\begin{scope}[shift={(-0.5,2.05)}]
    \leftstub{.78}{0}
    \pathblock{1.38}{2.62}{0}{blue ribbon}
    \pathblock{3.22}{4.86}{0}{red ribbon}
    \rightstub{5.46}{0}
    \draw[tour edge] (.78,0)--(1.38,0);
    \draw[tour edge] (2.62,0)--(3.22,0);
    \draw[tour edge] (4.86,0)--(5.46,0);
    \node[endpoint] at (.78,.21) {$a$};
    \node[endpoint] at (3.22,.21) {$b$};
  \end{scope}
  \begin{scope}[shift={(8.05,2.05)}]
    \leftstub{.78}{0}
    \pathblock{1.38}{3.02}{0}{red ribbon}
    \pathblock{3.62}{4.86}{0}{blue ribbon}
    \rightstub{5.46}{0}
    \draw[tour edge] (.78,0)--(1.38,0);
    \draw[tour edge] (3.02,0)--(3.62,0);
    \draw[tour edge] (4.86,0)--(5.46,0);
    \node[endpoint] at (.78,.21) {$a$};
    \node[endpoint] at (1.38,.21) {$b$};
  \end{scope}

  \draw[transition] (6.54,2.05)--(7.23,2.05);
  \node[font=\small,anchor=south] at (6.885,2.21) {3-\opt};

  \begin{scope}[shift={(-0.5,0)}]
    \leftstub{.78}{0}
    \pathblock{1.16}{5.19}{0}{blue ribbon}
    \rightstub{5.57}{0}
    \draw[tour edge] (.78,0)--(1.16,0);
    \draw[tour edge] (5.19,0)--(5.57,0);
    \directionmarks{1}
    \node[endpoint] at (.78,.21) {$a$};
    \node[endpoint] at (1.16,.21) {$c$};
    \node[endpoint] at (5.19,.21) {$d$};
    \node[endpoint] at (5.57,.21) {$b$};
  \end{scope}
  \begin{scope}[shift={(8.05,0)}]
    \leftstub{.78}{0}
    \pathblock{1.16}{5.19}{0}{blue ribbon}
    \rightstub{5.57}{0}
    \draw[tour edge] (.78,0)--(1.16,0);
    \draw[tour edge] (5.19,0)--(5.57,0);
    \directionmarks{-1}
    \node[endpoint] at (.78,.21) {$a$};
    \node[endpoint] at (1.16,.21) {$d$};
    \node[endpoint] at (5.19,.21) {$c$};
    \node[endpoint] at (5.57,.21) {$b$};
  \end{scope}
  \draw[transition] (6.54,0)--(7.23,0);
  \node[font=\small,anchor=south] at (6.885,.16) {2-\opt};
\end{tikzpicture}
    \caption{Structural difference between 3-\opt and 2-\opt moves.
    \textbf{Top}: a 3-\opt move can remove the three black edges and reconnect the subpaths, effectively relocating the blue subpath and introducing the edge $ab$.
    \textbf{Bottom}: a 2-\opt move replaces the edges $ac$ and $db$ with $ad$ and $cb$, effectively reversing the blue subpath between $c$ and $d$. 
    The blue arrowheads indicate this reversal.}
    \label{fig:compare}
    \end{figure}

The recursive idea behind this construction is easiest to develop for the Maximum Constraint Assignment (MCA) problem with the flip neighborhood, a generalization of Max-Cut in which the local structure of the instances is controlled by three parameters. Moreover, MCA is of independent interest, since for several small parameter values the all-exp question has been open for a long time.

An instance of $(p,q,r)$-MCA consists of variables, each with a domain of size~$r$, and constraints, 
each of which depends on at most~$p$ variables and assigns a weight to every combination of values of these variables; every variable occurs in at most~$q$ constraints. The goal is to find an assignment of values to the variables that maximizes the sum of the resulting constraint weights. A local-search step in the flip neighborhood changes the value of a single variable; for binary variables, that is, $r=2$, this is simply a flip. Throughout this paper, we consider MCA only with the flip neighborhood.

In the terminology of constraint satisfaction, MCA is a valued constraint satisfaction problem (VCSP) with finite weights and a maximization objective. In particular, MCA naturally generalizes weighted satisfiability problems. Dumrauf and Monien~\cite{DM2013} studied the problem systematically and, building on techniques introduced by Krentel~\cite{Kre1989}, proved that $(3,2,3)$-MCA, $(2,3,6)$-MCA, and the binary problem $(6,2,2)$-MCA are PLS-complete via tight reductions and hence have the all-exp property. Moreover, Max-Cut on graphs of maximum degree four is a special case of $(2,4,2)$-MCA, so the result of Monien and Tscheuschner~\cite{MT2010} implies the all-exp property for $(2,4,2)$-MCA.
More recently, Heimann, Hoang, and Hougardy~\cite{HHH2026b} proved the all-exp property for $(2,2,5)$-MCA. 

On the positive side, Kaznatcheev, Cohen, and Jeavons~\cite{KCJ2020,Kaz2020} proved that for $(2,2,2)$-MCA every sequence of improving flips has length at most quadratic in the instance size. 
In contrast, the same authors established the is-exp property for $(2,2,3)$- and $(2,3,2)$-MCA~\cite{KCJ2020}.
More recently, Kaznatcheev and Vazquez Alferez~\cite{MR4967954} strengthened the is-exp result for $(2,3,2)$-MCA by showing that even the steepest-ascent pivot rule can require exponentially many flips, although their instances admit polynomial-length improving sequences under other pivot rules.
Whether these problems also have the all-exp property remained open. 
Taking into account the known all-exp results and the fact that the all-exp property carries over to larger parameter values, exactly eleven cases were open: $(2,2,3)$-, $(2,2,4)$-, $(2,3,2)$-, $(2,3,3)$-, $(2,3,4)$-, $(3,2,2)$-, \mbox{$(4,2,2)$-,} $(5,2,2)$-, $(3,3,2)$-, $(4,3,2)$-, and $(5,3,2)$-MCA.

In this paper, we prove the following two theorems.

\begin{restatable}{theorem}{thmthreetwotwo}
\label{thm:322}
    $(3,2,2)$-MCA with the flip neighborhood has the all-exp property.
\end{restatable}

\begin{restatable}{theorem}{thmtwothreetwo}
\label{thm:232}
    $(2,3,2)$-MCA with the flip neighborhood has the all-exp property.
\end{restatable}

Since every open case other than $(2,2,3)$- and $(2,2,4)$-MCA has parameters that are componentwise at least $(3,2,2)$ or at least $(2,3,2)$, \cref{thm:322,thm:232} leave only these two cases open. Both theorems are tight: decreasing any of the three parameters yields a problem in which every sequence of improving flips has polynomial length. For $(2,2,2)$-MCA, this is the result of Kaznatcheev, Cohen, and Jeavons mentioned above; if one of the parameters equals~$1$, every sequence of improving flips even has at most linear length, since then each variable can change its value at most $r-1$ times ($p=1$), each constraint can increase its weight at most $r^p-1$ times ($q=1$), or no move is possible at all ($r=1$).

Since 2-\opt is the smallest \kopt neighborhood, \cref{thm:2opt} is best possible as well. The three theorems thus show that local search can require exponentially many iterations under every pivot rule, even in remarkably restricted settings. These results settle long-standing open questions at two of the most basic boundaries in discrete local search: the last open \kopt neighborhood for the TSP and the dividing line between polynomial and exponential behavior for binary MCA. In particular, binary MCA is now completely classified: $(p,q,2)$-MCA has the all-exp property if and only if $p,q \geq 2$ and $\max\{p,q\} \geq 3$.

\paragraph{Overview of the constructions.}
Our three proofs are based on a common construction idea, which reveals a close connection between the local-search dynamics of the TSP with the 2-\opt neighborhood and of MCA. This idea is similar to the recursive approach used in previous all-exp constructions, such as that of Michel and Scott~\cite{MS2023}.
Each instance consists of a linear number of layers, with interactions between consecutive layers restricted to two interfaces. 
One interface controls when an improving sequence in the preceding layer can start, while the other ensures that this sequence must finish before further progress in the current layer is possible. 
The key idea behind the exponential behavior is that a traversal of one layer requires at least two traversals of the preceding layer. Applied recursively, this ensures that, from a suitably chosen initial solution, every maximal improving sequence, that is, every improving sequence ending in a local optimum, has exponential length. The main challenge is to enforce these dependencies so that no sequence of improving moves can bypass the recursive traversals.

Our construction for $(3,2,2)$-MCA implements this idea very directly. 
The structure and constraint weights are chosen such that, starting from a suitable initial assignment, there is exactly one improving flip at every step, and each flip increases the objective value by exactly one. Moreover, our construction is $3$-partite, meaning that we can partition the variables into three parts such that no constraint contains two variables from the same part. 
Despite these restrictions, the unique maximal improving sequence has exponential length.

The construction for 2-\opt follows a very similar recursive pattern but requires an additional idea. 
We can assign states to different parts of a tour, but not every assignment of states represents a valid tour, and not every change between valid state assignments corresponds to a 2-\opt move. 
We deliberately exploit these restrictions to exclude shortcuts that would otherwise be improving. 
Unlike in the previous constructions for $k\geq3$~\cite{HHH2024,HHH2026}, where unwanted improving moves could be excluded by the choice of edge weights alone, the restricted nature of 2-\opt makes these feasibility restrictions particularly important. 
They also allow us to choose the weights of two parts of each layer, which we call the upper and lower gadgets, largely independently. 
Further, unlike the previous constructions~\cite{HHH2024, HHH2026}, which rely on a reduction from Max-Cut with the flip neighborhood, we present a ``direct'' construction here.
While we cannot enforce a unique improving move at every step, we show that all maximal improving sequences differ only in the order of a few moves that can be performed independently.

For $(2,3,2)$-MCA, we build on the recursive construction of Michel and Scott~\cite{MS2023} for Max-Cut with the flip neighborhood. 
However, their construction contains vertices of degree four, whereas every variable in $(2,3,2)$-MCA may occur in at most three constraints. 
Reducing the degree while preserving the behavior of all improving sequences is a challenge: introducing auxiliary variables creates additional improving flips, which may allow the recursive dependencies to be bypassed. 
We overcome this difficulty by designing specific gadgets that reproduce the required dependencies using only variables that occur in at most three constraints. 
In particular, one gadget implements the dependence on three variables through Boolean operations, while another propagates the change of one variable to the next layer and ensures that the next recursive traversal can begin only after this propagation has been adequately completed.
Although several improving flips may now be available simultaneously, we show that their dependencies are captured by a directed acyclic graph. 
Consequently, all maximal improving sequences correspond to linear extensions of this graph and have the same exponential length.

\paragraph{Outline.}
\cref{sec:prelims} establishes the necessary background, including key terms and notations. 
We then prove the results for $(3,2,2)$-MCA, the TSP with the 2-\opt neighborhood, and $(2,3,2)$-MCA in Sections~\ref{sec:322}, \ref{sec:2opt}, and \ref{sec:232}, respectively.
The three proofs are self-contained and can be read independently. 
We occasionally reuse notation, with its meaning defined separately in each section.

\section{Preliminaries}
\label{sec:prelims}

Denote by $[n]$ the set $\{1, \dots, n\}$.
Let $\mathbb{N}$ denote the set of all natural numbers (including zero).
For $a\in\{0,1\}$, let $\bar a=1-a$, and write $\oplus$ for addition modulo~$2$ (i.e., the XOR operation).
For a statement $A$, we write $[A]$ for its indicator (i.e., it is 1 if $A$ is true and 0 otherwise).
For an assignment $\sigma\colon X\to\{0,1\}$, a variable $x\in X$, and $a\in\{0,1\}$, we write $\sigma[x\mapsto a]$ for the assignment that agrees with $\sigma$ on $X\setminus\{x\}$ and assigns the value $a$ to $x$.

For a directed acyclic graph $D$, a \emph{linear extension} of $D$ is an ordering $(x_1,\dots,x_{|V(D)|})$ of its vertices such that every edge $x_i x_j\in E(D)$ satisfies $i<j$.
Equivalently, a linear extension is a topological ordering of $D$.

\subsection{Local-search problems}
A \defi{local-search problem} $P$ is an optimization problem that consists of a set of instances $D_{P}$, a finite set of (feasible) solutions $F_{P}(I)$ for each instance $I\in D_{P}$, an objective function $f_{P}$ that assigns an integer value to each instance $I\in D_{P}$ and solution $s\in F_P(I)$, and a neighborhood $N_{P}(s,I)\subseteq F_{P}(I)$ for each solution $s\in F_{P}(I)$. 
The size of every solution $s \in F_{P}(I)$ must be bounded by a polynomial in the size of $I$. 
The goal is to find a \defi{locally optimal solution} for a given instance $I$; that is, a solution $s \in F_{P}(I)$, such that no solution $s' \in N_{P}(s,I)$ yields a better objective value than $f_P(s,I)$.
Formally, this means, for all $s'\in N_{P}(s,I)$, $f_{P}(s,I)\leq f_{P}(s',I)$ if $P$ is a minimization problem, and $f_{P}(s,I)\geq f_{P}(s',I)$ if $P$ is a maximization problem.

The \defi{standard local-search algorithm} for an instance $I$ proceeds as follows.
It starts with some initial solution $s \in F_{P}(I)$.
Then it iteratively visits a neighbor with a strictly better objective value until it reaches a local optimum. 
If a solution has more than one better neighbor, the algorithm has to choose one by some prespecified rule, often referred to as a \defi{pivot rule}.

For a move from a feasible solution $s$ to $s'$, we define its \defi{gain} as $f_P(s,I)-f_P(s',I)$ for a minimization problem, and as $f_P(s',I)-f_P(s,I)$ for a maximization problem (with the instance fixed).
A move is \defi{improving} if and only if its gain is positive. 
A sequence of solutions $s_0, s_1, \dots, s_\ell$ is an \defi{improving sequence} of length~$\ell$ if, for each $i \in \{1,\ldots, \ell\}$, we have $s_i \in N_P(s_{i-1}, I)$ and the move from $s_{i-1}$ to $s_i$ is improving. It is \defi{maximal} if $s_\ell$ is a local optimum.

\begin{definition}[All-exp property]
    A local-search problem~$P$ has the \defi{all-exp} property if there are infinitely many pairs of an instance~$I\in D_P$ and an initial solution $s \in F_P(I)$, for which the standard local-search algorithm needs an exponential number of iterations 
    for all possible pivot rules. (By exponential, we mean a growth of at least $\Omega(2^{n^{\epsilon}})$ for some constant $\epsilon>0$ and instances of size $n$.)
\end{definition}

An alternative way to define the all-exp property is to consider the \defi{transition graph} of an instance $I$
of a local-search problem $P$. This graph has as vertices all the solutions $F_{P}(I)$ and there exists a directed edge $(s, s')$ if the solution $s'$ is in the neighborhood of $s$ and has a strictly better 
objective value than $s$. By definition, the transition graph is acyclic, and a directed path of length $\ell$ in the transition graph corresponds to $\ell$ iterations of the standard local-search algorithm. A local optimum corresponds to a sink in the transition graph. 
Hence, a local-search problem has the all-exp property if and only if there exist
infinitely many pairs of instances and initial solutions such that all paths in the transition graph
from the initial solution to a sink have exponential length. 

\subsection{\texorpdfstring{TSP with the \kopt neighborhood}{TSP with the k-opt neighborhood}}
A \defi{Hamiltonian cycle}
of a graph is a cycle 
that contains all vertices of the graph.
In this paper, we also refer to a Hamiltonian cycle as a \defi{tour}.

A TSP instance is a tuple $(G, c)$, where $G$ is a complete undirected graph $(V, E)$, and $c:E\to \mathbb{N}$ is a function that assigns a  weight to each edge of $G$.
The goal is to find a tour of $G$ that minimizes the sum of the edge weights in the tour.
We call this sum the \defi{weight of the tour}.

A \defi{swap} is a tuple $(E_1, E_2)$ of subsets $E_1, E_2 \subseteq E$ with $|E_1| = |E_2|$.
For an integer $k$, a swap $(E_1,E_2)$ with $|E_1| \le k$ is called a \defi{$k$-swap}.
For a subgraph $G'$ of $G$, and two edge sets $E_1, E_2$ with $E_1\subseteq E(G')$, performing a swap $(E_1, E_2)$ on $G'$ refers to the act of removing $E_1$ from $G'$ and adding $E_2$ to $G'$. 
We also call it swapping $E_1$ for $E_2$ in $G'$.
Given a tour~$\tau$, a swap $(E_1, E_2)$ has gain $\sum_{e\in E_1} c(e) - \sum_{e \in E_2} c(e)$, and it is improving for $\tau$ if after swapping $E_1$ for $E_2$ in $\tau$, we obtain a tour of smaller weight.

A \defi{($k$-)swap sequence} is a sequence $L = (S_1, \dots, S_{\ell})$, such that each $S_i$ is a ($k$-)swap. A sequence $L$ is \defi{improving} for a tour~$\tau$ if, for each $i \in \{1,\ldots,\ell\}$, the swap $S_i$ is improving for the tour obtained from $\tau$ by performing $S_1, \dots, S_{i-1}$.

In the local-search problem TSP with the \kopt neighborhood, the neighbors of a tour~$\tau$ are those that can be obtained from~$\tau$ by a $k$-swap.
The \kopt algorithm is then the standard local-search algorithm for this problem, and an execution of the algorithm corresponds to an improving $k$-swap sequence.

\subsection{Maximum Constraint Assignment}
An instance of $(p,q,r)$-\emph{Maximum Constraint Assignment}, abbreviated as $(p,q,r)$-MCA, consists of a set of variables $X=\{x_1,\dots,x_n\}$ and a set of constraints $\mathcal C=\{C_1,\dots,C_m\}$.
Each variable takes a value in $\{0,\dots,r-1\}$.
A constraint $C_i(x_{i_1},\dots,x_{i_{p_i}})$, for $p_i\le p$, assigns an integer weight to every possible assignment of the variables appearing in it.
Moreover, every variable appears in at most $q$ constraints.

For an assignment $\sigma\colon X\to\{0, \dots, r-1\}$, its value is
\[
\Phi(\sigma)
=
\sum_{C_i\in\mathcal C}
C_i\bigl(\sigma(x_{i_1}),\dots,\sigma(x_{i_{p_i}})\bigr).
\]
The objective is to maximize $\Phi$.

In this paper, we only consider the binary case $r=2$.
For an assignment $\sigma$ and a variable $x\in X$, the \emph{flip of $x$} is the assignment obtained from $\sigma$ by replacing $\sigma(x)$ with $1-\sigma(x)$ and leaving all other variables unchanged.
The \defi{flip neighborhood} of an assignment~$\sigma$ consists of all assignments obtained from $\sigma$ by flipping exactly one variable.
The flip of $x$ at $\sigma$ has gain $\Phi\bigl(\sigma[x\mapsto1-\sigma(x)]\bigr)-\Phi(\sigma)$; it is improving if it increases the value of the assignment.
For brevity, we also say that \defi{$x$ is improving} if the flip of $x$ is improving.
An assignment with no improving variable is a \defi{local optimum}.
An \defi{improving flip sequence} is a sequence $(y_1,\dots,y_\ell)$ such that, starting from a given assignment, the flip of $y_j$ is improving after the preceding flips $y_1,\dots,y_{j-1}$ have been performed.

\section{(3,2,2)-MCA has the all-exp property}
\label{sec:322}
In this section, we prove the following.

\thmthreetwotwo*

In this setting, every constraint has at most three variables, each variable occurs in at most two constraints, and all variables are binary.
We show the theorem above via a direct construction. 
For each $n\in\mathbb{N}$, we give an instance $I_n$ with $O(n)$ variables and constraints, together with an initial assignment. 
We then show that there is a unique maximal improving sequence of flips of length $\Omega(2^n)$ starting from the initial assignment. Since all constraint weights are at most $2^{n+4}$, the instance $I_n$ can be encoded with $O(n^2)$ bits, so this length is exponential in the instance size in the sense of our definition.

\paragraph{The construction.}
The instance $I_n$ consists of variables and constraints organized in $n+1$ layers. 
The base layer ($i=0$) consists of two variables $X_0 = \{v_{0,1}, v_{0,6}\}$ and one constraint $C_{0}(v_{0,1},v_{0,6})$, while each layer $i\in[n]$ consists of six variables $X_i = \{v_{i,1}, v_{i,2},v_{i,3}, v_{i,4}, v_{i,5}, v_{i,6}\}$ and four constraints $C_{i,1}(v_{i-1,1},v_{i,1},v_{i,3})$, $C_{i,2}(v_{i-1,6},v_{i,2},v_{i,4})$, $C_{i,3}(v_{i,2},v_{i,3},v_{i,5})$, and $C_{i,4}(v_{i,4},v_{i,5},v_{i,6})$. All variables can assume the values $0$ and $1$, which we also identify with $\greenb$ (green) and $\redb$ (red), respectively. 
See \cref{fig:simplifiedmcaconstruction}.
The weights of $C_0$ and other constraints are given by \cref{tab:const_C_0} and \cref{tab:constraints_C_i}, respectively.

\begin{remark}
The instance $I_n$ is \defi{$3$-partite}, that is, its variables can be partitioned into three parts such that no constraint contains two variables of the same part. Such a partition is given by $\pi\colon\bigcup_{i=0}^n X_i \to \{1,2,3\}$ with $\pi(v_{0,1})=1$, $\pi(v_{0,6})=3$, and, for all $i\in[n]$, $\pi(v_{i,3})=\pi(v_{i,4})=2$,
\begin{equation*}
    \pi(v_{i,1})= \pi(v_{i,5})= \begin{cases}
        3 & \text{if $i$ is odd}\\ 
        1 & \text{if $i$ is even}
    \end{cases} \hspace{2em}
    \pi(v_{i,2})= \pi(v_{i,6})= \begin{cases}
        1 & \text{if $i$ is odd}\\ 
        3 & \text{if $i$ is even}
    \end{cases}
\end{equation*}
\end{remark}


\pgfdeclarelayer{hyperedgefill}
\pgfdeclarelayer{hyperedgecontour}
\pgfdeclarelayer{vertices}
\pgfsetlayers{hyperedgefill,main,hyperedgecontour,vertices}

\newcommand{\hyperedgewidth}{10pt}      
\newcommand{\hyperedgeopacity}{0.15}
\newcommand{\hyperedgelinewidth}{0.3pt}

\newcommand{\AppendToTubePath}[2]{%
  \xdef\TubePath{\TubePath (#1pt,#2pt) -- }%
}

\newcommand{\NormalizeVector}[2]{%
  \pgfmathsetmacro{\scale}{max(abs(#1),abs(#2),0.0001)}%
  \pgfmathsetmacro{\len}{\scale*sqrt((#1/\scale)*(#1/\scale)+(#2/\scale)*(#2/\scale))}%
}

\newcommand{\BuildTubePath}[3]{%
  \path let \p1=(#1), \p2=(#2), \p3=(#3) in
    \pgfextra{%
      \xdef\Ax{\the\dimexpr\x1\relax}%
      \xdef\Ay{\the\dimexpr\y1\relax}%
      \xdef\Bx{\the\dimexpr\x2\relax}%
      \xdef\By{\the\dimexpr\y2\relax}%
      \xdef\Cx{\the\dimexpr\x3\relax}%
      \xdef\Cy{\the\dimexpr\y3\relax}%
    };

  \pgfmathsetmacro{\ax}{\Ax/1pt}
  \pgfmathsetmacro{\ay}{\Ay/1pt}
  \pgfmathsetmacro{\bx}{\Bx/1pt}
  \pgfmathsetmacro{\by}{\By/1pt}
  \pgfmathsetmacro{\cx}{\Cx/1pt}
  \pgfmathsetmacro{\cy}{\Cy/1pt}

  \pgfmathsetmacro{\qx}{2*\bx - 0.5*\ax - 0.5*\cx}
  \pgfmathsetmacro{\qy}{2*\by - 0.5*\ay - 0.5*\cy}

  \gdef\TubePath{}%

  \foreach \i in {0,...,32}{%
    \pgfmathsetmacro{\t}{\i/32}
    \pgfmathsetmacro{\px}{(1-\t)^2*\ax + 2*(1-\t)*\t*\qx + \t^2*\cx}
    \pgfmathsetmacro{\py}{(1-\t)^2*\ay + 2*(1-\t)*\t*\qy + \t^2*\cy}
    \pgfmathsetmacro{\dx}{2*(1-\t)*(\qx-\ax) + 2*\t*(\cx-\qx)}
    \pgfmathsetmacro{\dy}{2*(1-\t)*(\qy-\ay) + 2*\t*(\cy-\qy)}
    \NormalizeVector{\dx}{\dy}
    \pgfmathsetmacro{\nx}{-\dy/\len}
    \pgfmathsetmacro{\ny}{ \dx/\len}
    \pgfmathsetmacro{\xx}{\px+\hyperedgewidth*\nx}
    \pgfmathsetmacro{\yy}{\py+\hyperedgewidth*\ny}
    \AppendToTubePath{\xx}{\yy}
  }

  \pgfmathsetmacro{\dxend}{2*(\cx-\qx)}
  \pgfmathsetmacro{\dyend}{2*(\cy-\qy)}
  \pgfmathsetmacro{\angend}{atan2(\dyend,\dxend)}
  \foreach \i in {1,...,12}{%
    \pgfmathsetmacro{\ang}{\angend+90-180*\i/12}
    \pgfmathsetmacro{\xx}{\cx+\hyperedgewidth*cos(\ang)}
    \pgfmathsetmacro{\yy}{\cy+\hyperedgewidth*sin(\ang)}
    \AppendToTubePath{\xx}{\yy}
  }

  \foreach \i in {32,...,0}{%
    \pgfmathsetmacro{\t}{\i/32}
    \pgfmathsetmacro{\px}{(1-\t)^2*\ax + 2*(1-\t)*\t*\qx + \t^2*\cx}
    \pgfmathsetmacro{\py}{(1-\t)^2*\ay + 2*(1-\t)*\t*\qy + \t^2*\cy}
    \pgfmathsetmacro{\dx}{2*(1-\t)*(\qx-\ax) + 2*\t*(\cx-\qx)}
    \pgfmathsetmacro{\dy}{2*(1-\t)*(\qy-\ay) + 2*\t*(\cy-\qy)}
    \NormalizeVector{\dx}{\dy}
    \pgfmathsetmacro{\nx}{-\dy/\len}
    \pgfmathsetmacro{\ny}{ \dx/\len}
    \pgfmathsetmacro{\xx}{\px-\hyperedgewidth*\nx}
    \pgfmathsetmacro{\yy}{\py-\hyperedgewidth*\ny}
    \AppendToTubePath{\xx}{\yy}
  }

  \pgfmathsetmacro{\dxstart}{2*(\qx-\ax)}
  \pgfmathsetmacro{\dystart}{2*(\qy-\ay)}
  \pgfmathsetmacro{\angstart}{atan2(\dystart,\dxstart)}
  \foreach \i in {1,...,12}{%
    \pgfmathsetmacro{\ang}{\angstart-90-180*\i/12}
    \pgfmathsetmacro{\xx}{\ax+\hyperedgewidth*cos(\ang)}
    \pgfmathsetmacro{\yy}{\ay+\hyperedgewidth*sin(\ang)}
    \AppendToTubePath{\xx}{\yy}
  }

  \xdef\TubePath{\TubePath cycle}%
}

\newcommand{\Hyperedge}[4]{%
  \BuildTubePath{#1}{#2}{#3}%
  \begin{pgfonlayer}{hyperedgefill}
    \fill[#4,opacity=\hyperedgeopacity] \TubePath;
  \end{pgfonlayer}
  \begin{pgfonlayer}{hyperedgecontour}
    \draw[black,line width=\hyperedgelinewidth] \TubePath;
  \end{pgfonlayer}
}

\newcommand{\RoundHyperedge}[4]{%
  \coordinate (Hc) at ($1/3*(#1)+1/3*(#2)+1/3*(#3)$);

  \coordinate (H1) at ($(Hc)!1.2!(#1)$);
  \coordinate (H2) at ($(Hc)!1.2!(#2)$);
  \coordinate (H3) at ($(Hc)!1.2!(#3)$);

  \filldraw[
    fill=#4,
    fill opacity=\hyperedgeopacity,
    draw=black,
    line width=\hyperedgelinewidth
  ]
  plot[smooth cycle,tension=0.75] coordinates {
    (H1) (H2) (H3)
  };
}

\newcommand{\BuildTubePathTwo}[2]{%
  \path let \p1=(#1), \p2=(#2) in
    \pgfextra{%
      \xdef\Ax{\the\dimexpr\x1\relax}%
      \xdef\Ay{\the\dimexpr\y1\relax}%
      \xdef\Bx{\the\dimexpr\x2\relax}%
      \xdef\By{\the\dimexpr\y2\relax}%
    };

  \pgfmathsetmacro{\ax}{\Ax/1pt}
  \pgfmathsetmacro{\ay}{\Ay/1pt}
  \pgfmathsetmacro{\bx}{\Bx/1pt}
  \pgfmathsetmacro{\by}{\By/1pt}

  \pgfmathsetmacro{\dx}{\bx-\ax}
  \pgfmathsetmacro{\dy}{\by-\ay}
  \NormalizeVector{\dx}{\dy}

  \pgfmathsetmacro{\nx}{-\dy/\len}
  \pgfmathsetmacro{\ny}{ \dx/\len}

  \pgfmathsetmacro{\ang}{atan2(\dy,\dx)}

  \gdef\TubePath{}%

  \pgfmathsetmacro{\xx}{\ax+\hyperedgewidth*\nx}
  \pgfmathsetmacro{\yy}{\ay+\hyperedgewidth*\ny}
  \AppendToTubePath{\xx}{\yy}

  \pgfmathsetmacro{\xx}{\bx+\hyperedgewidth*\nx}
  \pgfmathsetmacro{\yy}{\by+\hyperedgewidth*\ny}
  \AppendToTubePath{\xx}{\yy}

  \foreach \i in {1,...,16}{%
    \pgfmathsetmacro{\a}{\ang+90-180*\i/16}
    \pgfmathsetmacro{\xx}{\bx+\hyperedgewidth*cos(\a)}
    \pgfmathsetmacro{\yy}{\by+\hyperedgewidth*sin(\a)}
    \AppendToTubePath{\xx}{\yy}
  }

  \pgfmathsetmacro{\xx}{\ax-\hyperedgewidth*\nx}
  \pgfmathsetmacro{\yy}{\ay-\hyperedgewidth*\ny}
  \AppendToTubePath{\xx}{\yy}

  \foreach \i in {1,...,16}{%
    \pgfmathsetmacro{\a}{\ang-90-180*\i/16}
    \pgfmathsetmacro{\xx}{\ax+\hyperedgewidth*cos(\a)}
    \pgfmathsetmacro{\yy}{\ay+\hyperedgewidth*sin(\a)}
    \AppendToTubePath{\xx}{\yy}
  }

  \xdef\TubePath{\TubePath cycle}%
}

\newcommand{\HyperedgeTwo}[3]{%
  \BuildTubePathTwo{#1}{#2}%
  \begin{pgfonlayer}{hyperedgefill}
    \fill[#3,opacity=\hyperedgeopacity] \TubePath;
  \end{pgfonlayer}
  \begin{pgfonlayer}{hyperedgecontour}
    \draw[black,line width=\hyperedgelinewidth] \TubePath;
  \end{pgfonlayer}
}

\def\defvertex#1#2#3#4{\coordinate (c#1) at (#2,#3);
                       \begin{pgfonlayer}{vertices}\fill[#4] (#2,#3) circle (1.4mm);
                       \end{pgfonlayer}
                       \node[] (#1) at (c#1) {};}

\def\defvertexw#1#2#3#4{\coordinate (c#1) at (#2,#3);
                       \begin{pgfonlayer}{vertices}\draw[#4,very thick] (#2,#3) circle (1.4mm);
                       \end{pgfonlayer}
                       \node[] (#1) at (c#1) {};}

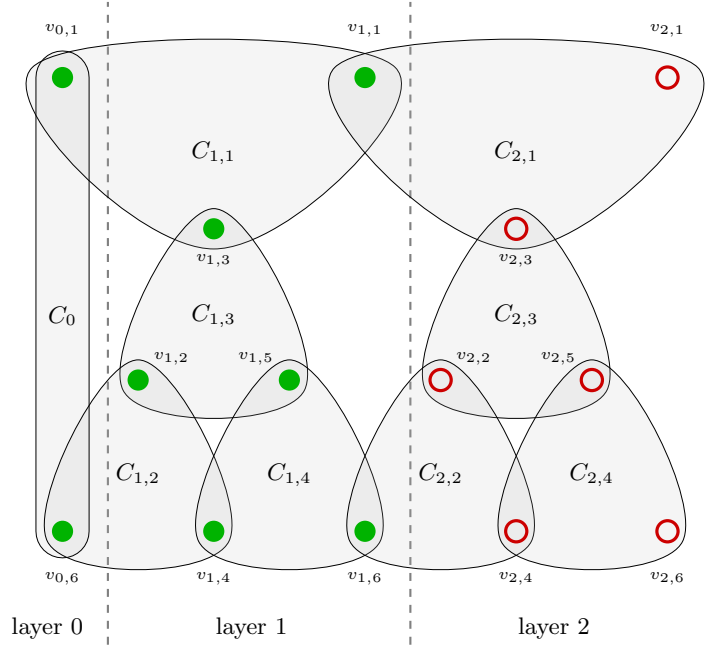
\begin{figure}[t]
\centering

\begin{tikzpicture}[scale=1]

\defvertex{A0}{ 0}{6}{\mygreen}
\defvertex{F0}{ 0}{0}{\mygreen}

\defvertex{A1}{4}{6}{\mygreen}
\defvertex{B1}{1}{2}{\mygreen}
\defvertex{C1}{2}{4}{\mygreen}
\defvertex{D1}{2}{0}{\mygreen}
\defvertex{E1}{3}{2}{\mygreen}
\defvertex{F1}{4}{0}{\mygreen}

\draw[thick, gray, dashed] (0.6, -1.5) -- (0.6, 7); 
\draw[thick, gray, dashed] (4.6, -1.5) -- (4.6, 7);

\HyperedgeTwo{A0}{F0}{lightgray}

\RoundHyperedge{A0}{A1}{C1}{lightgray}
\RoundHyperedge{C1}{B1}{E1}{lightgray}
\RoundHyperedge{F0}{D1}{B1}{lightgray}
\RoundHyperedge{D1}{F1}{E1}{lightgray}

\begin{scope}[xshift=4cm]
\defvertexw{A2}{4}{6}{\myred}
\defvertexw{B2}{1}{2}{\myred}
\defvertexw{C2}{2}{4}{\myred}
\defvertexw{D2}{2}{0}{\myred}
\defvertexw{E2}{3}{2}{\myred}
\defvertexw{F2}{4}{0}{\myred}

\RoundHyperedge{A1}{A2}{C2}{lightgray}
\RoundHyperedge{C2}{B2}{E2}{lightgray}
\RoundHyperedge{F1}{D2}{B2}{lightgray}
\RoundHyperedge{D2}{F2}{E2}{lightgray}
\end{scope}

\node at (2, 5) {\footnotesize$C_{1,1}$};
\node at (1, 0.75) {\footnotesize$C_{1,2}$};
\node at (2, 2.85) {\footnotesize$C_{1,3}$};
\node at (0, 2.85) {\footnotesize$C_0$};
\node at (3, 0.75) {\footnotesize$C_{1,4}$};
\node at (6, 5) {\footnotesize$C_{2,1}$};
\node at (5, 0.75) {\footnotesize$C_{2,2}$};
\node at (6, 2.85) {\footnotesize$C_{2,3}$};
\node at (7, 0.75) {\footnotesize$C_{2,4}$};

\node[above=4mm] at (A0) {\tiny $v_{0,1}$};
\node[above=4mm] at (A1) {\tiny $v_{1,1}$};
\node[above=4mm] at (A2) {\tiny $v_{2,1}$};

\node[above right = 1mm] at (B1) {\tiny $v_{1,2}$};
\node[above right = 1mm] at (B2) {\tiny $v_{2,2}$};

\node[above left = 1mm] at (E1) {\tiny $v_{1,5}$};
\node[above left = 1mm] at (E2) {\tiny $v_{2,5}$};

\node[below=4mm] at (F0) {\tiny $v_{0,6}$};
\node[below=4mm] at (F1) {\tiny $v_{1,6}$};
\node[below=4mm] at (F2) {\tiny $v_{2,6}$};

\node[below=4mm] at (D1) {\tiny $v_{1,4}$};
\node[below=4mm] at (D2) {\tiny $v_{2,4}$};

\node[below=2mm] at (C1) {\tiny $v_{1,3}$};
\node[below=2mm] at (C2) {\tiny $v_{2,3}$};

\node[font=\footnotesize,anchor=north] at (-0.2,-1) {layer 0};
\node[font=\footnotesize,anchor=north] at (2.5,-1) {layer 1};
\node[font=\footnotesize,anchor=north] at (6.5,-1) {layer 2};

\end{tikzpicture}

\caption{The $(3,2,2)$-MCA instance $I_2$, showing the values of the variables in the local optimum reached from the initial all-zero assignment: the variables of layer~$2$ have value~$1$ (red), and all other variables have value~$0$ (green).}
\label{fig:simplifiedmcaconstruction}
\end{figure}

\begin{table}[t!]
    \centering
    \renewcommand{\arraystretch}{1.5}
    \begin{tabular}{l|c|c|c|c}
        & $00/\greenb\greenb$ & $01/\greenb\redb$ & $10/\redb\greenb$ & $11/\redb\redb$ \\\hline
        $C_{0}(v_{0,1},v_{0,6})$
        & $14$ & $1$ & $15$ & $16$ \\
    \end{tabular}
    \caption{The weights of constraint $C_0$.}
    \label{tab:const_C_0}
\end{table}

\begin{table}[t!]
    \centering
    \begin{adjustbox}{width=1\textwidth}
    \renewcommand{\arraystretch}{1.5}
    \begin{tabular}{l|c|c|c|c|c|c|c|c}
          & $000/\greenb\greenb\greenb$ & $001/\greenb\greenb\redb$ & $010/\greenb\redb\greenb$ & $011/\greenb\redb\redb$ & $100/\redb\greenb\greenb$ & $101/\redb\greenb\redb$ & $110/\redb\redb\greenb$ & $111/\redb\redb\redb$ \\\hline
         $C_{i,1}(v_{i-1,1},v_{i,1},v_{i,3})$
         &  $2^{i+3}$ & $1$ & $2^{i+3} + 1$ & $2^{i+4}$ & $0$ & $1$ & $2^{i+3} + 1$ & $2^{i+3}$\\\hline
         $C_{i,2}(v_{i-1,6},v_{i,2},v_{i,4})$
         & $2$ & $0$ & $0$ & $14$ & $2$ & $12$ & $12$ & $14$\\\hline
         $C_{i,3}(v_{i,2},v_{i,3},v_{i,5})$
         & $2^{i+3} + 9$ & $2^{i+3} + 8$ & $0$ & $2^{i+3} + 6$ & $2^{i+3}$ & $0$  & $2^{i+3} + 2$ & $2^{i+3} + 3$\\\hline
         $C_{i,4}(v_{i,4},v_{i,5},v_{i,6})$
         & $13$ & $0$ & $0$ & $0$ & $0$ & $0$ & $0$ & $1$ \\
    \end{tabular}
    \end{adjustbox}
    \caption{The weights of the four constraints $C_{i,1},C_{i,2},C_{i,3}$ and $C_{i,4}$ for all $i\in [n]$.}
    \label{tab:constraints_C_i}
\end{table}

\paragraph{Initial assignment.}
Every variable is initialized as zero ($\greenb$).

\paragraph{Unique maximal improving flip sequence.}
We now characterize the unique maximal improving flip sequence for the construction above.

\begin{lemma}\label{lem:322-mca}
Starting from the all-zero assignment, the unique maximal improving flip sequence $S_n$ for $I_n$ is given recursively by
\begin{align}
S_0  & = (v_{0,1},v_{0,6}),  \label{eq:S0}\\
S_i  &=  (v_{i,1}, S_{i-1}, v_{i,2},v_{i,3}, S_{i-1}, v_{i,4}, v_{i,5}, v_{i,6}) \qquad (i\geq1).    \label{eq:Si}
\end{align}
Every flip in this sequence increases the objective value by exactly $1$.  
The resulting number of flips is $8\cdot2^n-6$.
\end{lemma}

\begin{proof}
Recall that we identify green ($\greenb$) with $0$ and red ($\redb$) with $1$.
For $i\geq0$, let $F_i$ denote the objective consisting of the sum of $C_0$
and all constraints $C_{j,k}$ with $1\le j\le i$, $1\leq k\leq4$.
Thus, $F_n$ is the objective of the complete instance.

Let $\mathbf G_i$ denote the assignment in which all variables in
$X_0\cup\cdots\cup X_i$ are green, and let $\mathbf R_i$ denote the
assignment in which all variables of $X_i$ are red while all
variables in $X_0\cup\cdots\cup X_{i-1}$ are green.
For $i\ge0$, define the auxiliary objective
\begin{equation*}
    \widehat F_i := F_i+2^{i+4}(1-v_{i,1})+12v_{i,6}.
\end{equation*}
This auxiliary objective is used only in the proof to capture the contribution of layer $i+1$ to the preceding layers during the second occurrence of $S_i$ in $S_{i+1}$.

We prove simultaneously by induction on $i$ that
\begin{enumerate}
    \item[(A$_i$)]
    with respect to $F_i$, starting from $\mathbf G_i$, the unique maximal
    improving sequence is $S_i$; every flip has gain $1$, and the
    sequence ends in the local optimum $\mathbf R_i$;

    \item[(B$_i$)]
    with respect to $\widehat F_i$, starting from $\mathbf R_i$, the unique
    maximal improving sequence is again $S_i$; every flip has gain
    $1$, and the sequence ends in the local optimum $\mathbf G_i$.
\end{enumerate}

For $i=0$, we have $F_0=C_0$ and 
\[
1 = C_0(01) < C_0(00)=14<C_0(10)=15<C_0(11)=16.
\]
Hence, the unique maximal improving sequence from $00$ is $S_0=(v_{0,1},v_{0,6})$, it ends in $11 = \mathbf{R}_0$, and both flips have gain $1$.
Moreover, $\widehat F_0 = F_0+16(1-v_{0,1})+12v_{0,6}$. Starting from $11$, the objective values along the same sequence of flips are $28, 29, 30$ on the assignments $11, 01, 00 = \mathbf{G}_0$, respectively, while the objective value for the assignment $10$ is 15.
Therefore, (A$_0$) and (B$_0$) hold.

Now let $i\geq1$ and assume (A$_{i-1}$) and (B$_{i-1}$).
Let $P_i:=2^{i+3}$, $x:=v_{i-1,1}$, and $y:=v_{i-1,6}$.
We write the assignment of the variables in $X_i$ as a six-bit word
in the order $v_{i,1},\ldots,v_{i,6}$. For a fixed assignment of $X_i$, we call the objective ($F_i$ in (A$_i$) and $\widehat F_i$ in (B$_i$)), viewed as a function of the variables in $X_0\cup\cdots\cup X_{i-1}$, the \emph{objective on the preceding layers}.

We prove (A$_i$) and (B$_i$) simultaneously. 
Initially, in (A$_i$), the variables in $X_i$ have assignment $000000$, whereas in (B$_i$) they
have assignment $111111$. 
The variables in $X_0\cup\cdots\cup X_{i-1}$ are initially in $\mathbf G_{i-1}$ in both cases.

We claim that the objective on the preceding layers is the same in both cases, up to an additive constant.
Suppose that $X_i=000000$, as in (A$_i$). The constraints $C_{i,3}$ and $C_{i,4}$
contain only variables of $X_i$ and therefore contribute only an
additive constant to the objective on the preceding layers. Moreover,
$C_{i,2}(y,0,0)=2$ for both values of $y$, while
$C_{i,1}(x,0,0)$ equals $P_i$ for $x=0$ and $0$ for $x=1$.
Thus, up to an additive constant, the objective on the preceding
layers is
\[
    F_{i-1}+P_i(1-x).
\]
For the corresponding assignment $X_i=111111$ in (B$_i$),
$C_{i,1}(x,1,1)$ has the same dependence on $x$, up to an additive
constant, and $C_{i,2}(y,1,1)=14$ for both values of $y$.
The additional terms in $\widehat F_i$ depend only on variables in
$X_i$ and hence also contribute only an additive constant. Therefore,
the objective on the preceding layers is again
$F_{i-1}+P_i(1-x)$, up to an additive constant.

By (A$_{i-1}$), at $\mathbf G_{i-1}$ the unique improving variable
with respect to $F_{i-1}$ is $x=v_{i-1,1}$, and flipping it increases
the objective value by $1$. The additional term $P_i(1-x)$ decreases
the change caused by flipping $x$ by $P_i$, while it does not affect
the change caused by flipping any other variable. Hence, no variable
in $X_0\cup\cdots\cup X_{i-1}$ is improving in both (A$_i$) and (B$_i$).

At this point $(x,y)=00$. 
Comparing the values of $F_i$ in (A$_i$) and of $\widehat F_i$ in (B$_i$) 
before and after each possible flip in $X_i$ shows that $v_{i,1}$ is the unique
improving variable in $X_i$ in both (A$_i$) and (B$_i$), and that
flipping it increases the objective value by exactly $1$. Thus,
$v_{i,1}$ is the unique improving variable in both cases.

The remaining steps use the same type of calculation.
\cref{tab:direct-induction} summarizes the information about layer
$i$ that we will use below. In each row, the first and second columns give
the corresponding assignments of $X_i$ in (A$_i$) and (B$_i$),
respectively. The third column gives, for both assignments, the
objective on $X_0\cup\cdots\cup X_{i-1}$ up to an additive constant.
The last column gives all improving variables in $X_i$ for the values
of $(x,y)$ that occur at the corresponding point of the sequence;
``none'' means that no variable in $X_i$ is improving. Every flip
listed in the last column increases the objective value by exactly
$1$.

\begin{table}[htbp]
  \centering
  \caption{Information about layer $i$ used in the induction. The third column gives the objective on $X_0\cup\cdots\cup X_{i-1}$, up to an additive constant, for both assignments in the corresponding row.
  The last column applies to both assignments and gives all improving variables in $X_i$ for the indicated values of $(x,y)$; each listed flip increases the objective value by exactly $1$.}
  \label{tab:direct-induction}
  \renewcommand{\arraystretch}{1.25}
  \begin{tabular}{c | c | c | c}
    $X_i$ under $F_i$
    & $X_i$ under $\widehat F_i$
    & objective on preceding layers
    & improving variables in $X_i$\\
    \hline
    $000000$
    & $111111$
    & $F_{i-1}+P_i(1-x)$
    & $v_{i,1}$ for $(x,y)=00$
    \\
    \hline
    $100000$
    & $011111$
    & $F_{i-1}$
    &
    \begin{tabular}[c]{@{}c@{}}
      none for $(x,y)=00,10$,\\
      $v_{i,2}$ for $(x,y)=11$
    \end{tabular}
    \\
    \hline
    $110000$
    & $001111$
    & $F_{i-1}+12y$
    & $v_{i,3}$ for $(x,y)=11$
    \\
    \hline
    $111000$
    & $000111$
    & $\widehat F_{i-1}$
    &
    \begin{tabular}[c]{@{}c@{}}
      none for $(x,y)=11,01$,\\
      $v_{i,4}$ for $(x,y)=00$
    \end{tabular}
    \\
    \hline
    $111100$
    & $000011$
    & $F_{i-1}+P_i(1-x)$
    & $v_{i,5}$ for $(x,y)=00$
    \\
    \hline
    $111110$
    & $000001$
    & $F_{i-1}+P_i(1-x)$
    & $v_{i,6}$ for $(x,y)=00$
    \\
    \hline
    $111111$
    & $000000$
    & $F_{i-1}+P_i(1-x)$
    & none for $(x,y)=00$
  \end{tabular}
\end{table}

For the fourth row, we use
\[
    F_{i-1}+P_i(1-x)+12y
    =
    F_{i-1}+2^{i+3}(1-v_{i-1,1})+12v_{i-1,6}
    =
    \widehat F_{i-1}.
\]

We now continue the argument simultaneously for (A$_i$) and (B$_i$).
Whenever we refer to a row of \cref{tab:direct-induction}, the first
column gives the assignment of $X_i$ in (A$_i$), while the second
column gives its assignment in (B$_i$). By the third and fourth
columns, however, the objective on the preceding layers and the
improving variables in $X_i$ are the same in both cases. Thus, the
following argument applies to both.

After flipping $v_{i,1}$, we reach the second row of
\cref{tab:direct-induction}, and the objective on the preceding layers
is $F_{i-1}$. Hence, by (A$_{i-1}$), their unique improving sequence
is $S_{i-1}$, taking $\mathbf G_{i-1}$ to $\mathbf R_{i-1}$.

By the recursive definition of $S_{i-1}$, the variable $x=v_{i-1,1}$ is its first flip and $y=v_{i-1,6}$ is its last flip, and these two variables are flipped exactly once.
Consequently, while $S_{i-1}$ is performed, $(x,y)$ changes from
$00$ to $10$ and finally to $11$. By the second row of
\cref{tab:direct-induction}, no variable in $X_i$ is improving while
$(x,y)$ is $00$ or $10$. Thus, the sequence $S_{i-1}$ is forced to
continue. When it is complete, the preceding layers are in
$\mathbf R_{i-1}$, which is a local optimum with respect to
$F_{i-1}$ by (A$_{i-1}$), and $(x,y)=11$. The second row of the table
then shows that $v_{i,2}$ is the unique improving variable.

After flipping $v_{i,2}$, we reach the third row of the table, and
the objective on the preceding layers is $F_{i-1}+12y$. These
variables are still in $\mathbf R_{i-1}$. By (A$_{i-1}$),
$\mathbf R_{i-1}$ is a local optimum with respect to $F_{i-1}$.
Since $y=1$, the additional term $12y$ decreases the change in the
objective caused by flipping $y$, while it does not affect any other
variable. Hence, no variable in the preceding layers is improving.
By the third row of \cref{tab:direct-induction}, $v_{i,3}$ is therefore
the unique improving variable.

After flipping $v_{i,3}$, we reach the fourth row of the table. The
objective on the preceding layers is now $\widehat F_{i-1}$.
By (B$_{i-1}$), their unique improving sequence is $S_{i-1}$, now
taking $\mathbf R_{i-1}$ back to $\mathbf G_{i-1}$.

Again, $x$ is the first variable and $y$ the last variable flipped in
$S_{i-1}$. Hence, during this second occurrence of $S_{i-1}$,
$(x,y)$ changes from $11$ to $01$ and finally to $00$. By the fourth
row of \cref{tab:direct-induction}, no variable in $X_i$ is improving
while $(x,y)$ is $11$ or $01$. Thus, the sequence $S_{i-1}$ is again
forced to continue. When it is complete, the preceding layers are in
$\mathbf G_{i-1}$, which is a local optimum with respect to
$\widehat F_{i-1}$ by (B$_{i-1}$), and $(x,y)=00$. The fourth row of
the table then shows that $v_{i,4}$ is the unique improving variable.

After flipping $v_{i,4}$, we reach the fifth row of the table. The
objective on the preceding layers is again
$F_{i-1}+P_i(1-x)$, and these variables are in $\mathbf G_{i-1}$.
By the same argument as for the initial assignment, no variable in
the preceding layers is improving. The fifth row of the table
therefore shows that $v_{i,5}$ is the unique improving variable.

After flipping $v_{i,5}$, we reach the sixth row. The objective on
the preceding layers is unchanged, and their assignment is still
$\mathbf G_{i-1}$. Hence, no variable in the preceding layers is
improving, and the sixth row of the table shows that $v_{i,6}$ is the
unique improving variable.

We have therefore shown that in both (A$_i$) and (B$_i$) the unique
improving sequence is
\[
    (v_{i,1},S_{i-1},v_{i,2},v_{i,3},
     S_{i-1},v_{i,4},v_{i,5},v_{i,6})
    =S_i.
\]
Every flip in this sequence increases the objective value by exactly
$1$: this follows from the induction hypotheses for the two
occurrences of $S_{i-1}$ and from
\cref{tab:direct-induction} for the flips of variables in $X_i$.

After the final flip, the preceding layers are in
$\mathbf G_{i-1}$, and the last row of the table applies. No variable
in $X_i$ is improving, and the objective on the preceding layers is
$F_{i-1}+P_i(1-x)$. As above, no variable in the preceding layers is
improving either. Hence, the resulting assignment is a local optimum.

In (A$_i$), the variables in $X_i$ have changed from $000000$ to
$111111$, while the preceding layers have returned to
$\mathbf G_{i-1}$. Thus, the sequence ends in $\mathbf R_i$.
In (B$_i$), the variables in $X_i$ have changed from $111111$ to
$000000$, while the preceding layers have again returned to
$\mathbf G_{i-1}$. Thus, the sequence ends in $\mathbf G_i$.
This completes the induction for (A$_i$) and (B$_i$). For $i=n$, statement (A$_n$) is exactly the claim of the lemma, since $F_n$ is the objective of $I_n$ and $\mathbf G_n$ is the all-zero assignment.

It remains to count the flips. Let $L_i:=|S_i|$. From
\eqref{eq:S0} and \eqref{eq:Si}, we have $L_0=2$ and
$L_i=2L_{i-1}+6$. Therefore, $L_i=8\cdot2^i-6$ for all $i$.
\end{proof}

We now prove the main result of this section.

\begin{proof}[Proof of \cref{thm:322}]
    Observe that $I_n$ has $6n+2$ variables and $4n+1$ constraints, and every constraint weight is at most $2^{n+4}$ and can thus be encoded with $O(n)$ bits.
    Combining this with \cref{lem:322-mca}, we obtain the all-exp property for $(3,2,2)$-MCA.
\end{proof}


\section{TSP with the 2-\opt neighborhood has the all-exp property}
\label{sec:2opt}

In this section, we prove the following:

\thmtwoopt*

In particular, we construct an instance $I_n$ for each $n \in \mathbb{N}$ and an initial tour with $O(n)$ vertices.
We then show that all maximal improving swap sequences starting from this tour have length $\Omega(2^{n})$; in fact, they are all identical up to interchanging certain commuting swaps and terminate with the same local optimum.

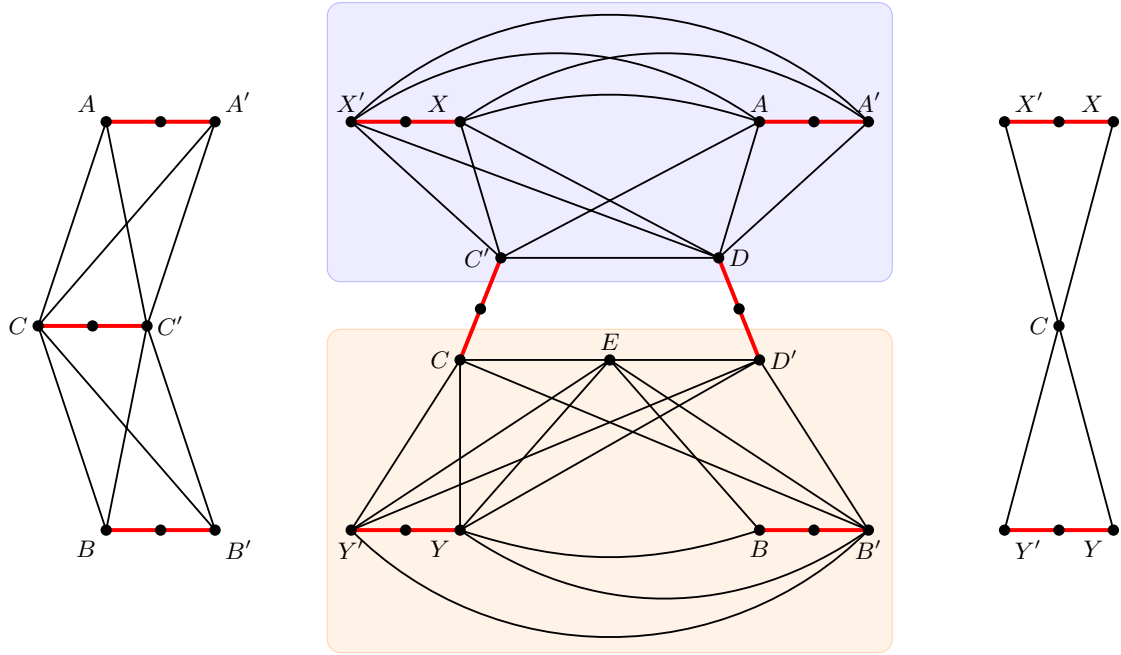
\begin{figure}[ht]
    \centering
\begin{tikzpicture}[scale=0.9,
  v/.style={circle,fill=black,inner sep=1.5pt},
  edge/.style={black,line width=0.7pt},
  rededge/.style={red,line width=1.5pt},
  lab/.style={font=\footnotesize\itshape}
]

\begin{scope}[xshift=0cm]
  \coordinate (A)  at (0,6);
  \coordinate (Ap) at (1.6,6);
  \coordinate (C)  at (-1,3);
  \coordinate (Cp) at (0.6,3);
  \coordinate (B)  at (0,0);
  \coordinate (Bp) at (1.6,0);

  \draw[rededge] (A)--(Ap) (C)--(Cp) (B)--(Bp);
  \draw[edge] (A)--(C)--(B) (Ap)--(Cp)--(Bp);
  \draw[edge] (A)--(Cp) (C)--(Ap) (C)--(Bp) (B)--(Cp);

  \foreach \p in {A,Ap,C,Cp,B,Bp} \node[v] at (\p) {};
  \node[v] at ($(A)!0.5!(Ap)$) {};
  \node[v] at ($(B)!0.5!(Bp)$) {};
  \node[v] at ($(C)!0.5!(Cp)$) {};

  \node[lab,above left]  at (A)  {$A$};
  \node[lab,above right] at (Ap) {$A'$};
  \node[lab,left]        at (C)  {$C$};
  \node[lab,right]       at (Cp) {$C'$};
  \node[lab,below left]  at (B)  {$B$};
  \node[lab,below right] at (Bp) {$B'$};
\end{scope}

\begin{scope}[xshift=3.6cm]
  \coordinate (X)   at (1.6,6);
  \coordinate (Xp)  at (0,6);
  \coordinate (A)   at (6,6);
  \coordinate (Ap)  at (7.6,6);
  \coordinate (Cp)  at (2.2,4);
  \coordinate (D)   at (5.4,4);
  \coordinate (C)   at (1.6,2.5);
  \coordinate (E)   at (3.8,2.5);
  \coordinate (Dp)  at (6,2.5);
  \coordinate (Yp)  at (0,0);
  \coordinate (Y)   at (1.6,0);
  \coordinate (B)   at (6,0);
  \coordinate (Bp)  at (7.6,0);

    \filldraw[
        rounded corners=5pt,
        fill=blue!7,
        draw=blue!28,
        line width=0.4pt
    ] (-0.35,3.65) rectangle (7.95,7.75);
    
    \filldraw[
        rounded corners=5pt,
        fill=orange!10,
        draw=orange!32,
        line width=0.4pt
    ] (-0.35,-1.80) rectangle (7.95,2.95);

  \draw[rededge]
      (X)--(Xp) (A)--(Ap) (C)--(Cp) (D)--(Dp)
      (Y)--(Yp) (B)--(Bp);

  \draw[edge]
      (X)--(Cp) (X)--(D) (Xp)--(Cp) (Xp)--(D)
      (A)--(Cp) (A)--(D) (Ap)--(D) (Cp)--(D);
  \draw[edge,bend left=45] (Xp) to (Ap);
  \draw[edge,bend left=35] (Xp) to (A);
  \draw[edge,bend left=35] (X) to (Ap);
  \draw[edge,bend left=18] (X) to (A);

  \draw[edge] (B)--(E);
  \draw[edge,bend left=18] (B) to (Y);
  \draw[edge] (Bp)--(C) (Bp)--(Dp) (Bp)--(E);
  \draw[edge,bend left=35] (Bp) to (Y);
  \draw[edge,bend left=45] (Bp) to (Yp);
  \draw[edge] (C)--(E);
  \draw[edge] (C)--(Y) (C)--(Yp);
  \draw[edge] (Dp)--(E);
  \draw[edge] (Dp)--(Y) (Dp)--(Yp);
  \draw[edge] (E)--(Y) (E)--(Yp);

  \foreach \p in {X,Xp,A,Ap,Cp,D,C,E,Dp,Y,Yp,B,Bp}
      \node[v] at (\p) {};
  \node[v] at ($(X)!0.5!(Xp)$) {};
  \node[v] at ($(A)!0.5!(Ap)$) {};
  \node[v] at ($(Y)!0.5!(Yp)$) {};
  \node[v] at ($(B)!0.5!(Bp)$) {};
  \node[v] at ($(C)!0.5!(Cp)$) {};
  \node[v] at ($(D)!0.5!(Dp)$) {};

  \node[lab,above left]  at (X)  {$X$};
  \node[lab,above]  at (Xp) {$X'$};
  \node[lab,above]       at (A)  {$A$};
  \node[lab,above] at (Ap) {$A'$};
  \node[lab,left]        at (Cp) {$C'$};
  \node[lab,right]       at (D)  {$D$};
  \node[lab,left]        at (C)  {$C$};
  \node[lab,above]       at (E)  {$E$};
  \node[lab,right]       at (Dp) {$D'$};
  \node[lab,below left]  at (Y)  {$Y$};
  \node[lab,below]  at (Yp) {$Y'$};
  \node[lab,below]       at (B)  {$B$};
  \node[lab,below] at (Bp) {$B'$};
\end{scope}

\begin{scope}[xshift=13.2cm]
  \coordinate (Xp) at (0,6);
  \coordinate (X)  at (1.6,6);
  \coordinate (C)  at (0.8,3);
  \coordinate (Yp) at (0,0);
  \coordinate (Y)  at (1.6,0);

  \draw[rededge] (X)--(Xp) (Y)--(Yp);
  \draw[edge] (X)--(C)--(Xp) (Y)--(C)--(Yp);

  \foreach \p in {X,Xp,C,Y,Yp} \node[v] at (\p) {};
  \node[v] at ($(X)!0.5!(Xp)$) {};
  \node[v] at ($(Y)!0.5!(Yp)$) {};

  \node[lab,above left]  at (X)  {$X$};
  \node[lab,above right] at (Xp) {$X'$};
  \node[lab,left]        at (C)  {$C$};
  \node[lab,below left]  at (Y)  {$Y$};
  \node[lab,below right] at (Yp) {$Y'$};
\end{scope}

\end{tikzpicture}%

    \caption{\textbf{Left}: Gadget for $F_0$. \textbf{Middle}: Upper gadget and lower gadget for $F_i, i \in [n-1]$. \textbf{Right}: Gadget for $F_n$. All subscripts are omitted.
    }
    \label{fig:circle}
\end{figure}

\begin{table}[ht!]
    \centering
\begin{tabular}{@{}l | *{9}{c}@{}}
Edge
& $AD$ & $XA$ & $XA'$ & $XC'$ & $XD$
& $X'A$ & $X'A'$ & $X'C'$ & $X'D$ \\
Weight
& $2a2^i-3$ & $1$ & $2a2^i-3$ & $1$ & $2a2^i-3$
& $a2^i-1$ & $a2^i$ & $a2^i-1$ & $a2^i$ \\
\hline
Edge
& $BY$ & $B'E$ & $B'Y'$ & $CE$
& $CY'$ & $D'Y$ & $EY$ & & \\
Weight
& $3$ & $2^{j_i}-1$ & $1$ & $1$
& $2$ & $4$ & $2^{j_i}$ & & \\
\end{tabular}
    
    \caption{Weight assignment of the non-zero-weighted edges in the upper (top) and lower gadgets (bottom)  for $F_i$, $i \in [n-1]$. 
    Subscripts $i$ are omitted.
    $a = 2^{n+2}+3$. $j_i = n+1-i$. 
    }
    
    \label{tab:circle_weight}
\end{table}

\subsection{The construction}
The instance $I_n$ consists of $n+1$ layers $F_0, F_1, \dots, F_n$.
Each layer contains one or two gadgets as depicted in \cref{fig:circle} and described in more detail below.
Two consecutive layers interact via \emph{interfaces}.
In particular, for a symbol $T$, we call the path $(T, T'', T')$ with edge weights zero a \defi{$T$-interface}.
We also define $a := 2^{n+2}+3$ and $j_i := n+1-i$ for $i \in [n-1]$.

\paragraph{Gadgets.}
First, $F_0$ has an $A_0$-interface, a $B_0$-interface, together with three vertices $C_0$, $C''_0$, and $C'_0$.
The vertices $C_0$ and $C'_0$ are connected with $C''_0$, $A_0, A'_0, B_0$ and $B'_0$.
Among these edges, $A_0C'_0$, $A'_0C_0$, $C_0B_0$, and $C'_0B'_0$ have weights $a$, $4$, $3$, and $a-1$, respectively, while other edges have weight zero.

Second, for $i \in [n-1]$, $F_i$ consists of two gadgets connected by two paths of length two.
The \defi{upper gadget} of $F_i$ consists of an $X_i$-interface, an $A_i$-interface, and two vertices $C'_i$ and $D_i$.
The \defi{lower gadget} of $F_i$ consists of a $Y_i$-interface, a $B_i$-interface, and three vertices $C_i$, $D'_i$, and $E_i$.
We add two paths of length two $(C'_i,C''_i,C_i)$ and $(D'_i,D''_i,D_i)$ with edge weights zero.
The edges with nonzero weights in this layer are listed in \cref{tab:circle_weight}.
Additionally, we have the following edges with zero weights:
\[
A_iC'_i, A'_iD_i, C'_iD_i, D'_iE_i, C_iY_i, E_iY'_i, B_iE_i, D'_iY'_i, B'_iD'_i, B'_iC_i, B'_iY_i.
\]

Last, $F_n$ has an $X_n$-interface, a $Y_n$-interface, and a vertex $C_n$. 
We then add the edge $C_nX_n$ with weight $a 2 ^n$ (recall that $a = 2^{n+2}+3$), and edges $C_nX'_n, C_nY_n$, and $C_nY'_n$ with weight zero. 

\paragraph{Interface identification.}
For two consecutive layers $F_i$ and $F_{i+1}$, we identify the $A_i$-interface with the reverse of the $X_{i+1}$-interface; more precisely, $A_i \equiv X'_{i+1}, A''_i \equiv X''_{i+1}$, and $A'_i\equiv X_{i+1}$.
Similarly, we identify the $B_i$-interface with the reverse of the $Y_{i+1}$-interface; that is $B_i \equiv Y'_{i+1}, B''_i \equiv Y''_{i+1}$, and $B'_i\equiv Y_{i+1}$.
    
\paragraph{Initial tour.} 
For a symbol $T$, we denote by $T \to T'$ the path $(T, T'', T')$ and $T' \to T$ the reverse.
Then consider $\tau_0$ to be the union of the edges in the following paths:
\begin{itemize}
    \item In $F_0$, the path $(A'_0 \to A_0, C_0 \to C'_0, B_0 \to B'_0)$; 
    \item In $F_i$ for $i \in [n-1]$, the two paths 
    \[
    (A'_i \to A_i, X_i \to X'_i) \qquad \text{and} \qquad (Y'_i \to Y_i, C_i \to C'_i, D_i \to D'_i, E_i, B_i \to B'_i);
    \]
    \item In $F_n$, the path $(Y'_n \to Y_n, C_n, X_n \to X'_n)$. 
\end{itemize}

By construction, $\tau_0$ is then a Hamiltonian cycle.
Note that descriptions of paths in consecutive layers above overlap precisely on their identified interfaces.
See \cref{fig:initialtour} for an illustration of $I_4$ together with the prescribed initial tour.

Let $G_n$ denote the graph obtained from the construction above.
Note that formally a TSP instance requires a complete graph.
We can simply obtain this complete graph $G'_n$ from $G_n$ by adding the missing edges, each with a very large weight (say $2M + 1$, where $M$ is the largest edge weight in $G_n$).
The initial tour contains only edges of $G_n$.
Moreover, any $2$-swap that introduces an edge of $E(G'_n)\setminus E(G_n)$ adds weight at least $2M+1$, while the two removed edges have total weight at most $2M$.
Hence, such a swap cannot be improving.
Consequently, no improving 2-\opt sequence starting from the initial tour ever uses an edge of $E(G'_n)\setminus E(G_n)$, and we may restrict our analysis from now on to $G_n$.

\begin{figure}[ht!]
    \centering
\begin{tikzpicture}[scale=0.7,
  v/.style={circle,fill=black,inner sep=1.5pt},
  edge/.style={gray,line width=0.25pt},
  rededge/.style={red,line width=1.5pt},
  lab/.style={font=\footnotesize\itshape}
]

\begin{scope}[xshift=0cm]
  \coordinate (A)  at (0,6);
  \coordinate (Ap) at (1.6,6);
  \coordinate (C)  at (-1,3);
  \coordinate (Cp) at (0.6,3);
  \coordinate (B)  at (0,0);
  \coordinate (Bp) at (1.6,0);

  \draw[edge] (C)--(B) (Ap)--(Cp)--(Bp);
  \draw[edge] (A)--(Cp) (C)--(Ap) (C)--(Bp);
  \draw[rededge] (Ap)--(A)--(C)--(Cp)--(B)--(Bp);
  
  \foreach \p in {A,Ap,C,Cp,B,Bp} \node[v] at (\p) {};
  \node[v] at ($(A)!0.5!(Ap)$) {};
  \node[v] at ($(B)!0.5!(Bp)$) {};
  \node[v] at ($(C)!0.5!(Cp)$) {};
\end{scope}

\begin{scope}[xshift=0cm]
  \coordinate (X)   at (1.6,6);
  \coordinate (Xp)  at (0,6);
  \coordinate (A)   at (6,6);
  \coordinate (Ap)  at (7.6,6);
  \coordinate (Cp)  at (2.2,4);
  \coordinate (D)   at (5.4,4);
  \coordinate (C)   at (1.6,2.5);
  \coordinate (E)   at (3.8,2.5);
  \coordinate (Dp)  at (6,2.5);
  \coordinate (Yp)  at (0,0);
  \coordinate (Y)   at (1.6,0);
  \coordinate (B)   at (6,0);
  \coordinate (Bp)  at (7.6,0);

  \draw[rededge] (A)--(Ap) (C)--(Cp) (D)--(Dp) (B)--(Bp);

  \draw[edge] (X)--(Cp) (X)--(D) (Xp)--(Cp) (Xp)--(D) (A)--(Cp) (A)--(D) (Ap)--(D); 
  \draw[rededge] (Cp)--(D);
  \draw[edge,bend left=45] (Xp) to (Ap);
  \draw[edge,bend left=35] (Xp) to (A);
  \draw[edge,bend left=35] (X) to (Ap);
  \draw[rededge,bend left=18] (X) to (A);

  \draw[edge,bend left=18] (B) to (Y);
  \draw[edge] (Bp)--(C) (Bp)--(Dp)  (Bp)--(E);
  \draw[edge,bend left=35] (Bp) to (Y);
  \draw[edge,bend left=45] (Bp) to (Yp);
  \draw[edge] (C)--(E);
  \draw[edge]  (C)--(Yp);
  \draw[edge] (Dp)--(Y) (Dp)--(Yp);
  \draw[edge] (E)--(Y)  (E)--(Yp);
  \draw[rededge] (C)--(Y) (Dp)--(E) (B)--(E);

  \foreach \p in {X,Xp,Cp,D,C,E,Dp,Y,Yp}
      \node[v] at (\p) {};
  \node[v] at ($(A)!0.5!(Ap)$) {};
  \node[v] at ($(B)!0.5!(Bp)$) {};
  \node[v] at ($(C)!0.5!(Cp)$) {};
  \node[v] at ($(D)!0.5!(Dp)$) {};
\end{scope}

\begin{scope}[xshift=6cm]
  \coordinate (X)   at (1.6,6);
  \coordinate (Xp)  at (0,6);
  \coordinate (A)   at (6,6);
  \coordinate (Ap)  at (7.6,6);
  \coordinate (Cp)  at (2.2,4);
  \coordinate (D)   at (5.4,4);
  \coordinate (C)   at (1.6,2.5);
  \coordinate (E)   at (3.8,2.5);
  \coordinate (Dp)  at (6,2.5);
  \coordinate (Yp)  at (0,0);
  \coordinate (Y)   at (1.6,0);
  \coordinate (B)   at (6,0);
  \coordinate (Bp)  at (7.6,0);

  \draw[rededge] (A)--(Ap) (C)--(Cp) (D)--(Dp) (B)--(Bp);

  \draw[edge] (X)--(Cp) (X)--(D) (Xp)--(Cp) (Xp)--(D) (A)--(Cp) (A)--(D) (Ap)--(D); 
  \draw[rededge] (Cp)--(D);
  \draw[edge,bend left=45] (Xp) to (Ap);
  \draw[edge,bend left=35] (Xp) to (A);
  \draw[edge,bend left=35] (X) to (Ap);
  \draw[rededge,bend left=18] (X) to (A);

  \draw[edge,bend left=18] (B) to (Y);
  \draw[edge] (Bp)--(C) (Bp)--(Dp)  (Bp)--(E);
  \draw[edge,bend left=35] (Bp) to (Y);
  \draw[edge,bend left=45] (Bp) to (Yp);
  \draw[edge] (C)--(E);
  \draw[edge]  (C)--(Yp);
  \draw[edge] (Dp)--(Y) (Dp)--(Yp);
  \draw[edge] (E)--(Y)  (E)--(Yp);
  \draw[rededge] (C)--(Y) (Dp)--(E) (B)--(E);

  \foreach \p in {X,Xp,Cp,D,C,E,Dp,Y,Yp}
      \node[v] at (\p) {};
  \node[v] at ($(A)!0.5!(Ap)$) {};
  \node[v] at ($(B)!0.5!(Bp)$) {};
  \node[v] at ($(C)!0.5!(Cp)$) {};
  \node[v] at ($(D)!0.5!(Dp)$) {};

\end{scope}

\begin{scope}[xshift=12cm]
  \coordinate (X)   at (1.6,6);
  \coordinate (Xp)  at (0,6);
  \coordinate (A)   at (6,6);
  \coordinate (Ap)  at (7.6,6);
  \coordinate (Cp)  at (2.2,4);
  \coordinate (D)   at (5.4,4);
  \coordinate (C)   at (1.6,2.5);
  \coordinate (E)   at (3.8,2.5);
  \coordinate (Dp)  at (6,2.5);
  \coordinate (Yp)  at (0,0);
  \coordinate (Y)   at (1.6,0);
  \coordinate (B)   at (6,0);
  \coordinate (Bp)  at (7.6,0);

  \draw[rededge] (A)--(Ap) (C)--(Cp) (D)--(Dp) (B)--(Bp);

  \draw[edge] (X)--(Cp) (X)--(D) (Xp)--(Cp) (Xp)--(D) (A)--(Cp) (A)--(D) (Ap)--(D); 
  \draw[rededge] (Cp)--(D);
  \draw[edge,bend left=45] (Xp) to (Ap);
  \draw[edge,bend left=35] (Xp) to (A);
  \draw[edge,bend left=35] (X) to (Ap);
  \draw[rededge,bend left=18] (X) to (A);

  \draw[edge,bend left=18] (B) to (Y);
  \draw[edge] (Bp)--(C) (Bp)--(Dp)  (Bp)--(E);
  \draw[edge,bend left=35] (Bp) to (Y);
  \draw[edge,bend left=45] (Bp) to (Yp);
  \draw[edge] (C)--(E);
  \draw[edge]  (C)--(Yp);
  \draw[edge] (Dp)--(Y) (Dp)--(Yp);
  \draw[edge] (E)--(Y)  (E)--(Yp);
  \draw[rededge] (C)--(Y) (Dp)--(E) (B)--(E);

  \foreach \p in {X,Xp,Cp,D,C,E,Dp,Y,Yp}
      \node[v] at (\p) {};
  \node[v] at ($(A)!0.5!(Ap)$) {};
  \node[v] at ($(B)!0.5!(Bp)$) {};
  \node[v] at ($(C)!0.5!(Cp)$) {};
  \node[v] at ($(D)!0.5!(Dp)$) {};

\end{scope}

\begin{scope}[xshift=18cm]
  \coordinate (Xp) at (0,6);
  \coordinate (X)  at (1.6,6);
  \coordinate (C)  at (0.8,3);
  \coordinate (Yp) at (0,0);
  \coordinate (Y)  at (1.6,0);

  \draw[rededge] (X)--(C)--(Y);
  \draw[edge] (Xp)--(C)--(Yp);

  \foreach \p in {X,Xp,Y,Yp,C}
      \node[v] at (\p) {};
\end{scope}

\end{tikzpicture}%

    \caption{Instance for $n=4$ with the initial tour shown in red.}
    \label{fig:initialtour}
\end{figure}
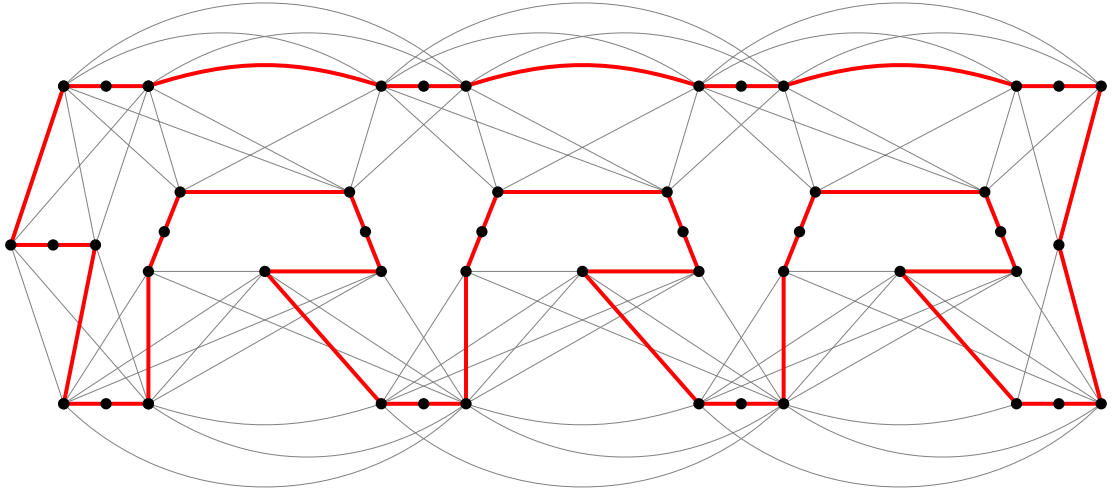

\subsection{Tour characterization}
In this subsection, we characterize the tours of $G_n$ via the notion of the layer-path invariant defined below.

\begin{definition}
\label{def:2opt_layer_path}
An edge set $H$ of $G_n$ satisfies the \defi{layer-path invariant} if the following holds:
\begin{enumerate}
\item The restriction $H_0$ of $H$ to $F_0$ is a spanning path of $F_0$ whose two endpoints lie on the $A_0$- and $B_0$-interfaces, respectively.

\item The restriction $H_n$ of $H$ to $F_n$ is a spanning path of $F_n$ whose two endpoints lie on the $X_n$- and $Y_n$-interfaces, respectively.

\item For every $i\in[n-1]$, the restriction $H_i$ of $H$ to $F_i$ is the disjoint union of two paths such that this union spans $F_i$. 
Further, each of these paths has one endpoint on one of the $X_i$- and $Y_i$-interfaces and one endpoint on one of the $A_i$- and $B_i$-interfaces. Moreover, the two paths have their endpoints on distinct interfaces; thus each of the four interfaces $X_i,Y_i,A_i,B_i$ contains exactly one endpoint.
\end{enumerate}
\end{definition}

We have the following simple observation.

\begin{observation}
\label{lem:2opt_edges}
    Every tour of $G_n$ contains the edges incident to degree-two vertices (i.e., $A''_i$, $B''_i$, and $C''_i$ for $i \in \{0, \dots, n-1\}$ and $D''_i$ for $i \in [n-1]$).
\end{observation}

Then we can characterize all tours of $G_n$ as follows.

\begin{lemma}
\label{lem:2opt_tour_char}
    An edge set $H$ forms a tour of $G_n$, if and only if it is a 2-regular subgraph of $G_n$ satisfying the layer-path invariant.
\end{lemma}
\begin{proof}
    $(\Rightarrow)$: 
    Suppose $H$ forms a tour of $G_n$, then it is clearly a 2-regular subgraph.
    Observe that for every layer $F_i$ for $i \in \{0, \dots, n-1\}$, if we remove the $A_i$- and $B_i$-interfaces, then we disconnect the remainder of $F_0 \cup \dots \cup F_i$ from the remainder of $F_{i+1} \cup \dots \cup F_n$.
    This implies that
    \begin{itemize}
        \item[(a)] $H$ restricted to the graph $F_0 \cup \dots \cup F_i$ is a path between a vertex of the $A_i$-interface and a vertex of the $B_i$-interface, and 
        \item[(b)] $H$ restricted to the graph $F_{i+1} \cup \dots \cup F_n$ is a path between a vertex of the $X_{i+1}$-interface and a vertex of the $Y_{i+1}$-interface.
    \end{itemize}
    Together with \cref{lem:2opt_edges}, these already imply the first two conditions in \cref{def:2opt_layer_path}.
    For the last condition, by (a), we also obtain that for $i \geq 1$, $H$ restricted to $F_0 \cup \dots \cup F_{i-1}$ is a path from a vertex of the $A_{i-1}$-interface and a vertex of the $B_{i-1}$-interface.
    Combining this with (a) and \cref{lem:2opt_edges}, we obtain the last condition of \cref{def:2opt_layer_path}.

    $(\Leftarrow)$: Starting with the unique path in $H_0$, choose one of its endpoints, say the one on the $A_0$-interface, and traverse the path to its endpoint on the $B_0$-interface.
    Consider first the $A_0$-endpoint. 
    Since the $A_0$-interface is identified with the (reverse of) $X_1$-interface, it belongs to one of the two paths of $H_1$. 
    By the layer-path invariant, the other endpoint of this path lies on either the $A_1$- or the $B_1$-interface. 
    Continuing in this way, at every intermediate layer we enter through one of its $X$- or $Y$-interfaces and leave through one of its $A$- or $B$-interfaces. 
    Hence we eventually reach one of the two interfaces $X_n,Y_n$ of $F_n$.

    The unique path $H_n$ joins these two interfaces. 
    After traversing it, we therefore return through the other sequence of paths: in every intermediate layer, the second path joins the remaining $X$- or $Y$-interface to the remaining $A$- or $B$-interface. 
    Eventually, this traversal reaches the $B_0$-interface and hence the other endpoint of $H_0$.

    Thus all paths occurring in every $H_i$ belong to a single connected component of $H$. 
    Since the paths in the layer-path invariant span every layer, this component contains every vertex of $G_n$. 
    As $H$ is $2$-regular, it follows that $H$ is a Hamiltonian cycle.
\end{proof}

\subsection{States of the tour}
We describe a tour of $G_n$ by assigning states to its interfaces, to the first layer, and to the upper and lower gadgets.
More specifically, we associate $C_0$ with the part of $F_0$ between its two interfaces.
For $i\in\{0,\ldots,n-1\}$, we associate the symbols $A_i$ and $B_i$ with the corresponding interfaces.
For $i\in[n-1]$, we associate $D_i$ and $E_i$ with the upper and lower gadgets of $F_i$, respectively.
We use the same notation for these symbols and the corresponding vertices; the meaning will be clear from the context.

Fix a tour $\tau$ of $G_n$.
For a $T$-interface, we denote by $[T]$ the set $\{T,T'\}$.
By \cref{lem:2opt_edges}, both edges of the interface belong to $\tau$, so each vertex in $[T]$ has precisely one neighbor outside of the interface in $\tau$.

\begin{definition}\label{def:2opt_states}
The \defi{states} associated with $\tau$ are defined as follows.
\begin{itemize}
\item For $i\in\{0,\ldots,n-1\}$, the state of $A_i$ is $0$ if $A_i$ is adjacent in $\tau$ to a vertex of $F_i$ outside its interface, and $1$ if $A'_i$ is. The state of $B_i$ is defined analogously. By \cref{lem:2opt_tour_char}, exactly one of the two endpoints has such a neighbor in $F_i$.
\item The state of $C_0$ is $0$ if its neighbor other than $C''_0$ belongs to $[A_0]$, and $1$ if that neighbor belongs to $[B_0]$. The other interface is joined to $C'_0$.
\item For $i\in[n-1]$, the state of $D_i$ is $0$, $1$, or $2$ according as its neighbor other than $D''_i$ is $C'_i$, belongs to $[A_i]$, or belongs to $[X_i]$. The resulting three pairings in the upper gadget are shown in \cref{fig:2opt_state_D}.
\item For $i\in[n-1]$, the state of $E_i$ is determined by its two neighbors in $\tau$ as follows:
\end{itemize}
\begin{center}
\begin{tabular}{@{}cl@{}}
\hline
State & Neighbors of $E_i$ \\
\hline
$0$ & $D'_i$ and $[B_i]$ \\
$1$ & $C_i$ and $[B_i]$ \\
$2$ & $C_i$ and $D'_i$ \\
$3$ & $C_i$ and $[Y_i]$ \\
$4$ & $D'_i$ and $[Y_i]$ \\
\hline
\end{tabular}
\end{center}
The five corresponding pairings in the lower gadget are shown in \cref{fig:circle_states}.
\end{definition}

The states of the interfaces determine which endpoint is used by each gadget. Together with the pairings in \cref{fig:2opt_state_D,fig:circle_states} and the fixed edges from \cref{lem:2opt_edges}, they therefore determine every edge of a reachable tour. All states of $\tau_0$ are $0$.

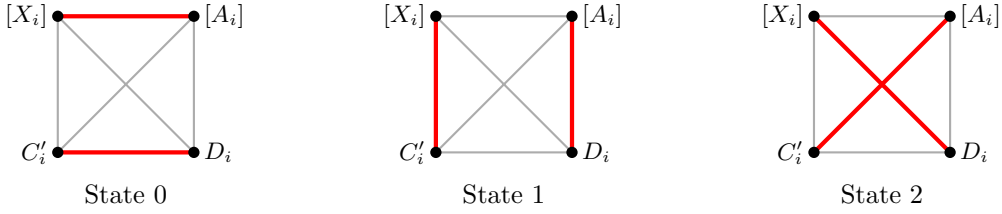
\begin{figure}[ht!]
    \centering
    \begin{tikzpicture}[
      v/.style={circle,fill=black,inner sep=1.5pt},
      edge/.style={gray!65,line width=0.8pt},
      rededge/.style={red,line width=1.6pt},
      lab/.style={font=\footnotesize\itshape},
      statelab/.style={font=\small}
    ]
    
    \newcommand{\EstatePanel}[3]{%
      \begin{scope}[shift={#1}]
        \coordinate (X)  at (0,1.8);
        \coordinate (A) at (1.8,1.8);
        \coordinate (C)  at (0,0);
        \coordinate (D)  at (1.8,0);
    
        \draw[edge] (X) -- (A) -- (D) -- (C) -- (X) -- (D) (C) -- (A);

        \draw[rededge] #3;
    
        \foreach \p in {X, A, D, C}
          \node[v] at (\p) {};
    
        \node[lab,left]       at (X)  {$[X_i]$};
        \node[lab,right]      at (A) {$[A_i]$};
        \node[lab,right]      at (D)  {$D_i$};
        \node[lab,left]      at (C) {$C'_i$};
    
        \node[statelab] at (0.9,-0.55) {State #2};
      \end{scope}%
    }
    
    \EstatePanel{(0,0)}{0}
      {(C)--(D) (X)--(A)}
    
    \EstatePanel{(5,0)}{1}
      {(A)--(D) (X)--(C)}
    
    \EstatePanel{(10,0)}{2}
      {(X)--(D) (A)--(C)}
    
    \end{tikzpicture}
    \caption{The three states of $D_i$ and the inferred connection in the upper gadget of $F_i$}
    \label{fig:2opt_state_D}
\end{figure}

\begin{figure}[ht!]
    \centering
    \begin{tikzpicture}[
      v/.style={circle,fill=black,inner sep=1.5pt},
      edge/.style={gray!65,line width=0.8pt},
      rededge/.style={red,line width=1.6pt},
      lab/.style={font=\footnotesize\itshape},
      statelab/.style={font=\small}
    ]
    
    
    \newcommand{\EstatePanel}[5]{%
      \begin{scope}[shift={#1}]
        \coordinate (C)  at (0,1.8);
        \coordinate (E)  at (1.5,1.8);
        \coordinate (Dp) at (3,1.8);
        \coordinate (Y)  at (0,0);
        \coordinate (B)  at (3,0);
    
        \draw[edge] (C) -- (E) -- (Dp);
        \draw[edge]              (E) -- (Y);
        \draw[edge]              (E) -- (B);
        \draw[edge]              (B) -- (C) -- (Y) -- (B) -- (Dp);
        \draw[edge]              (Dp) -- (Y);
    
        #3
        #4
        #5
    
        \foreach \p in {C,E,Dp,Y,B}
          \node[v] at (\p) {};
    
        \node[lab,above]       at (C)  {$C_i$};
        \node[lab,above]      at (E)  {$E_i$};
        \node[lab,above]      at (Dp) {$D'_i$};
        \node[lab,below] at (Y)  {$[Y_i]$};
        \node[lab,below]      at (B)  {$[B_i]$};
    
        \node[statelab] at (1.5,-0.55) {State #2};
      \end{scope}%
    }
    
    \EstatePanel{(0,0)}{0}
      {\draw[rededge] (E)--(Dp);}
      {\draw[rededge] (E)--(B);}
      {\draw[rededge] (C)--(Y);}
    
    \EstatePanel{(5,0)}{1}
      {\draw[rededge] (C) -- (E);}
      {\draw[rededge] (E)--(B);}
      {\draw[rededge] (Dp)--(Y);}
    
    \EstatePanel{(10,0)}{2}
      {\draw[rededge] (C) -- (E);}
      {\draw[rededge] (E) -- (Dp);}
      {\draw[rededge] (Y)--(B);}
    
    \EstatePanel{(2.5,-3.5)}{3}
      {\draw[rededge] (C) -- (E);}
      {\draw[rededge] (E)--(Y);}
      {\draw[rededge] (Dp)--(B);}
    
    \EstatePanel{(7.5,-3.5)}{4}
      {\draw[rededge] (E)--(Dp);}
      {\draw[rededge] (E)--(Y);}
      {\draw[rededge] (C)--(B);}
    
    \end{tikzpicture}
    \caption{The five states of $E_i$ and the inferred connection in the lower gadget of $F_i$}
    \label{fig:circle_states}
\end{figure}
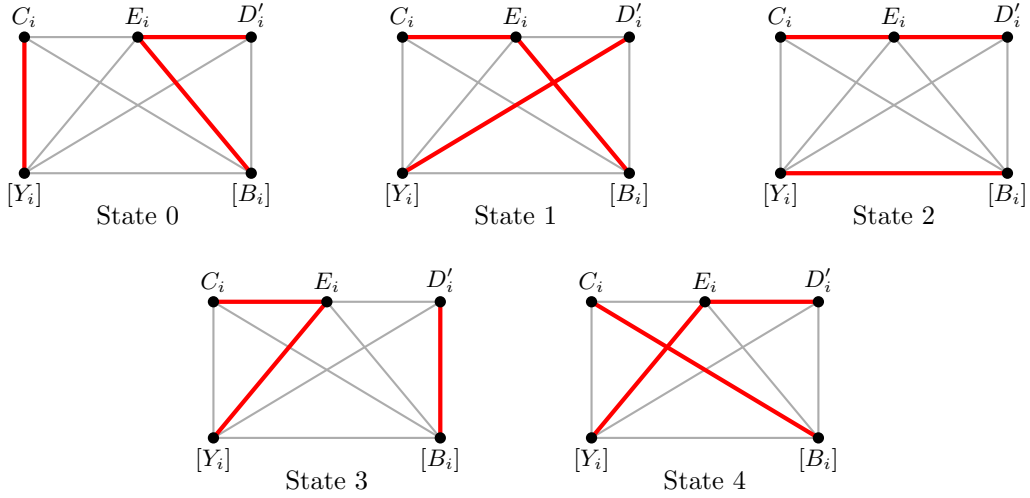

The lemma below describes which states encode valid tours and also asserts that a 2-swap corresponds to a change in exactly one state.

\begin{lemma}\label{lem:2opt_valid_states}
An assignment of states determines a tour if and only if, for every $i\in[n-1]$, the following hold:
\begin{itemize}
    \item the pair $(D_i,E_i)$ is not among $(0,2)$, $(1,0)$, $(1,3)$, $(2,1)$, and $(2,4)$;
    \item $D_i=2$ implies $A_i=0$;
    \item if $E_i=2$, then $(B_{i-1},B_i)\neq(1,0)$;
    \item if $E_i\in\{3,4\}$, then $B_i=1$.
\end{itemize}
Further, every $2$-swap between tours changes exactly one state. 
Conversely, if two valid state assignments differ in exactly one state, then the corresponding tours differ by a $2$-swap, except when that state is $E_i$ and its two values are $0$ and $3$, or $1$ and $4$.
\end{lemma}
\begin{figure}[ht!]
    \centering
\begin{tikzpicture}[scale=0.7,
  v/.style={circle,fill=black,inner sep=1.5pt},
  edge/.style={gray,line width=0.25pt},
  rededge/.style={line width=1.25pt},
  blueedge/.style={blue,line width=1.5pt},
  blackedge/.style={red, line width=1.5pt},
  lab/.style={font=\footnotesize\itshape}
]

\begin{scope}[xshift=0cm]
  \coordinate (A)  at (0,6);
  \coordinate (Ap) at (1.6,6);
  \coordinate (C)  at (-1,3);
  \coordinate (Cp) at (0.6,3);
  \coordinate (B)  at (0,0);
  \coordinate (Bp) at (1.6,0);

  \draw[edge] (C)--(B) (Ap)--(Cp)--(Bp);
  \draw[edge] (A)--(Cp) (C)--(Ap) (C)--(Bp);
  \draw[rededge] (Ap)--(A) (C)--(Cp) (B)--(Bp);
  \draw[blackedge] (A)--(C)(Cp)--(B);
  \draw[blueedge] (A)--(Cp)(C)--(B);
  
  \foreach \p in {A,Ap,C,Cp,B,Bp} \node[v] at (\p) {};
  \node[v] at ($(A)!0.5!(Ap)$) {};
  \node[v] at ($(B)!0.5!(Bp)$) {};
  \node[v] at ($(C)!0.5!(Cp)$) {};
\end{scope}

\begin{scope}[xshift=12cm]
  \coordinate (X)   at (1.6,6);
  \coordinate (Xp)  at (0,6);
  \coordinate (A)   at (6,6);
  \coordinate (Ap)  at (7.6,6);
  \coordinate (Cp)  at (2.2,4);
  \coordinate (D)   at (5.4,4);
  \coordinate (C)   at (1.6,2.5);
  \coordinate (E)   at (3.8,2.5);
  \coordinate (Dp)  at (6,2.5);
  \coordinate (Yp)  at (0,0);
  \coordinate (Y)   at (1.6,0);
  \coordinate (B)   at (6,0);
  \coordinate (Bp)  at (7.6,0);

  \draw[rededge] (A)--(Ap) (C)--(Cp) (D)--(Dp) (B)--(Bp);

  \draw[edge] (X)--(D) (Xp)--(Cp) (Xp)--(D) (X)--(Cp)  (A)--(D) -- (Ap); 
  \draw[blueedge] (A)--(Cp) (X)--(D);
  \draw[blackedge] (Cp)--(D);
  \draw[edge,bend left=45] (Xp) to (Ap);
  \draw[edge,bend left=35] (Xp) to (A);
  \draw[edge,bend left=35] (X) to (Ap);
  \draw[blackedge,bend left=18] (X) to (A);

  \draw[edge,bend left=18] (B) to (Y);
  \draw[edge] (Bp)--(C) (Bp)--(Dp)  (Bp)--(E);
  \draw[edge,bend left=35] (Bp) to (Y);
  \draw[edge,bend left=45] (Bp) to (Yp);
  \draw[edge] (C)--(E);
  \draw[edge]  (C)--(Yp);
  \draw[edge] (Dp)--(Y) (Dp)--(Yp);
  \draw[edge] (E)--(Y)  (E)--(Yp);
  \draw[rededge] (C)--(Y) (Dp)--(E) (B)--(E);

  \foreach \p in {X,Xp,Cp,D,C,E,Dp,Y,Yp}
      \node[v] at (\p) {};
  \node[v] at ($(A)!0.5!(Ap)$) {};
  \node[v] at ($(B)!0.5!(Bp)$) {};
  \node[v] at ($(C)!0.5!(Cp)$) {};
  \node[v] at ($(D)!0.5!(Dp)$) {};
\end{scope}

\begin{scope}[xshift=0cm]
  \coordinate (X)   at (1.6,6);
  \coordinate (Xp)  at (0,6);
  \coordinate (A)   at (6,6);
  \coordinate (Ap)  at (7.6,6);
  \coordinate (Cp)  at (2.2,4);
  \coordinate (D)   at (5.4,4);
  \coordinate (C)   at (1.6,2.5);
  \coordinate (E)   at (3.8,2.5);
  \coordinate (Dp)  at (6,2.5);
  \coordinate (Yp)  at (0,0);
  \coordinate (Y)   at (1.6,0);
  \coordinate (B)   at (6,0);
  \coordinate (Bp)  at (7.6,0);

  \draw[rededge] (A)--(Ap) (C)--(Cp) (D)--(Dp) (B)--(Bp);

  \draw[edge] (X)--(Cp) (X)--(D) (Xp)--(Cp) (Xp)--(D) (A)--(Cp) (A)--(D) (Ap)--(D); 
  \draw[rededge] (Cp)--(D);
  \draw[edge,bend left=45] (Xp) to (Ap);
  \draw[edge,bend left=35] (Xp) to (A);
  \draw[edge,bend left=35] (X) to (Ap);
  \draw[rededge,bend left=18] (X) to (A);

  \draw[edge,bend left=18] (B) to (Y);
  \draw[edge] (Bp)--(C) (Bp)--(Dp)  (Bp)--(E);
  \draw[edge,bend left=35] (Bp) to (Y);
  \draw[edge,bend left=45] (Bp) to (Yp);
  \draw[edge] (C)--(E);
  \draw[edge]  (C)--(Yp);
  \draw[edge] (Dp)--(Y) (Dp)--(Yp);
  \draw[edge] (E)--(Y)  (E)--(Yp);
  \draw[rededge] (C)--(Y) (Dp)--(E);
  \draw[blackedge] (B)--(E);
  \draw[blueedge] (Bp)--(E);

  \foreach \p in {X,Xp,Cp,D,C,E,Dp,Y,Yp}
      \node[v] at (\p) {};
  \node[v] at ($(A)!0.5!(Ap)$) {};
  \node[v] at ($(B)!0.5!(Bp)$) {};
  \node[v] at ($(C)!0.5!(Cp)$) {};
  \node[v] at ($(D)!0.5!(Dp)$) {};

\end{scope}

\begin{scope}[xshift=6cm]
  \coordinate (X)   at (1.6,6);
  \coordinate (Xp)  at (0,6);
  \coordinate (A)   at (6,6);
  \coordinate (Ap)  at (7.6,6);
  \coordinate (Cp)  at (2.2,4);
  \coordinate (D)   at (5.4,4);
  \coordinate (C)   at (1.6,2.5);
  \coordinate (E)   at (3.8,2.5);
  \coordinate (Dp)  at (6,2.5);
  \coordinate (Yp)  at (0,0);
  \coordinate (Y)   at (1.6,0);
  \coordinate (B)   at (6,0);
  \coordinate (Bp)  at (7.6,0);

  \draw[rededge] (A)--(Ap) (C)--(Cp) (D)--(Dp) (B)--(Bp);

  \draw[edge] (X)--(Cp) (X)--(D) (Xp)--(Cp) (Xp)--(D) (A)--(Cp) (A)--(D) (Ap)--(D); 
  \draw[rededge] (Cp)--(D);
  \draw[edge,bend left=45] (Xp) to (Ap);
  \draw[edge,bend left=35] (Xp) to (A);
  \draw[edge,bend left=35] (X) to (Ap);
  \draw[rededge,bend left=18] (X) to (A);

  \draw[edge,bend left=18] (B) to (Y);
  \draw[edge] (Bp)--(C) (Bp)--(Dp)  (Bp)--(E);
  \draw[edge,bend left=35] (Bp) to (Y);
  \draw[edge,bend left=45] (Bp) to (Yp);
  \draw[edge] (C)--(E);
  \draw[edge]  (C)--(Yp);
  \draw[edge] (Dp)--(Y) (Dp)--(Yp);
  \draw[edge] (E)--(Y)  (E)--(Yp);
  \draw[rededge] (Dp)--(E) (B)--(E);
  \draw[blackedge] (C)--(Y);
  \draw[blueedge] (C)--(Yp);
  
  \foreach \p in {X,Xp,Cp,D,C,E,Dp,Y,Yp,A,Ap,B,Bp}
      \node[v] at (\p) {};
  \node[v] at ($(A)!0.5!(Ap)$) {};
  \node[v] at ($(B)!0.5!(Bp)$) {};
  \node[v] at ($(C)!0.5!(Cp)$) {};
  \node[v] at ($(D)!0.5!(Dp)$) {};

\end{scope}

\begin{scope}[xshift=18cm]
  \coordinate (Xp) at (0,6);
  \coordinate (X)  at (1.6,6);
  \coordinate (C)  at (0.8,3);
  \coordinate (Yp) at (0,0);
  \coordinate (Y)  at (1.6,0);

  \draw[rededge] (X)--(C)--(Y);
  \draw[edge] (Xp)--(C)--(Yp);

  \foreach \p in {X,Xp,Y,Yp,C}
      \node[v] at (\p) {};
\end{scope}

\end{tikzpicture}%

    \caption{Instance for $n=4$ with the initial tour (shown in bold black) and 3 possible $2$-swaps. Each is depicted by a $4$-cycle with the edges that can be removed from the tour shown in red and the edges that can be added to the tour shown in blue. The $2$-swaps on the left, middle and right change $C_0$, $B_1$ and $D_3$, respectively.
    }
    \label{fig:swapsintour}
\end{figure}
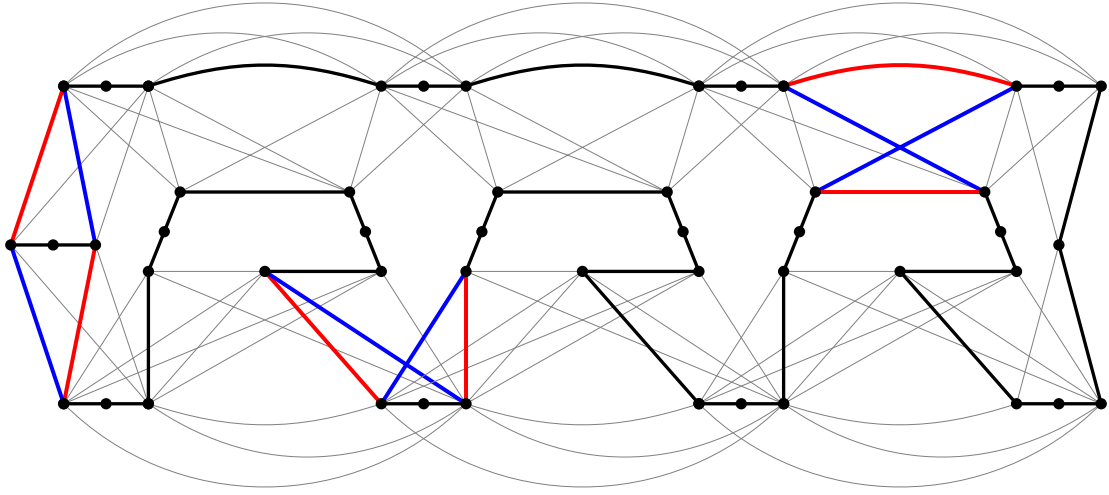

\begin{proof}
    The edges incident to degree-two vertices and the states determine the pairings within each upper and lower gadget, as shown in \cref{fig:2opt_state_D,fig:circle_states}.
    The five forbidden pairs for $(D_i,E_i)$ are exactly those that either create a cycle within $F_i$ or fail to connect the interfaces according to \cref{def:2opt_layer_path}. 
    Further, the last three conditions in the lemma statement exclude exactly the assignments requiring graph edges that do not exist.
    Conversely, under these conditions, the pairings form two spanning paths in every interior layer, so the layer-path invariant holds. 
    By \cref{lem:2opt_tour_char}, they determine a tour.

    Every $2$-swap corresponds to a $4$-cycle (i.e., a cycle of length four).
    By \cref{lem:2opt_edges}, no edge incident to a degree-two vertex can be removed. 
    Hence, by the structure of $G_n$, a $4$-cycle either lies within $F_0$, within one upper or lower gadget, or contains edges from two gadgets connected through an interface between consecutive layers.
    In the first case, the swap changes only $C_0$; in the second case, it changes only $D_i$ or $E_i$; and in the last case, it changes only $A_i$ or $B_i$.
    Thus every $2$-swap changes exactly one state. \cref{fig:swapsintour} depicts examples for the three kinds of $2$-swaps. 

    Conversely, consider two valid state assignments differing in exactly one state. 
    If this is $A_i$ or $B_i$, the corresponding interface has its two edges to the adjacent layers exchanged.
    If it is $C_0$, the two edges joining $C_0,C'_0$ to the interfaces are exchanged.
    If it is $D_i$, the three states in \cref{fig:2opt_state_D} specify the three possible pairings of four vertices, so changing the state exchanges exactly two edges.
    
    Finally, consider $E_i$.
    Each state specifies three edges, as shown in \cref{fig:circle_states}.
    The edge $E_iC_i$ occurs precisely in states
    $1,2,3$, and $E_iD'_i$ precisely in states
    $0,2,4$. Moreover, the edge from $E_i$ to
    the $B_i$-interface occurs in states $0,1$,
    and the edge from $E_i$ to the $Y_i$-interface
    occurs in states $3,4$.
    These are the only edges shared by different
    states. Thus any two distinct states share
    exactly one of their three edges, except for
    the pairs $\{0,3\}$ and $\{1,4\}$, which
    share none.
    
    Consequently, changing one state exchanges
    exactly two edges, except in the two stated
    cases, where it exchanges three edges.
    Since both state assignments determine tours,
    every exchange of exactly two edges is a
    valid $2$-swap.
\end{proof}

We say a change in a state is \defi{improving} if it corresponds to an improving 2-swap.
For a state $T$, we write $T:r_0\to r_1\to\cdots\to r_m$ if, with all other states fixed, its contribution
to the objective is strictly decreasing in the displayed order.
Thus a change from any earlier state to any later state decreases the objective and hence is improving, as long as the resulting state assignment is valid and the change corresponds to a $2$-swap.
Subject to the validity conditions, the following lemma essentially states that:
\begin{itemize}
    \item the direction of an improving change of $A_i$ for $i < n-1$ depends only on the upper gadget of $F_{i+1}$;
    \item that of $B_i$ for $i \geq 1$ depends only on the lower gadget of $F_i$;
    \item that of $B_0$ depends only on $C_0$;
    \item that of $C_0$ depends only on $A_0$;
    \item that of $D_i$ depends only on $A_i$;
    \item that of $E_i$ depends on $B_{i-1}$ and $B_i$.
\end{itemize}

\begin{lemma}\label{lem:flip_rules}
Suppose a state assignment describes a tour of $G_n$.
Then all the improving changes of all states are exactly as stated in \cref{tab:2opt_flip_rule}, subject to the validity conditions in \cref{lem:2opt_valid_states}.

\begin{table}[ht]
    \centering
    \begin{tabular}{c l c}
        Item & Improving changes & Conditions \\ 
        \hline
        (a1) & $A_{n-1}:0\to1$ & none \\ 
        (a2) & $A_i: 0 \to 1$, $i < n-1$ &  $(A_{i+1},D_{i+1})\in\{(0,2),(1,0)\}$ \\ 
        (a3) & $A_i:1\to0$, $i < n-1$ & 
      $(A_{i+1},D_{i+1}) \in\{(0,0),(0,1),(1,1)\}$ \\
      \hline
      (b1) & $B_0: 0 \to 1$ & $C_0=1$ \\ 
      (b2) & $B_0: 1 \to 0$ & $C_0=0$ \\
      (b3) & $B_i:0\to1$, $i \geq 1$ & $E_i=2$ \\ 
      (b4) & $B_i:1\to0$, $i \geq 1$ & $E_i\in\{0,1\}$ \\
        \hline
        (c1) & $C_0: 0 \to 1$ & $A_0 = 1$ \\
        (c2) & $C_0: 1 \to 0$ & $A_0 = 0$ \\
      \hline
      (d1) & $D_i: 1 \to 2 \to 0$ & $A_i = 0$ \\
      (d2) & $D_i: 0 \to 1$ & $A_i = 1$ \\
      \hline
      (e1) & $E_i: 1 \to 2 \to 0$ & $(B_{i-1},B_i) = (0,0)$ \\
      (e2) & $E_i: 1 \to 3 \to 4 \to 0 \to 2$ & $(B_{i-1},B_i) = (0,1)$ \\
      (e3) & $E_i: 0 \to 1$ & $(B_{i-1},B_i) = (1,0)$ \\
      (e4) & $E_i: 0 \to 1 \to 2 \to 3 \to 4$ & $(B_{i-1},B_i) = (1,1)$
    \end{tabular}
    \caption{List of all improving changes together with the conditions for these changes (in addition to validity conditions in \cref{lem:2opt_valid_states}).}
    \label{tab:2opt_flip_rule}
\end{table}

\end{lemma}

\begin{proof}
Note that the contribution to the objective by $F_0$ is determined precisely by the states of $A_0, B_0, C_0$.
Likewise, for $i\in [n-1]$, the contribution by the upper (lower) gadget of $F_i$ is determined by the states of $A_{i-1}, A_i, D_i$ (respectively, $B_{i-1}, B_i, E_i$).
From the weight assignment of the edges, we can compute these contributions precisely in \cref{tab:2opt_local_weights}.

\begin{table}[h!]
\centering
\small

\renewcommand{\arraystretch}{1.15}

\begin{tabular}{@{}c | cc@{}}

$(A_0,B_0)$ & $C_0=0$ & $C_0=1$ \\
\hline
$(0,0)$ & $0$   & $a+3$ \\
$(0,1)$ & $a-1$ & $a$   \\
$(1,0)$ & $4$   & $3$   \\
$(1,1)$ & $a+3$ & $0$   \\
\end{tabular}
\qquad
\begin{tabular}{@{}c | ccc@{}}
$(A_{i-1},A_i)$ & $D_i=0$ & $D_i=1$ & $D_i=2$ \\
\hline
$(0,0)$ & $1$       & $2b_i-2$ & $2b_i-3$ \\
$(0,1)$ & $2b_i-3$  & $1$      & --       \\
$(1,0)$ & $b_i-1$   & $3b_i-4$ & $b_i$    \\
$(1,1)$ & $b_i$     & $b_i-1$  & --       \\
\end{tabular}

\medskip

\begin{tabular}{@{}c | ccccc@{}}
$(B_{i-1},B_i)$ & $E_i=0$ & $E_i=1$ & $E_i=2$
                & $E_i=3$ & $E_i=4$ \\
\hline
$(0,0)$ & $0$     & $5$     & $4$ & --      & --    \\
$(0,1)$ & $c_i-1$ & $c_i+4$ & $1$ & $c_i+1$ & $c_i$ \\
$(1,0)$ & $2$     & $1$     & --  & --      & --    \\
$(1,1)$ & $c_i+1$ & $c_i$   & $2$ & $1$     & $0$   \\
\end{tabular}

\caption{Contributions to the tour weight from $F_0$
and the upper and lower gadgets of $F_i$, respectively,
where $i\in[n-1]$, $b_i=a2^i$, and
$c_i=2^{n+1-i}$. A dash indicates that the
corresponding edges do not exist.}
\label{tab:2opt_local_weights}
\end{table}

In the following, whenever we consider a state change, the contribution of a gadget to the gain is defined as its contribution to the tour weight before the change minus its contribution after the change. 
These gain contributions are obtained by subtracting the corresponding entries of \cref{tab:2opt_local_weights}.

Items (c1)--(e4) now follow immediately from \cref{tab:2opt_local_weights}.
(a1) follows immediately from the fact that the weight difference of $a2^n$ between $C_nA'_{n-1}$ and $C_nA_{n-1}$ is larger than any weight in the construction.

For (b1)--(b2), consider changing $B_0$ from $0$ to $1$.
If $n=1$, the claim follows by subtracting the corresponding entries of \cref{tab:2opt_local_weights}. 
Suppose $n\geq 2$.
The contribution from $F_0$ to the gain is $1-a$ if $C_0=0$, and $3$ if $C_0=1$.
The contribution from the lower gadget of $F_1$ to the gain is one of $-2,4,-1,c_1$, whenever the change is valid, and thus lies between $-2$ and $c_1$.
Since $a=4c_1+3$, the total gain is at most $(1-a)+c_1=-3c_1-2<0$ when $C_0=0$, and at least $3-2=1>0$ when $C_0=1$. Hence $B_0:0\to 1$ is improving exactly when $C_0=1$, proving (b1). 
The reverse change gives (b2).

For (a2)--(a3), consider first changing $A_i$ from $0$ to $1$, where $1\le i<n-1$. 
By \cref{tab:2opt_local_weights}, the contribution from $F_i$ to the gain lies between $4-2b_i$ and $2b_i-3$. 
On the other hand, the contribution from $F_{i+1}$ to the gain is $b_{i+1}-3$ if $(A_{i+1},D_{i+1})\in\{(0,2),(1,0)\}$, and $2-b_{i+1}$ otherwise.
Since $b_{i+1}=2b_i$, the total gain is at least $1$ in the former case and at most $-1$ in the latter.
For $i=0$, the contribution from $F_0$ is either $-4$ or $a$. 
Since $b_1=2a$, the same conclusion holds.
The reverse change has the opposite gain.

For (b3)--(b4), consider changing $B_i$ from
$0$ to $1$.
For $1\le i\le n-2$, the contribution from
$F_i$ to the gain is $1-c_i$ if $E_i\in\{0,1\}$,
and $3$ if $E_i=2$, whenever both entries exist.
The contribution from $F_{i+1}$ to the gain is again
one of $-2,4,-1,c_{i+1}$.
Since $c_i=2c_{i+1}\ge8$, the total gain
is negative in the case of $E_i \in \{0,1\}$ and positive
in the case of $E_i = 2$.
For $i=n-1$, the contribution from $F_i$ is as above with $c_{n-1}=4$ and the contribution from $F_n$
to the gain is zero, giving the same
conclusion.
The reverse change has the opposite gain.
\end{proof}

\subsection{Correctness}
We now show that in the construction above, any maximal improving sequence from the initial tour has exponential length.
In particular, we first describe a canonical improving swap sequence.
Then we show that every other maximal improving sequence differs from it only by reordering of some swaps.

\paragraph{The canonical sequence.}
In order to describe the canonical sequence, we first define its two building blocks, the \defi{forward traversal} $\mathbf{P}_i$ and the \defi{return traversal} $\mathbf{Q}_i$ for $i \in [n]$.
A traversal changes the lower layers so that an improving move in the next layer becomes possible.
For convenience, we label a swap by the state it changes: $A_i,B_i,C_0$ denote flips of the corresponding binary states, whereas $D_i^j$ and $E_i^j$ denote changes to state $j$ from state $j-1$ (cyclically).

Let $\mathbf{P_1} = \mathbf{Q_1} = A_0C_0B_0$.
Define $\mathbf{P}_2$ and $\mathbf{Q}_2$ as follows:
\begin{align}
\mathbf{P_2}={}&A_1\,\mathbf{P_1}\,E_1^1\,D_1^1\,
     \mathbf{Q_1}\,E_1^2\,B_1,\label{eq:P-base-state}\\
\mathbf{Q_2}={}&\mathbf{P_1}\,E_1^3\,D_1^0\,
     E_1^4\,\mathbf{Q_1}\,E_1^0\,B_1.\label{eq:Q-base-state}
\end{align}

For $3\leq k\leq n$, define $\mathbf{P_k}$, and for $3\leq k\leq n-1$ also define $\mathbf{Q_k}$, by
\begin{align}
\mathbf{P_k}={}&A_{k-1}\,\mathbf{P_{k-1}}\,
 E_{k-1}^1D_{k-1}^1A_{k-2}D_{k-2}^2\,
 \mathbf{Q_{k-1}}\,E_{k-1}^2B_{k-1},\label{eq:P-recursion}\\
\mathbf{Q_k}={}&\mathbf{P_{k-1}}\,
 E_{k-1}^3D_{k-1}^0A_{k-2}D_{k-2}^2E_{k-1}^4\,
 \mathbf{Q_{k-1}}\,E_{k-1}^0B_{k-1}.\label{eq:Q-recursion}
\end{align}

Then the \emph{canonical sequence} for $G_n$ is exactly $\mathbf{P_n}$.
We now show that this canonical sequence is indeed a valid improving swap sequence for $G_n$.

\begin{lemma}\label{lem:reachable_prefix}
Let $\widehat{\mathbf P}_n$ be obtained from $\mathbf P_n$ by deleting all occurrences of $E_i^4$.
Then $\mathbf P_n$ is improving, and in every improving swap sequence from $\tau_0$, the changes other than $E_i^4$ form a prefix of $\widehat{\mathbf P}_n$.
Further, at each reachable tour, the only improving changes are the next unperformed change of $\widehat{\mathbf P}_n$, if it is applicable, and occurrences of $E_i^4$ whose corresponding $D_i^0$ has occurred but which have not yet been performed.
The next unperformed change is inapplicable only when it is $E_i^0$ and the corresponding $E_i^4$ has not been performed.
\end{lemma}
\begin{proof}
For $n=1$, \cref{lem:flip_rules} shows that the only successive
improving changes are $A_0,C_0,B_0$, after which none is improving.
Hence suppose $n\ge2$.
For $k<n$, let $\widehat{\mathbf P}_k$ and
$\widehat{\mathbf Q}_k$ be obtained from $\mathbf P_k$ and
$\mathbf Q_k$ by deleting all occurrences of $E_j^4$.
We call a change \defi{ordinary} if it is not $E^4_j$ for any $j$.
We call a change \defi{lower than $k$} if it changes $C_0$ or a state $A_j,B_j,D_j,E_j$ with $j<k$, whenever defined.

We prove the following statements simultaneously by induction on $k$.
A \defi{rearrangement} of $\mathbf P_k$ or $\mathbf Q_k$ is obtained by moving each occurrence of $E_j^4$ anywhere between its corresponding $D_j^0$ and $E_j^0$, without changing the order of the ordinary changes.
Hence, deleting the changes $E_j^4$ from any prefix of a rearrangement leaves a prefix of $\widehat{\mathbf P}_k$ or $\widehat{\mathbf Q}_k$, respectively.

\begin{itemize}
\item[($J_k$)] Suppose all states in $F_0,\ldots,F_{k-1}$ are $0$.
If $k<n$, suppose also that $(A_k,B_k,D_k,E_k)$ is $(1,0,0,0)$ or $(0,1,2,2)$.
Then every rearrangement of $\mathbf P_k$ is improving.
After any proper prefix of a rearrangement, the only improving ordinary change lower than $k$ is the next unperformed change of $\widehat{\mathbf P}_k$, unless it is $E_j^0$ and $E_j^4$ is still pending, in which case there is none.
An occurrence of $E_j^4$ lower than $k$ is improving precisely when its $D_j^0$ has occurred and it has not yet been performed.
After the entire rearrangement, no change lower than $k$ is improving.
The final states are $(A_0,B_0,C_0)=(1,1,1)$ if $k=1$, and otherwise all states below $F_{k-1}$ are $0$ and $(A_{k-1},B_{k-1},D_{k-1},E_{k-1})=(1,1,1,2)$.

\item[($K_k$)] Suppose $k<n$ and $(A_k,B_k,D_k,E_k)$ is $(1,0,1,1)$, $(0,1,0,3)$, or $(0,1,0,4)$.
Suppose further that $(A_0,B_0,C_0)=(1,1,1)$ if $k=1$, or, if $k\ge2$, that all states below $F_{k-1}$ are $0$ and $(A_{k-1},B_{k-1},D_{k-1},E_{k-1})=(0,1,2,2)$.
Then every rearrangement of $\mathbf{Q}_k$ is improving, and the same assertions about improving changes in $(J_k)$ hold for $\mathbf Q_k$ and $\widehat{\mathbf Q}_k$ in place of $\mathbf P_k$ and $\widehat{\mathbf P}_k$.
After the entire rearrangement, no change lower than $k$ is improving, and all states in $F_0,\ldots,F_{k-1}$ return to $0$.
\end{itemize}

Note that in the case $(A_k,B_k,D_k,E_k)=(0,1,0,3)$ of $(K_k)$, $E_k^4$ may still be pending. 
It is a change outside the lower $k$ layers and can be performed during the traversal or afterwards.
Throughout this proof, whenever we claim whether a change is improving, we use \cref{lem:flip_rules} and exclude changes forbidden by \cref{lem:2opt_valid_states}. 

For $k=1$, both traversals are $A_0C_0B_0$.
The successive states of $(A_0,B_0,C_0)$ are
\[
(0,0,0)\to(1,0,0)\to(1,0,1)\to(1,1,1)
\]
in $(J_1)$, and occur in reverse order in $(K_1)$.
In each case, \cref{lem:flip_rules} shows that only the next change is improving among the changes lower than $1$. 
This proves the base case.

Let $k\ge2$, let $i=k-1$, and suppose the statements hold
for smaller indices.
We first consider how a recursive traversal affects the gadget $F_i$.
Note that $F_i$ is only connected to lower layers via the $A_{i-1}$- and $B_{i-1}$-interfaces.
By \cref{lem:flip_rules}(a1)--(b4), among the changes lower than $k$, none affects whether a change in $A_{i}$ is improving, and only $B_{i-1}$ affects $B_i$ and $E_i$.
While $\mathbf P_i$ is executed, $B_{i-1}$ remains $0$ until its final change to $1$; while $\mathbf Q_i$ is executed, $B_{i-1}$ remains $1$ until its final change to $0$.
These are also the final changes of any rearrangement of the respective
traversals, since by the inductive hypothesis, any $E_j^4$ must precede its corresponding $E_j^0$, which occurs before the final $B_{i-1}$-change.

\begin{table}[ht]
\centering
\begin{tabular}{@{}clll@{}}
Traversal & $(A_i,B_i,D_i,E_i)$ & Final $B_{i-1}$ & Effect on improving changes of $E_i$ \\
\hline
$\mathbf P_i$ & $(1,0,0,0)$ & $0\to1$ & enables $E_i:0\to1$ \\
$\mathbf P_i$ & $(0,1,2,2)$ & $0\to1$ & enables $E_i:2\to3$ \\
$\mathbf Q_i$ & $(1,0,1,1)$ & $1\to0$ & enables $E_i:1\to2$ \\
$\mathbf Q_i$ & $(0,1,0,4)$ & $1\to0$ & enables $E_i:4\to0$ \\
$\mathbf Q_i$ & $(0,1,0,3)$ & $1\to0$ & $E_i:3\to4$ remains improving \\
\end{tabular}
\caption{Effect of the final $B_{i-1}$-change of a recursive traversal.}
\label{tab:2opt_afterB}
\end{table}

Indeed, in the first four rows the indicated change has gain $1$ after the final change of $B_{i-1}$, but is not improving before it.
In the last row, $E_i:3\to4$ has gain $1$ both before and after $B_{i-1}$ changes; once it has occurred, $E_i:4\to0$ has gain $1$ when $B_{i-1}=0$.
A direct comparison of all states of $D_i,E_i$ in \cref{tab:2opt_local_weights}, subject to \cref{lem:2opt_valid_states}, shows that there is no other improving change of these states during the recursive traversal.
For example, $E_i:2\to4$ is forbidden when $D_i=2$, $E_i:1\to0$ is forbidden when $D_i=1$, and $E_i:3\to0$ is not a $2$-swap.
The flip rules also exclude a newly improving change of $A_i$ or $B_i$ while the recursive traversal is running.

We now prove $(J_k)$.
Initially, $A_i:0\to1$ is improving: for $k=n$ this follows from \cref{lem:flip_rules}(a1), and for $k<n$ from (a2) and the specified state of $F_k$.
Further, by \cref{lem:flip_rules}, since all lower layers are 0, no other change is improving.

After this change, by \cref{lem:flip_rules}, the only potential improving change in $F_i$ is $D_i: 0 \to 1$, since $A_i=1$ and $D_i = 0$.
However, by \cref{lem:2opt_valid_states}, as long as $E_i = 0$, this change is not possible.
For the layers lower than $i$, since $(A_i,B_i,D_i,E_i)=(1,0,0,0)$, so $(J_i)$ applies and we traverse a rearrangement of $\mathbf P_i$.
The first row of \cref{tab:2opt_afterB} shows that after this traversal, $E^1_i$ is the only improving change.

After this change $E^1_i$, as noted above the improving change $D_i: 0 \to 1$ is now possible.
Since $B_i = 0$, there is no improving change for it, and by \cref{lem:flip_rules}, there is no improving change newly arising from this change of $E_i$.
Hence, $D_i: 0 \to 1$ is the only improving change.

After this change $D^1_i$, by \cref{lem:flip_rules}, the only state that may have an improving change is $A_{i-1}$.
By $(J_i)$, after the previous traversal $\mathbf{P}_i$, $(A_{i-1}, B_{i-1}, D_{i-1}, E_{i-1})$ is now $(1,1,1,2)$. 
Combined with the fact that $(A_i, D_i)$ is now $(1,1)$, $A_{i-1}: 1 \to 0$ is improving.

 If $k = 2$ (i.e., $i = 1$), then by \cref{lem:flip_rules}(c2) and (b2), the next two improving changes have to be $C_0: 1 \to 0$ and $B_0: 1 \to 0$.
 That is, we have just performed exactly $\mathbf{Q}_1$.

If $k > 2$, then after the change $A_{i-1}: 1 \to 0$, since $D_{i-1} = 1$, and $A_{i-2} = 0$, by \cref{lem:flip_rules}, only state $D_{i-1}$ can potentially have an improving change.
As $E_{i-1} = 2$, by \cref{lem:2opt_valid_states}, $D_{i-1}$ cannot be 0.
Hence, the only improving change is $D_{i-1} : 1 \to 2$.

This change of $D_{i-1}$, by \cref{lem:flip_rules}, does not affect any vertex in $F_i$.
However, we now have $(A_{i}, B_{i}, D_{i}, E_{i}) = (1,0,1,1)$, $(A_{i-1},B_{i-1},D_{i-1},E_{i-1}) = (0,1,2,2)$, and all states below $F_{i-1}$ are 0.
Hence, $(K_i)$ applies and we perform a rearrangement of $\mathbf{Q}_i$.

In both cases (i.e., whether $k = 2$ or $k > 2$), we can use the third row of \cref{tab:2opt_afterB} and obtain that the only improving change after this rearrangement of $\mathbf{Q}_i$ is $E^2_i$. 
By \cref{lem:flip_rules}, this change of $E_i$ makes the change $B_i: 0 \to 1$ improving, and no other state depends on $E_i$.
Hence, this is the only improving change.
We have now exactly all the changes of $\mathbf P_k$.
Note that $(A_{i}, B_{i}, D_{i}, E_{i}) = (1,1,1,2)$, while other states below $F_i$ are 0, as claimed.
Also, checking against \cref{lem:flip_rules} verifies that no change below $k$ is improving and we complete the inductive proof of $(J_k)$.

We use a similar analysis for $(K_k)$, which we briefly go through as follows.
We start with $(A_i,B_i,D_i,E_i)=(0,1,2,2)$.
By $(J_i)$ the improving ordinary changes first follow a rearrangement of $\mathbf P_i$.
The second row of \cref{tab:2opt_afterB} then gives the next change $E_i^3$, which then enables $D_i^0$.
At this point $A_{i-1}:1\to0$ and $E_i^4$ are both improving, and no other change is improving.
For $i\ge2$, changing $A_{i-1}$ enables $D_{i-1}:1\to2$, after which $(K_i)$ applies and we can perform a rearrangement of $\mathbf{Q}_i$.
For $i=1$, $A_0$ is the first change of $\mathbf Q_1$, and the rest of the changes of $\mathbf Q_1$ then follow.
After that, by the last row of \cref{tab:2opt_afterB} and by the argument for the case when $(A_i, B_i, D_i, E_i) = (0,1,0,3)$ before we apply $(K_i)$, either $E^0_i$ is the only improving change if $E^4_i$ has been performed by the end of the rearrangement of $\mathbf{Q}_i$; or there is no new change that becomes improving, and we have to perform $E_i^4$, after which, $E_i^0$ becomes improving.

We claim that we may perform $E_i^4$ at any time before $E_i^0$.
Indeed, after the change $D_i^0$ that enables $E_i^4$ and until $E_i^0$, we have $(A_i,B_i,D_i)=(0,1,0)$.
Since $D_i$ remains 0, by \cref{lem:2opt_valid_states}, it does not forbid the change $E_i^4$.
Further, by \cref{lem:flip_rules}, regardless of whether $B_{i-1}$ is 0 or 1, $E_i^4$ remains improving, and it does not depend on any other state.

Moreover, we cannot omit this change $E_i^4$, as by \cref{lem:flip_rules}, since $D_i$ remains 0, $E_i$ cannot take value 2, and by \cref{lem:2opt_valid_states}, the change of $E_i$ from 3 to 0 does not correspond to a 2-swap.
Hence, $E_i^4$ is the only improving change of $E_i$ and remains so until it is performed.

As argued above, after both $E_i^4$ and a rearrangement of $\mathbf{Q}_i$ are performed, the only improving change is $E_i^0$.
After that, we perform the only subsequent improving change $B_i:1\to0$, and all states lower than $k$ return to $0$.
Again, checking against \cref{lem:flip_rules} verifies that no change below $k$ is improving, and we complete the induction for $(K_k)$.

The induction also shows that after every prefix of a rearrangement, all improving changes lower than $k$ are exactly those stated in $(J_k)$ and $(K_k)$: the next unperformed ordinary change, unless blocked by a pending $E_j^4$, and the available changes $E_j^4$.
In particular, $(J_n)$ shows that $\mathbf P_n$ is improving.
Finally, starting from $\tau_0$, induction on the number of performed changes shows that every improving sequence is a prefix of a rearrangement of $\mathbf P_n$: every improving choice is either the next unperformed ordinary change or a currently available $E_j^4$, and either choice preserves the required order.
This proves the lemma.
\end{proof}

\begin{corollary}\label{cor:2opt_maximal}
The canonical sequence $\mathbf P_n$ is maximal, and every maximal improving swap sequence from $\tau_0$ is obtained from it by moving occurrences of $E_i^4$ between their corresponding $D_i^0$ and $E_i^0$. 
All these sequences have length $3$ if $n=1$, and $9\cdot2^{n-1}-7$ if $n\ge2$.
\end{corollary}
\begin{proof}
The first statement follows immediately from \cref{lem:reachable_prefix}. 
By the recursion of $\mathbf P_n$, we obtain that $|\mathbf{P}_k| = |\mathbf{Q}_k|$ for all applicable $k$, 
$|\mathbf P_1|=3$, $|\mathbf P_2|=11$, and for $k\ge3$, $|\mathbf{P}_k| = 2|\mathbf{P}_{k-1}|+7$. 
Hence, $|\mathbf{P}_n|=9\cdot2^{n-1}-7$ for $n\ge2$.
Since all maximal improving swap sequences have the same length as $\mathbf{P}_n$, the corollary then follows.
\end{proof}

We are now ready to prove the main theorem of this section.

\begin{proof}[Proof of \cref{thm:2opt}]
    Observe that $G_n$ has $O(n)$ vertices and each weight can be encoded by $O(n)$ bits.
    Combining this with \cref{cor:2opt_maximal}, we obtain the all-exp property for TSP with the 2-\opt neighborhood.
\end{proof}


\section{(2,3,2)-MCA has the all-exp property}
\label{sec:232}
This section is dedicated to showing the following.

\thmtwothreetwo*

Our construction is based on the construction to prove the all-exp property for Max-Cut with the flip neighborhood by Michel and Scott~\cite{MS2023}.
We create one binary variable for each vertex in the Michel-Scott construction and one two-variable constraint for each edge, choosing the constraint weights so that the two local-search problems behave identically at every vertex of degree at most three.
The only obstacle is that the Michel-Scott construction has vertices of degree four, namely $v_{i,1}$ and $v_{i,8}$, which would occur in four constraints and violate the bound $q=3$.
We replace each of them by a small gadget of degree-three variables.
The gadget for $v_{i,1}$ consists of two XOR gates; the gadget for $v_{i,8}$ is a chain of copies that broadcasts the value of $v_{i,8}$ to the next layer and prevents one update from overtaking the preceding one.

All gadgets use a common \emph{following} primitive, formalized in Section~\ref{subsec:following}.
A vertex $x$ \emph{follows} a Boolean function $\phi$ determined by its neighbors if, for every assignment $\sigma$ of other variables, $x$ has a strictly preferred value of $\phi(\sigma)$.
Hence, when $\sigma(x)\neq \phi(\sigma)$, flipping $x$ improves the objective.

For brevity, we do not separately recap the Michel-Scott construction. 
Instead, we directly describe our modified construction and establish the local properties needed for the proof. 
The relation between our construction and the original Michel-Scott construction is discussed in Remark~\ref{rem:MS-comparison}.

\subsection{Graph representation of binary MCA}
We use the standard constraint-graph representation of Max 2-CSP~\cite{MR2571908}, in which variables are represented by vertices, constraints by edges, and values by unary scores on vertices and binary scores on edges.
For our construction, it is convenient to use the following weighted graph representation of binary MCA.
Specifically, we represent a $(2,q,2)$-MCA instance by a graph $G = (V,E)$ of maximum degree at most $q$, edge weights $w\colon E\to\mathbb Z$, and vertex weights $b\colon V\to\mathbb Z$, with objective

\[
\Phi(\sigma) = \sum_{uv\in E} w(uv)\,[\sigma(u)\ne\sigma(v)] + \sum_{v\in V} b_v\,\sigma(v) .
\]

For a vertex $x$ we use the signed difference
\[
\Delta_x(\sigma)\ :=\ \Phi(\sigma[x\mapsto 1])-\Phi(\sigma[x\mapsto 0])
 = b_x+\sum_{ux \in E} w(xu)\,(1-2\sigma(u)).
\]

Thus flipping $x$ from $0$ to $1$ is improving iff $\Delta_x(\sigma)>0$, and flipping it from $1$ to $0$ is improving iff $\Delta_x(\sigma)<0$.

The following standard representation allows us to specify our MCA instance entirely by vertex and edge weights.

\begin{lemma}
\label{lem:mca-form}
Let $q\ge1$.
If $G$ has maximum degree at most $q$, then any objective $\Phi$ of the above form defines a $(2,q,2)$-MCA instance with the same objective function.
\end{lemma}

\begin{proof}
Associate with every edge $uv$ the binary constraint $C_{uv}(\sigma(u),\sigma(v)) = w(uv)[\sigma(u)\ne\sigma(v)].$
The unary term $b_x\sigma(x)$ can be incorporated into any binary constraint incident with $x$ by adding $b_x$ to both entries corresponding to $\sigma(x)=1$. If $x$ is isolated, use the unary constraint $C_x(0)=0$ and $C_x(1)=b_x$. Hence every variable occurs in at most $\deg_G(x)\le q$ binary constraints, or in one unary constraint if it is isolated.
\end{proof}

From this point on, we identify each binary MCA variable with the corresponding vertex in the graph representation and use the terms variable and vertex interchangeably.
By Lemma~\ref{lem:mca-form}, it therefore suffices to construct a graph of maximum degree three with suitable edge and vertex weights.

\subsection{The ``following'' primitive}
\label{subsec:following}
For a vertex $x$, a Boolean function $\phi$ of the values of some neighbors of $x$, and a set $\mathcal S$ of assignments, we say $x$ \emph{follows} $\phi$ on $\mathcal S$, denoted by $x\leftarrow\phi$ on $\mathcal S$, if for every assignment $\sigma \in \mathcal{S}$ we have $\Delta_x(\sigma)>0$ when $\phi(\sigma)=1$ and $\Delta_x(\sigma)<0$ when $\phi(\sigma)=0$. 
Equivalently, $x$ has a unique locally preferred value among the assignments in $\mathcal S$, namely $\phi(\sigma)$, and if $\sigma(x)\ne\phi(\sigma)$ then flipping $x$ is improving. 
If $\mathcal S$ is the set of all assignments, we simply say $x$ follows $\phi$.
The most important case is that $\phi$ is a \emph{literal} $y$ or $\bar y$ of a single neighbor $y$.

For a neighbor $y$ of $x$, define the \emph{residual incident weight}
\[
B_x(y)\ :=\ \sum_{ux \in E,\,u\ne y}|w(xu)|.
\]

The contribution to $\Delta_x$ from all neighbors other than $y$ lies in $[-B_x(y),B_x(y)]$.
Note that when the neighbors' values can be chosen arbitrarily, both endpoints of this range can be attained.

The following two lemmas list several statements that will be useful in our construction.

\begin{lemma}[Following a literal]
\label{lem:following}
    Let $x$ have neighbor $y$ and let $B=B_x(y)$. Then
    \begin{enumerate}[label=\textup{(\alph*)}]
        \item $x\leftarrow y$ holds if and only if $b_x-w(xy)-B>0$ and $b_x+w(xy)+B<0$;
        \item $x\leftarrow \bar y$ holds if and only if $b_x+w(xy)-B>0$ and $b_x-w(xy)+B<0$.
    \end{enumerate}
    In particular, if $b_x=0$, then $x\leftarrow y$ holds if and only if $w(xy)<-B$, and $x\leftarrow\bar y$ holds if and only if $w(xy)>B$.
\end{lemma}
\begin{proof}
    We have $\Delta_x = b_x + w(xy)(1-2\sigma(y)) + r$ with $r\in[-B,B]$. 
    For (a) we need $\Delta_x>0$ whenever $\sigma(y)=1$, i.e. $b_x-w(xy)+r>0$ for every admissible $r$, which is equivalent to $b_x-w(xy)-B>0$; and we need $\Delta_x<0$ whenever $\sigma(y)=0$, which is equivalent to $b_x+w(xy)+B<0$. 
    Part (b) is the same with the roles of $\sigma(y)=0,1$ exchanged. The statement for $b_x=0$ follows immediately.
\end{proof}

\begin{lemma}
\label{lem:gates}
    Let $D>0$, let $z$ have neighbors $x,y$ and possibly another neighbor $t$.
    Let $r$ be $w(zt)(1 - 2\sigma({t}))$, if $t$ exists, and $0$ otherwise.
    Suppose that $|r|<D$.
    \begin{enumerate}[label=\textup{(\roman*)}]
        \item \textup{(AND-NOT)} If $w(zx)=-D$, $w(zy)=+D$ and $b_z=-D$, then $\Delta_z=2D(\sigma({x})-\sigma({y}))-D+r$, which is $D+r$ if $(\sigma({x}),\sigma({y}))=(1,0)$ and at most $-D+r$ otherwise; hence $z\leftarrow x\wedge\bar y$.
        \item \textup{(OR)} If $w(zx)=w(zy)=-D$ and $b_z=+D$, then $\Delta_z=2D(\sigma({x})+\sigma({y}))-D+r$; hence $z\leftarrow x\vee y$.
        \item \textup{(tie breaker)} 
        Suppose $t$ exists.
        If $w(zx)=w(zy)=-D$ and $b_z=0$, then
        \[
        \Delta_z=2D(\sigma({x})+\sigma({y})-1)+r.
        \]
        Hence, $z\leftarrow x$ on assignments with $\sigma({x})=\sigma({y})$. 
        On assignments with $\sigma({x})\ne\sigma({y})$, if $w(zt) > 0$, $z\leftarrow\bar t$; if $w(zt)  < 0$, $z\leftarrow t$.
    \end{enumerate}
\end{lemma}
\begin{proof}
    Each formula follows by substituting the weights into the definition of $\Delta_z$.
    The first two items use $|r|<D$.
    For the last item, when $\sigma(x) = \sigma(y)$, since $|r| < D < 2D$, we have $z \leftarrow x$ (or equivalently $z \leftarrow y$).
    When $\sigma(x) \ne \sigma(y)$, we have $\Delta_z = r = w(zt)(1 - 2\sigma(t))$.
    This implies that on these assignments, $z \leftarrow \bar{t}$ when $w(zt) > 0$ and $z \leftarrow t$ when $w(zt) < 0$.
\end{proof}

\subsection{The construction}
\label{subsec:construction}

Fix $n\ge1$. 
Let $\lambda_i:=2^{13i+4}$ for $0\le i\le n+1.$
Thus $\lambda_{i+1}=2^{13}\lambda_i$. 
The instance $\mathcal I_n=(G,w,b)$ has the following vertices and edges; all vertex weights are $0$ unless stated otherwise. 

\paragraph{Paths.}
First, we add two new vertices $v_{0,1}$ and $v_{0,8}$ and an edge between them with weight $-7 \lambda_0$.
Next, for $i \in \{1, \dots, n\}$, we add a path of eight new vertices $v_{i,1}, \dots, v_{i,8}$ in that order, with edge weights $-7\lambda_i,\ -5\lambda_i,\ -5\lambda_i,\ -3\lambda_i,\ -3\lambda_i,\ -\lambda_i,\ -\lambda_i$ from $v_{i,1}$ to $v_{i,8}$.
Last, we add three further vertices $v'_0$, $v'_1$, and $v'_2$ with edges $v_{n,8} v'_0$, $v_{n,1}v'_1$, and $v'_1 v'_2$ with weights $-\lambda_n$, $\lambda_{n+1}$ and $-2\lambda_{n+1}$.
Additionally, $v'_0$ has vertex weight $-2\lambda_n$.

\paragraph{Gadget for $v_{i-1,1}$ ($i\in[n]$).} 
We want the locally preferred value of $v_{i-1,1}$, after the auxiliary variables have settled, to be
$v_{i,2}\oplus v_{i,4}\oplus v_{i,6}$, via the gadget below.

We add nine new vertices $f_1,f_2,f_3,f_4,f_z,g_1,g_2,g_3,g_4$ (all indexed by $i$) and edges as shown in  \cref{fig:v1-gadget}. 
\cref{tab:232_v1_gadget} specifies each new edge and gives the relation enforced at the vertex in the first column; the correctness of the following relations is shown in \cref{lem:local}.

\begin{table}[h!]
\caption{Data for the gadget to connect $v_{i-1,1}$ with $v_{i,2}$, $v_{i,4}$ and $v_{i,6}$. 
The $f$ and $g$ variables are additionally indexed by $i$.
Each edge is listed exactly once, in the row where it is introduced.
}
\small
\label{tab:232_v1_gadget}
\centering
\begin{tabular}{llll}
\hline
Variable & Follows & Edges (Weights) & Vertex weight\\ \hline
$f_1$ & $v_{i,2}$ & $f_1v_{i,2}\ (-\lambda_{i})$ & $0$\\
$f_2$ & $v_{i,4}$ & $f_2v_{i,4}\ (-\lambda_{i})$ & $0$\\
$g_2$ & $v_{i,6}$ & $g_2v_{i,6}\ (-\lambda_{i})$ & $0$\\
$f_3$ & $f_1\wedge\bar f_2$ & $f_3f_1\ (-\lambda_{i}/2^2)$, $f_3f_2\ (+\lambda_{i}/2^2)$ & $-\lambda_{i}/2^2$\\
$f_4$ & $\bar f_1\wedge f_2$ & $f_4f_1\ (+\lambda_{i}/2^2)$, $f_4f_2\ (-\lambda_{i}/2^2)$ & $-\lambda_{i}/2^2$\\
$f_z$ & $f_3\vee f_4$ & $f_zf_3\ (-\lambda_{i}/2^4)$, $f_zf_4\ (-\lambda_{i}/2^4)$ & $+\lambda_{i}/2^4$\\
$g_1$ & $f_z$ & $g_1f_z\ (-\lambda_{i}/2^6)$ & $0$\\
$g_3$ & $g_1\wedge\bar g_2$ & $g_3g_1\ (-\lambda_{i}/2^8)$, $g_3g_2\ (+\lambda_{i}/2^8)$ & $-\lambda_{i}/2^8$\\
$g_4$ & $\bar g_1\wedge g_2$ & $g_4g_1\ (+\lambda_{i}/2^8)$, $g_4g_2\ (-\lambda_{i}/2^8)$ & $-\lambda_{i}/2^8$\\
$v_{i-1,1}$ & $g_3\vee g_4$ & $v_{i-1,1}g_3\ (-\lambda_{i}/2^{10})$, $v_{i-1,1}g_4\ (-\lambda_{i}/2^{10})$ & $+\lambda_{i}/2^{10}$\\
\hline
\end{tabular}
\end{table}

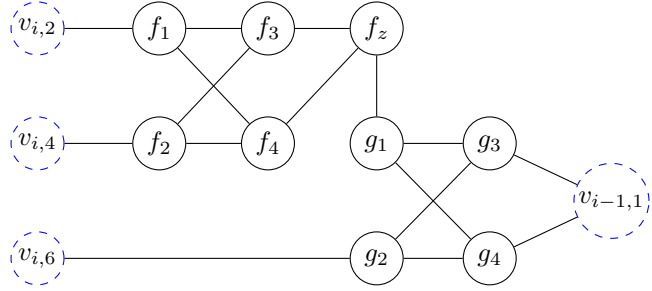
\begin{figure}[ht!]
\centering
\begin{tikzpicture}[x=1.15cm,y=.95cm,>=latex,
 n/.style={circle,draw,minimum size=7mm,inner sep=1pt,font=\small},
 inp/.style={circle,draw=blue,dashed,minimum size=7mm,inner sep=1pt,font=\small},
 ]
\node[inp] (x1) at (0,1.6) {$v_{i,2}$};
\node[inp] (x2) at (0,0) {$v_{i,4}$};
\node[inp] (x3) at (0,-1.6) {$v_{i,6}$};
\node[n] (f1) at (1.4,1.6) {$f_1$};
\node[n] (f2) at (1.4,0) {$f_2$};
\node[n] (fz) at (3.9,1.6) {$f_z$};
\node[n] (f3) at (2.65,1.6) {$f_3$};
\node[n] (f4) at (2.65,0) {$f_4$};
\node[n] (g1) at (3.9,0) {$g_1$};
\node[n] (g2) at (3.9,-1.6) {$g_2$};
\node[n] (g3) at (5.2,0) {$g_3$};
\node[n] (g4) at (5.2,-1.6) {$g_4$};
\node[inp] (next) at (6.6,-0.8) {$v_{i-1,1}$};
\draw (x1)--(f1); \draw (x2)--(f2); \draw (x3)--(g2);
\foreach \u in {f1,f2,fz}{\draw (\u)--(f3);\draw (\u)--(f4);}
\draw (fz)--(g1);
\foreach \u in {g1,g2,next}{\draw (\u)--(g3);\draw (\u)--(g4);}
\end{tikzpicture}
\caption{Gadget for $v_{i-1,1}$. The $f$ and $g$ variables are additionally indexed by $i$. Once its auxiliary variables are locally stable, $v_{i-1,1}$ has preferred value $v_{i,2}\oplus v_{i,4}\oplus v_{i,6}$.}
\label{fig:v1-gadget}
\end{figure}

\paragraph{Gadget for $v_{i-1,8}$ ($i\in [n]$).} 
We add five new vertices $u_0,u_1,u_2,u_3,q$ (indexed by $i$) and edges as shown in \cref{fig:v8-gadget}. 

\begin{table}
\centering
\caption{Data for the gadget to connect $v_{i-1,8}$ with $v_{i,3}$, $v_{i,5}$ and $v_{i,7}$.
Each edge is listed exactly once, in the row where it is introduced.
All vertex weights are zero.
}
\label{tab:232_v8_gadget}
\small
\begin{tabular}{lll}
\hline
Variable & Follows & Edges (Weights)\\
\hline
$u_0$ & common value of $v_{i-1,8},u_1$; otherwise $\bar q$ & $u_0v_{i-1,8}\ (-\lambda_{i-1})$, $u_0u_1\ (-\lambda_{i-1})$, $u_0q\ (+1)$\\
$u_1$ & $u_0$ & $u_1u_2\ (-8)$, $u_1v_{i,7}\ (-1)$\\
$u_2$ & $u_1$ & $u_2u_3\ (-6)$, $u_2v_{i,5}\ (+1)$\\
$u_3$ & $u_2$ & $u_3q\ (-4)$, $u_3v_{i,3}\ (-1)$\\
$q$   & $u_3$ & $qv_{i,8}\ (-1)$\\
\hline
\end{tabular}
\end{table}

\begin{figure}[ht!]
\centering
\begin{tikzpicture}[x=1.25cm,y=1cm,>=latex,
 n/.style={circle,draw,minimum size=7mm,inner sep=1pt,font=\small},
 ext/.style={circle,draw=red,dashed,minimum size=7mm,inner sep=1pt,font=\small},
 ]
\node[ext] (b4) at (1.25,0) {$v_{i-1,8}$};
\node[n] (p) at (2.5,0) {$u_{0}$};
\node[n] (O2) at (3.75,0) {$u_{1}$};
\node[n] (O4) at (5,0) {$u_{2}$};
\node[n] (q) at (6.25,0) {$u_{3}$};
\node[n] (qo) at (7.5,0) {$q$};
\draw (b4)--(p)--(O2)--(O4)--(q)--(qo);
\draw[bend left=35] (p) to node[above,font=\scriptsize] {} (qo);
\node[ext] (a4) at (3.75,-1.5) {$v_{i,7}$};
\node[ext] (a3) at (5,-1.5) {$v_{i,5}$};
\node[ext] (a2) at (6.25,-1.5) {$v_{i,3}$};
\draw (O2)--(a4); \draw (O4)--(a3); \draw (q)--(a2);
\node[ext] (pd) at (7.5,-1.5) {$v_{i,8}$};
\draw (qo)--(pd);
\end{tikzpicture}
\caption{Gadget for $v_{i-1,8}$. The $u$ and $q$ variables are additionally indexed by $i$.}
\label{fig:v8-gadget}
\end{figure}
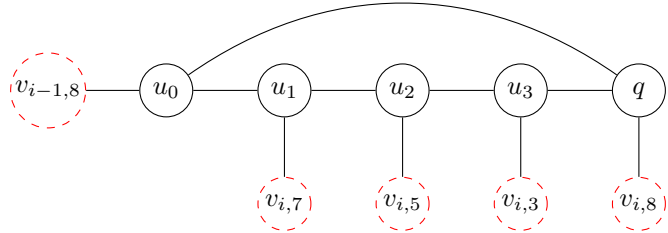

The edges $u_1v_{i,7}$, $u_2v_{i,5}$, $u_3v_{i,3}$ of weights $-1,+1,-1$ correspond to the incident edges of $v_{i-1,8}$ in the Michel-Scott construction, now attached to copies of $v_{i-1,8}$.
The purpose of the other edges is to make sure that the flip of $v_{i-1,8}$ gets propagated correctly to these copies.

\paragraph{Initial assignment $\sigma_0$.} We use the all-zero assignment (i.e., $\sigma_0(x)=0$ for all $x$).

\paragraph{Layer.}
For $0 \le i \le n$, we call the set of all vertices indexed by $i$ the \emph{layer} $F_i$ of the construction.
See \cref{fig:level} for an illustration of $F_i$ together with some following relations between variables.
Note that $F_i$ contains all vertices of the gadget for $v_{i-1,1}$ and the gadget for $v_{i-1,8}$,  except for $v_{i-1,1}$ and $v_{i-1,8}$ themselves.

\begin{figure}[ht!]
\centering
\begin{tikzpicture}[x=0.84cm,y=0.95cm,>=latex,
 ms/.style={circle,draw,thick,minimum size=7mm,inner sep=1pt,font=\small},
 n/.style={circle,draw,minimum size=6.5mm,inner sep=1pt,font=\small},
 ext/.style={circle,draw,dashed,fill=gray!20,minimum size=7mm,inner sep=1pt,font=\small},
 extlow/.style={circle,draw,dashed,minimum size=7mm,inner sep=1pt,font=\scriptsize},
 sync/.style={circle,draw,line width=1.1pt,minimum size=9.2mm,inner sep=0pt},
 follow/.style={-{latex[length=2.2mm,width=1.6mm]}},
 doublefollow/.style={<->},
 tie/.style={-{Latex[open,length=2.2mm,width=1.6mm]}},
 neg/.style={midway,fill=white,inner sep=0.5pt,font=\scriptsize}
]

\node[ms] (v1_new) at (-0.2,3.5) {$v_{i,1}$};
\node[ms] (P2) at (1.5,3.5) {$v_{i,2}$};
\node[ms] (P3) at (3.2,3.5) {$v_{i,3}$};
\node[ms] (P4) at (4.9,3.5) {$v_{i,4}$};
\node[ms] (P5) at (6.6,3.5) {$v_{i,5}$};
\node[ms] (P6) at (8.3,3.5) {$v_{i,6}$};
\node[ms] (P7) at (10.0,3.5) {$v_{i,7}$};
\node[ms] (P8) at (12.0,3.5) {$v_{i,8}$};

\node[sync] (SP3) at (P3.center) {};
\node[sync] (SP5) at (P5.center) {};
\node[sync] (SP7) at (P7.center) {};
\node[sync] (SP8) at (P8.center) {};

\draw[follow]
  (v1_new)--node[neg,above] {}(P2);

\draw[follow]
  (P2)--node[neg,above] {}(SP3);

\draw[doublefollow]
  (SP3)--node[neg,above] {}(P4);

\draw[follow]
  (P4)--node[neg,above] {}(SP5);

\draw[doublefollow]
  (SP5)--node[neg,above] {}(P6);

\draw[follow]
  (P6)--node[neg,above] {}(SP7);

\draw[doublefollow]
  (SP7)--node[neg,above] {}(SP8);

\node[n] (f1) at (1.5,2.0) {$f_1$};
\node[n] (f2) at (4.9,2.0) {$f_2$};
\node[n] (g2) at (8.3,2.0) {$g_2$};

\node[n] (f3) at (2.5,0.7) {$f_3$};
\node[n] (f4) at (3.9,0.7) {$f_4$};
\node[n] (fz) at (3.2,-0.55) {$f_z$};

\node[n] (g1) at (5.0,-0.55) {$g_1$};
\node[n] (g3) at (6.0,0.7) {$g_3$};
\node[n] (g4) at (7.4,0.7) {$g_4$};

\node[ext] (v1) at (6.7,-0.55) {$v_{i-1,1}$};

\draw (P2)--(f1);
\draw (P4)--(f2);
\draw (P6)--(g2);

\draw (f1)--(f3);
\draw (f1)--(f4);
\draw (f2)--(f3);
\draw (f2)--(f4);

\draw (f3)--(fz);
\draw (f4)--(fz);
\draw (fz)--(g1);

\draw (g1)--(g3);
\draw (g1)--(g4);
\draw (g2)--(g3);
\draw (g2)--(g4);

\draw (g3)--(v1);
\draw (g4)--(v1);

\node[ext] (Q8) at (14,6.5) {$v_{i-1,8}$};

\node[n] (u0) at (14.0,5.0) {$u_0$};
\node[n] (u1) at (10.0,5.0) {$u_1$};
\node[n] (u2) at (6.6,5.0) {$u_2$};
\node[n] (u3) at (3.2,5.0) {$u_3$};
\node[n] (q)  at (12.0,6.5) {$q$};

\node[sync] (Su0) at (u0.center) {};

\draw[follow] (Q8)--(Su0);

\draw[doublefollow] (Su0)--(u1);

\draw[follow] (u1)--(u2);
\draw[follow] (u2)--(u3);
\draw[follow] (u3)--(q);

\draw[tie]
  (q)--node[neg,above right] {$\neg$}(Su0);

\draw[tie] (u1) -- (SP7);
\draw[tie] (u2) --node[neg,right] {$\neg$} (SP5);
\draw[tie] (u3) -- (SP3);
\draw[tie] (q)--(SP8);

\node[ext] (Pu0) at (14,3.5) {$u_{i+1,0}$};
\node[ext] (Pg3) at (-1,5) {$g_{i+1,3}$};
\node[ext] (Pg4) at (0.6,5) {$g_{i+1,4}$};

\draw[follow] (Pu0) -- (SP8);
\draw (Pg3) -- (v1_new) -- (Pg4);

\end{tikzpicture}

\caption{A layer $F_i$ of $\mathcal I_n$ for $1 \leq i < n$ with the incident vertices of the adjacent layers (gray circles with dashed borders). 
Directed edges indicate following
relations, and $\neg$ denotes complementation. 
At a ringed vertex, filled arrowheads indicate the two main inputs and the open arrowhead indicates the tie-breaker used when they disagree. 
Unoriented edges record adjacency only; dependencies inside the gadget for $v_{i,1}$ and $v_{i-1,1}$ are omitted.
For $i = n \ge 1$, the vertex $u_{i+1,0}$ is replaced by $v'_0$, and $v_{i,1}$ is incident to $v'_1$ instead of $g_{i+1,3}$ and $g_{i+1,4}$.
}
\label{fig:level}
\end{figure}
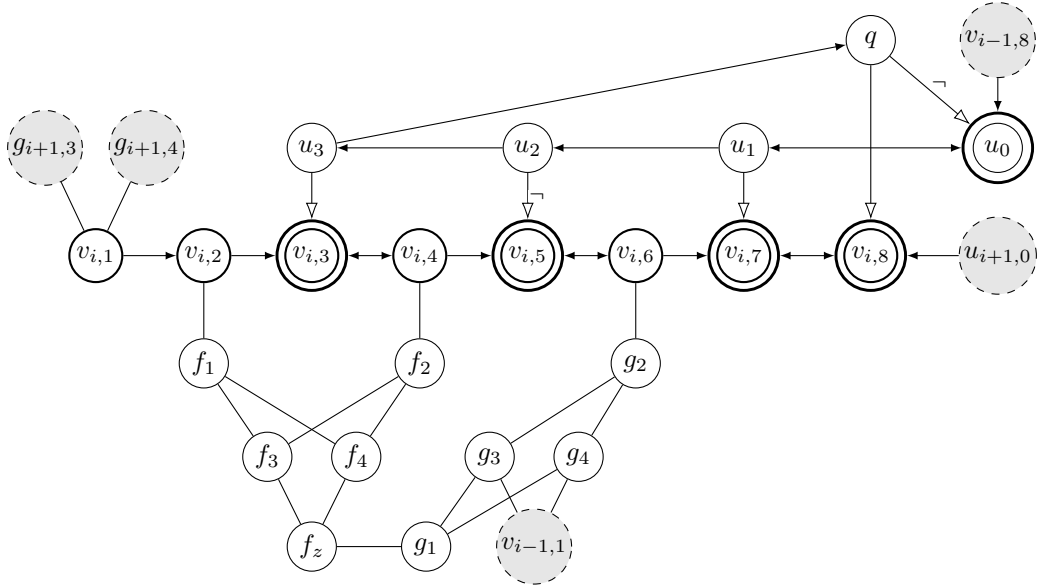

\begin{remark}
\label{rem:MS-comparison}
    Readers familiar with the construction of Michel and Scott~\cite{MS2023} for Max-Cut with the flip neighborhood will recognize much of its structure in ours, particularly the paths and the recursive pattern of flips along them.
    In fact, as we show in the next section, the projection of our canonical flip sequence onto the path vertices is
    \[
    R_0=(v_{0,1},v_{0,8}),\qquad
    R_n=(v_{n,1},v_{n,2},R_{n-1},v_{n,3},v_{n,4},R_{n-1},
    v_{n,5},v_{n,6},R_{n-1},v_{n,7},v_{n,8}),
    \]
    which is exactly the recursive flip sequence in their construction.
    
    There are two changes in the presentation.
    First, we use an all-zero initial assignment: relative to the Michel-Scott construction, the odd-indexed path variables and $v'_1$ are complemented.
    Second, we use larger scales $\lambda_i=2^{13i+4}$ at each layer instead of the scales $\lambda_i=8^{i}$ appearing in the original construction.
    Neither change affects the recursive flip sequence on the path vertices.
    
    The substantive modification is at the two degree-four interfaces.
    We replace the interaction at $v_{i-1,1}$ and $v_{i-1,8}$ by the respective gadgets, so that every variable has degree at most three.
    The additional variable $v'_0$ plays the same role for $v_{n,8}$ as $u_{i+1,0}$ does for $v_{i,8}$ at $i<n$.
    This allows us to treat all layers uniformly in the correctness proof.
\end{remark}

\subsection{Correctness}
The correctness proof has two stages.
First, we verify locally that every auxiliary variable has the intended preferred value; these properties hold for every assignment and are summarized in \cref{lem:local} below.
We then characterize all the maximal improving flip sequences by first defining a ``canonical'' sequence $\Omega_n$ and then showing that every maximal improving flip sequence is simply a rearrangement of some flips in $\Omega_n$.
More formally, we define a dependency graph $D_n$ that encodes a partial order of the flips in $\Omega_n$, and the rearrangement above is then simply a linear extension of this dependency graph.

\begin{lemma}
\label{lem:local}
    In $\mathcal I_n$, the following hold.
    \begin{enumerate}[label=\textup{(\alph*)}]
        \item \textup{(Gadget for $v_{i-1,1}$)} The column ``Follows'' of \cref{tab:232_v1_gadget} is correct. Specifically, $f_1\leftarrow v_{i,2}$, $f_2\leftarrow v_{i,4}$, $g_2\leftarrow v_{i,6}$, $f_3\leftarrow f_1\wedge\bar f_2$, $f_4\leftarrow\bar f_1\wedge f_2$, $f_z\leftarrow f_3\vee f_4$, $g_1\leftarrow f_z$, $g_3\leftarrow g_1\wedge\bar g_2$, $g_4\leftarrow\bar g_1\wedge g_2$, and $v_{i-1,1}\leftarrow{g_3\vee g_4}$.
        Further, this implies that in an assignment $\sigma$ such that none of the variables above can be flipped to improve the objective, we have $\sigma(v_{i-1,1}) = \sigma(v_{i,2})\oplus \sigma(v_{i,4})\oplus \sigma(v_{i,6})$.
        \item \textup{(Gadget for $v_{i-1,8}$)} The column ``Follows'' of \cref{tab:232_v8_gadget} is correct. Specifically, $u_1\leftarrow u_0$, $u_2\leftarrow u_1$, $u_3\leftarrow u_2$, $q\leftarrow u_3$; moreover $u_0\leftarrow v_{i-1,8}$ when $\sigma({v_{i-1,8}})=\sigma({u_1})$, while $u_0\leftarrow\bar q$ when $\sigma({v_{i-1,8}})\ne\sigma({u_1})$.
        \item \textup{(Even path vertices)} For $1\le i\le n$ and $j\in\{2,4,6\}$, $v_{i,j} \leftarrow {v_{i,j-1}}$.
        \item \textup{(Odd path vertices)} For $1\le i\le n$ and $j\in\{3,5,7\}$, 
        when $\sigma({v_{i,j-1}})=\sigma({v_{i,j+1}})$, $v_{i,j}\leftarrow v_{i,j-1}$.
        When $\sigma({v_{i,j-1}})\ne\sigma({v_{i,j+1}})$, $v_{i,3}\leftarrow {u}_{i,3}$, $v_{i,5}\leftarrow \bar{u}_{i,2}$, and $v_{i,7}\leftarrow {u}_{i,1}$.
        \item $v_{n,1}\leftarrow \bar v'_1$; $v'_1\leftarrow v'_2$; $v'_2\leftarrow v'_1$; $v_{0,8}\leftarrow v_{0,1}$, and $v'_0$ strictly prefers the value 0 in every assignment.
        \item 
        For $1\le i<n$, if
        $\sigma(v_{i,7})=\sigma(u_{i+1,0})$, then $v_{i,8} \leftarrow v_{i,7}$; otherwise $v_{i,8}\leftarrow q_i$.
        For $i=n$, the same statement holds with $v'_0$ in place of $u_{n+1,0}$.
    \end{enumerate}
\end{lemma}

\begin{proof}
    (a) We have
    \begin{align*}
        w(f_1v_{i,2}) &= -\lambda_i < -(\lambda_i/2^2 + \lambda_i/2^2) = -(|w(f_1f_3)| + |w(f_1f_4)|) = -B_{f_1}(v_{i,2}) \\
        w(f_2v_{i,4}) &= -\lambda_i < -(\lambda_i/2^2 + \lambda_i/2^2) = -(|w(f_2f_3)| + |w(f_2f_4)|) = -B_{f_2}(v_{i,4}) \\
        w(g_2v_{i,6}) &= -\lambda_i < -(\lambda_i/2^8 + \lambda_i/2^8) = -(|w(g_2g_3)| + |w(g_2g_4)|) = -B_{g_2}(v_{i,6}) \\
        w(g_1f_z) &= -\lambda_i/2^6 < -(\lambda_i/2^8 + \lambda_i/2^8) = -(|w(g_1g_3)| + |w(g_1g_4)|) = -B_{g_1}(f_z) \\
    \end{align*}
    Hence, by Lemma~\ref{lem:following}, we have $f_1 \leftarrow v_{i,2}$, $f_2 \leftarrow v_{i,4}$, $g_2 \leftarrow v_{i,6}$, and $g_1 \leftarrow f_z$. 
    For the remaining items in \cref{tab:232_v1_gadget}, we apply the following:
    \begin{itemize}
        \item Lemma~\ref{lem:gates}(i) to $f_3,f_4$ with $D=\lambda_{i}/2^2$ and residual incident weight $\lambda_{i}/2^4$; 
        \item Lemma~\ref{lem:gates}(ii) to $f_z$ with $D=\lambda_{i}/2^4$ and residual incident weight $\lambda_{i}/2^6$; 
        \item Lemma~\ref{lem:gates}(i) to $g_3,g_4$ with $D=\lambda_{i}/2^8$ and residual incident weight $\lambda_{i}/2^{10}$; and 
        \item Lemma~\ref{lem:gates}(ii) to $v_{i-1,1}$ with $D=\lambda_{i}/2^{10}=2^3\lambda_{i-1}>7\lambda_{i-1},$ where $7\lambda_{i-1}$ is the absolute weight of $v_{i-1,1}v_{i-1,2}$ for $i > 1$ and of $v_{i-1,1} v_{i-1,8}$ for $i = 1$.
    \end{itemize} 

    This proves every following relation in the table.
    If all variables appearing there cannot be flipped to improve the objective, these relations give
    \begin{align*}
    \sigma(f_{i,z})
     &=\sigma(f_{i,3})\vee\sigma(f_{i,4})
      =\sigma(f_{i,1})\oplus\sigma(f_{i,2})
      =\sigma(v_{i,2})\oplus\sigma(v_{i,4}),\\
    \sigma(v_{i-1,1})
     &=\sigma(g_{i,3})\vee\sigma(g_{i,4})
      =\sigma(g_{i,1})\oplus\sigma(g_{i,2}) =\sigma(v_{i,2})\oplus\sigma(v_{i,4})\oplus\sigma(v_{i,6}),
    \end{align*}
    as claimed.

    (b) We have
    \begin{align*}
        w(u_1u_0) &= -\lambda_{i-1} \le -2^4 < -(1 + 8) = -\left(|w(u_1u_2)| + |w(u_1v_{i,7})|\right) = -B_{u_1}(u_0),\\
        w(u_2u_1) &= -8 < -(6+1) = -\left(|w(u_2u_3)| + |w(u_2v_{i,5})|\right) = -B_{u_2}(u_1),\\
        w(u_3u_2) &= -6 < -(4+1) = -\left(|w(u_3q)| + |w(u_3v_{i,3})|\right) = -B_{u_3}(u_2),\\
        w(qu_3) &= -4 < -(1+1) = -\left(|w(qu_0)| + |w(qv_{i,8})|\right) = -B_{q}(u_3).
    \end{align*}
    Hence, by \cref{lem:following}, we have $u_1\leftarrow u_0$, $u_2\leftarrow u_1$, $u_3\leftarrow u_2$, $q\leftarrow u_3$, as claimed.
    For $u_0$, applying Lemma~\ref{lem:gates}(iii) with $D=\lambda_{i-1}$ yields the desired relations.

    (c) We have
    \begin{align*}
        w(v_{i,2}v_{i,1}) &= -7\lambda_i < -(5\lambda_i + \lambda_i) = -(|w(v_{i,2}v_{i,3})| + |w(v_{i,2}f_{i,1})|) = -B_{v_{i,2}}(v_{i,1}), \\
        w(v_{i,4}v_{i,3}) &= -5\lambda_i < -(3\lambda_i + \lambda_i) = -(|w(v_{i,4}v_{i,5})| + |w(v_{i,4}f_{i,2})|) = -B_{v_{i,4}}(v_{i,3}), \\
        w(v_{i,6}v_{i,5}) &= -3\lambda_i < -(\lambda_i + \lambda_i) = -(|w(v_{i,6}v_{i,7})| + |w(v_{i,6}g_{i,2})|) = -B_{v_{i,6}}(v_{i,5}).
    \end{align*}
    By \cref{lem:following}, we obtain $v_{i,2} \leftarrow {v}_{i,1}$, $v_{i,4} \leftarrow {v}_{i,3}$, and $v_{i,6} \leftarrow {v}_{i,5}$, as claimed.

    (d) The two path edges incident with $v_{i,3}$ have weight $-5\lambda_i$, those incident with $v_{i,5}$ have weight $-3\lambda_i$, and those incident with $v_{i,7}$ have weight $-\lambda_i$.
    The remaining incident edge has absolute weight $1$.
    Thus \cref{lem:gates}(iii) applies at each of these vertices.
    When the path neighbors agree, the vertex follows their common value.
    When they disagree, the sign of the unit edge gives that $v_{i,3}\leftarrow {u}_{i,3}$, $v_{i,5}\leftarrow \bar{u}_{i,2}$, and $v_{i,7}\leftarrow {u}_{i,1}$, as claimed.

    (e) These cases follow from \cref{lem:following} and the following observations.
    \begin{align*}
        w(v_{n,1}v'_1) &= \lambda_{n+1} > 7\lambda_n = |w(v_{n,1}v_{n,2})| = B_{v_{n,1}}(v'_1), \\
        w(v'_1v'_2) &= -2\lambda_{n+1} < -\lambda_{n+1} = -w(v'_1v_{n,1}) = -B_{v'_1}(v'_2), \\
        w(v'_2v'_1) &= -2\lambda_{n+1} < 0 = -B_{v'_2}(v'_1), \\
        w(v_{0,8}v_{0,1}) &= -7\lambda_0 < -\lambda_{0} = -|w(v_{0,8}u_{1,0})| = -B_{v_{0,8}}(v_{0,1}).
    \end{align*}
    
    Finally, $\Delta_{v'_0}=-2\lambda_n-\lambda_n(1-2\sigma(v_{n,8}))
\in\{-3\lambda_n,-\lambda_n\}$.
    So $v'_0$ always strictly prefers $0$.

    (f) When $i<n$, the two main edges at $v_{i,8}$, namely $v_{i,8}v_{i,7}$ and $v_{i,8}u_{i+1,0}$, both have weight $-\lambda_i$, while $v_{i,8}q_i$ has weight $-1$.
    Hence the statement follows from \cref{lem:gates}(iii).
    The case $i=n$ is identical because $v_{n,8}v'_0$ has weight $-\lambda_n$.
\end{proof}

\paragraph{The canonical flip sequence.} 
Recall the sequence $R_n$ defined in \cref{rem:MS-comparison}.
For $i \in [n]$, we define 
\begin{align*}
\rho^{(1)}_i&=f_{i,1}\,f_{i,3}\,f_{i,z}\,g_{i,1}\,g_{i,3}, & \rho^{(2)}_i&=f_{i,2}\,f_{i,3}\,f_{i,z}\,g_{i,1}\,g_{i,3}, & \rho^{(3)}_i&=g_{i,2}\,g_{i,4},\\
\rho^{(4)}_i&=f_{i,1}\,f_{i,4}\,f_{i,z}\,g_{i,1}\,g_{i,4}, & \rho^{(5)}_i&=f_{i,2}\,f_{i,4}\,f_{i,z}\,g_{i,1}\,g_{i,4}, & \rho^{(0)}_i&=g_{i,2}\,g_{i,4},
\end{align*}
\[
\alpha_i=(u_{i,0},u_{i,1},u_{i,2},u_{i,3},q_i).
\]
We obtain $\Omega_n$ from $R_n$ as follows: before the $m$-th occurrence of $v_{i-1,1}$, insert $\rho_i^{(m\bmod6)}$, and after every occurrence of $v_{i-1,8}$, insert $\alpha_i$.
Define $\bar\Omega_n$ in the same way, except that the string inserted before the $m$-th occurrence of $v_{i-1,1}$ is $\rho_i^{((m+3)\bmod6)}$.

Then it is easy to obtain the following explicit recursions.

\begin{observation}
\label{obs:232_omega_recursions}
    We have $\Omega_0 = \bar{\Omega}_0 = R_0$.
    Further, for $n \ge 1$,
    \begin{align*}
        \Omega_n &= (v_{n,1}, v_{n,2}, \rho^{(1)}_n, \Omega_{n-1}, \alpha_n, v_{n,3}, v_{n,4}, \rho^{(2)}_n, \bar{\Omega}_{n-1}, \alpha_n, v_{n,5}, v_{n,6}, \rho^{(3)}_n, \Omega_{n-1}, \alpha_n, v_{n,7}, v_{n,8}) \\
        \bar{\Omega}_n &= (v_{n,1}, v_{n,2}, \rho^{(4)}_n, \bar{\Omega}_{n-1}, \alpha_n, v_{n,3}, v_{n,4}, \rho^{(5)}_n, \Omega_{n-1}, \alpha_n, v_{n,5}, v_{n,6}, \rho^{(0)}_n, \bar{\Omega}_{n-1}, \alpha_n, v_{n,7}, v_{n,8})
    \end{align*}
\end{observation}

Let $T_n=|\Omega_n|=|\bar\Omega_n|$.
The recursions above give $T_0=2$ and $T_n=3T_{n-1}+35$.
Hence, we obtain the following.

\begin{observation}
For $n \geq 0$, $T_n = \frac{39\cdot 3^{n}-35}{2}$.
\end{observation}

We next describe the two phases used in the induction.
Let $\tau_0$ be the assignment on $F_0$ that gives value $1$ to both of its variables.
For $i\ge1$, let $\tau_i$ be the assignment on $F_i$ that gives value $1$ to
\[
\{v_{i,1},\dots,v_{i,8},f_{i,1},f_{i,2},g_{i,2},g_{i,4},u_{i,0},u_{i,1},u_{i,2},u_{i,3},q_i\}
\]
and value $0$ to the remaining variables of $F_i$.
We say $F_0,\dots,F_m$ are in \emph{phase $0$} if all their variables are $0$, and in \emph{phase $1$} if each $F_i$ is assigned according to $\tau_i$.

\begin{lemma}
\label{obs:phase-parity}
For every $0 \le m\le n$, $\Omega_m$ and $\bar\Omega_m$ have the same flip parity on every variable.  
Starting in either phase on $F_0,\dots,F_m$, applying either complete string therefore changes these layers to the other phase.
\end{lemma}
\begin{proof}
We prove by induction.
The lemma statement is immediate for $m=0$.
For $m\ge1$, each path variable $v_{m,j}$ occurs once and each variable of $\alpha_m$ occurs three times.
In both $\Omega_m$ and $\bar\Omega_m$, the three $\rho_m$-blocks flip $f_{m,1},f_{m,2},g_{m,2},g_{m,4}$ an odd number of times and every other $f$- or $g$-variable an even number of times.
The three recursive substrings have, by induction, the same parity on lower layers.
Hence, the net effect of these three is the same as the parity of one substring.
The lemma then follows from all of the above together with the inductive hypothesis.
\end{proof}

\paragraph{The dependency graph.}
We now characterize all (not necessarily maximal) improving flip sequences of $\mathcal{I}_n$ via a directed graph $D_n$.
Intuitively, the vertices of $D_n$ represent occurrences of flips in $\Omega_n$, and not the variables, since a variable such as $q_n$ and $u_{n,j}$ may occur several times in $\Omega_n$.
An edge $(r,s)$ in $D_n$ then represents that the flip at position $r$ must happen before the flip at position $s$.
In other words, $D_n$ encodes a partial order of the flips.

More formally, let $\Omega_n = (\omega_1, \dots, \omega_{T_n})$, and $\bar{\Omega}_n = (\bar{\omega}_1, \dots, \bar{\omega}_{T_n})$.
The vertices of $D_n$ are $[T_n]$, which we interpret as positions in either $\Omega_n$ or $\bar{\Omega}_n$.
For $n=0$, let $D_0$ consist of the edge $(1,2)$.
For $n \ge 1$, we construct $D_n$ on $[T_n]$ as follows, with $t=T_{n-1}$.
\begin{itemize}
    \item For $2 \le s \le T_n$, if $\omega_{s-1} \in F_n \setminus \{q_n\}$, we add the edge $(s-1,s)$.
    \item For $j \in \{7, t + 19, 2t + 28\}$, we add a copy of $D_{n-1}$ shifted by $j$ (i.e., for every edge $(r,s)$ of $D_{n-1}$, we add the edge $(r+j,s+j)$).
    \item We add the edges $(t+7, t+8)$, $(2t+19, 2t+20)$, and $(3t+28, 3t+29)$, which correspond to the three cases from $v_{n-1,8}$ to $u_{n,0}$ in the recursion.
    \item We add the edge $(r,s)$ for the following $r,s$:
    \[ \begin{array}{rcll} 
    r & s & \omega_r & \omega_s \\
    \hline
    t+11 & t+13 & u_{n,3}^{(1)} & v_{n,3}\\
    t+12 & 2t+20 & q_n^{(1)} & u_{n,0}^{(2)}\\ 
    2t+22 & 2t+25 & u_{n,2}^{(2)} & v_{n,5}\\ 
    2t+24 & 3t+29 & q_n^{(2)} & u_{n,0}^{(3)}\\
    3t+30 & 3t+34 & u_{n,1}^{(3)} & v_{n,7} \\
    T_n-2 & T_n & q_n^{(3)} & v_{n,8}
    \end{array} \]
    Here the superscript indicates the first, second, or third occurrence in $\Omega_n$ (which is the same as in $\bar{\Omega}_n$).
\end{itemize}

See \cref{fig:232_dependency_graph} for an illustration.

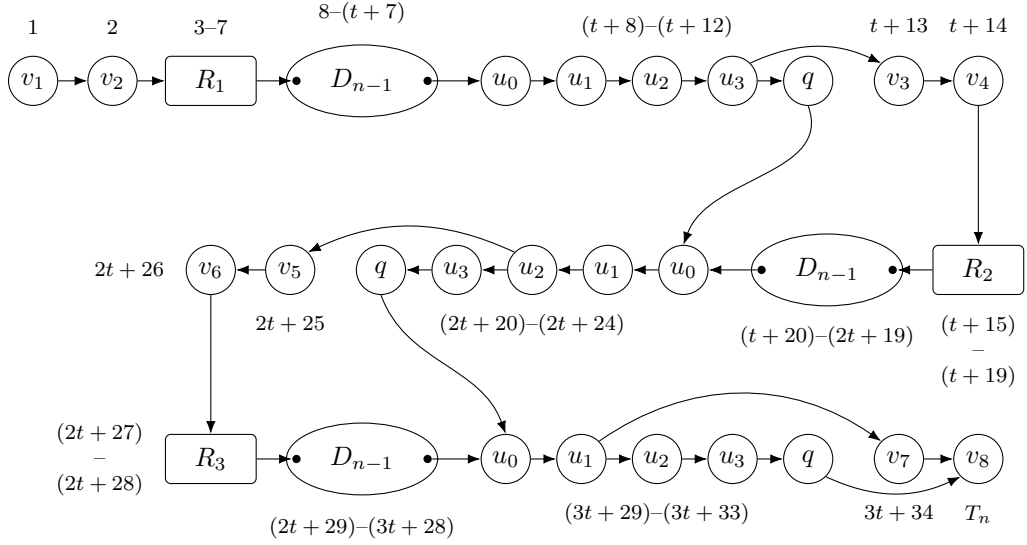
\begin{figure}[t]
\centering

\begin{tikzpicture}[
    >=Latex,
    vertex/.style={
        circle,
        draw,
        minimum size=6.5mm,
        inner sep=.5pt,
        font=\small
    },
    request/.style={
        draw,
        rounded corners=2pt,
        minimum width=12mm,
        minimum height=6.5mm,
        inner sep=1pt,
        font=\small
    },
    blob/.style={
        ellipse,
        draw,
        minimum width=20mm,
        minimum height=9.5mm,
        inner sep=1pt,
        font=\small
    },
    port/.style={
        circle,
        fill,
        inner sep=1.1pt
    },
    dep/.style={
        ->,
        thin
    },
    extra/.style={
        ->,
        thin,
        dashed
    },
    idx/.style={
        font=\scriptsize
    }
]


\node[vertex] (v1) at (0,0) {$v_1$};
\node[idx,above=1.5mm of v1] {$1$};

\node[vertex] (v2) at (1.05,0) {$v_2$};
\node[idx,above=1.5mm of v2] {$2$};

\node[request] (R1) at (2.35,0) {$R_1$};
\node[idx,above=1.5mm of R1] {$3\text{--}7$};

\node[blob] (B1) at (4.35,0) {$D_{n-1}$};
\node[idx,above=1.5mm of B1]
    {$8\text{--}(t+7)$};

\node[port] (s1) at ($(B1.west)+(1.4mm,0)$) {};
\node[port] (z1) at ($(B1.east)+(-1.4mm,0)$) {};

\node[vertex] (u01) at (6.25,0) {$u_0$};
\node[vertex] (u11) at (7.25,0) {$u_1$};
\node[vertex] (u21) at (8.25,0) {$u_2$};
\node[vertex] (u31) at (9.25,0) {$u_3$};
\node[vertex] (q1)  at (10.25,0) {$q$};

\node[idx] at ($(u01)!0.5!(q1)+(0,0.72)$)
    {$(t+8)\text{--}(t+12)$};

\node[vertex] (v3) at (11.45,0) {$v_3$};
\node[idx,above=1.5mm of v3] {$t+13$};

\node[vertex] (v4) at (12.50,0) {$v_4$};
\node[idx,above=1.5mm of v4] {$t+14$};

\draw[dep] (v1) -- (v2);
\draw[dep] (v2) -- (R1);
\draw[dep] (R1) -- (s1);
\draw[dep] (z1) -- (u01);

\draw[dep] (u01) -- (u11);
\draw[dep] (u11) -- (u21);
\draw[dep] (u21) -- (u31);
\draw[dep] (u31) -- (q1);

\draw[dep]
    (u31.north east)
    to[out=25,in=155]
    (v3.north west);

\draw[dep] (v3) -- (v4);


\node[request] (R2) at (12.50,-2.5) {$R_2$};
\node[idx,align=center,below=1.3mm of R2]
    {$(t+15)$\\--\\$(t+19)$};

\node[blob] (B2) at (10.50,-2.5) {$D_{n-1}$};
\node[idx,below=1.2mm of B2,xshift=-0mm]
    {$(t+20)\text{--}(2t+19)$};

\node[port] (s2) at ($(B2.east)+(-1.4mm,0)$) {};
\node[port] (z2) at ($(B2.west)+(1.4mm,0)$) {};

\node[vertex] (u02) at (8.60,-2.5) {$u_0$};
\node[vertex] (u12) at (7.60,-2.5) {$u_1$};
\node[vertex] (u22) at (6.60,-2.5) {$u_2$};
\node[vertex] (u32) at (5.60,-2.5) {$u_3$};
\node[vertex] (q2)  at (4.60,-2.5) {$q$};

\node[idx] at ($(u02)!0.5!(q2)+(0,-0.72)$)
    {$(2t+20)\text{--}(2t+24)$};

\node[vertex] (v5) at (3.40,-2.5) {$v_5$};
\node[idx,below=1.5mm of v5] {$2t+25$};

\node[vertex] (v6) at (2.35,-2.5) {$v_6$};
\node[idx,left=1.5mm of v6] {$2t+26$};

\draw[dep]
    (v4.south) -- (R2.north);

\draw[dep] (R2) -- (s2);
\draw[dep] (z2) -- (u02);

\draw[dep] (u02) -- (u12);
\draw[dep] (u12) -- (u22);
\draw[dep] (u22) -- (u32);
\draw[dep] (u32) -- (q2);

\draw[dep] (v5) -- (v6);

\draw[dep]
    (q1.south)
    to[out=-70,in=70]
    (u02.north);

\draw[dep]
    (u22.north west)
    to[out=155,in=25]
    (v5.north east);


\node[request] (R3) at (2.35,-5.0) {$R_3$};
\node[idx,align=center,left=1.5mm of R3]
    {$(2t+27)$\\--\\$(2t+28)$};

\node[blob] (B3) at (4.35,-5.0) {$D_{n-1}$};
\node[idx,below=1.5mm of B3]
    {$(2t+29)\text{--}(3t+28)$};

\node[port] (s3) at ($(B3.west)+(1.4mm,0)$) {};
\node[port] (z3) at ($(B3.east)+(-1.4mm,0)$) {};

\node[vertex] (u03) at (6.25,-5.0) {$u_0$};
\node[vertex] (u13) at (7.25,-5.0) {$u_1$};
\node[vertex] (u23) at (8.25,-5.0) {$u_2$};
\node[vertex] (u33) at (9.25,-5.0) {$u_3$};
\node[vertex] (q3)  at (10.25,-5.0) {$q$};

\node[idx] at ($(u03)!0.5!(q3)+(0,-0.72)$)
    {$(3t+29)\text{--}(3t+33)$};

\node[vertex] (v7) at (11.45,-5.0) {$v_7$};
\node[idx,below=1.5mm of v7] {$3t+34$};

\node[vertex] (v8) at (12.50,-5.0) {$v_8$};
\node[idx,below=1.5mm of v8] {$T_n$};

\draw[dep]
    (v6.south) -- (R3.north);

\draw[dep] (R3) -- (s3);
\draw[dep] (z3) -- (u03);

\draw[dep] (u03) -- (u13);
\draw[dep] (u13) -- (u23);
\draw[dep] (u23) -- (u33);
\draw[dep] (u33) -- (q3);

\draw[dep] (v7) -- (v8);

\draw[dep]
    (q2.south)
    to[out=-70,in=110]
    (u03.north);

\draw[dep]
    (u13.north east)
    to[out=35,in=145]
    (v7.north west);

\draw[dep]
    (q3.south east)
    to[out=-25,in=-155]
    (v8.south west);

\end{tikzpicture}

\caption{
Recursive construction of $D_n$, where $t=T_{n-1}$.
The vertices of $D_n$ are indices in $[T_n]$; the symbols inside the circles are the corresponding variables in $\Omega_n$, with the subscript $n$ suppressed.
Each ellipse represents a shifted copy of $D_{n-1}$, with only its unique source and sink displayed.
The boxes $R_1,R_2,R_3$ are the paths corresponding to $\rho_n^{(1)},\rho_n^{(2)},\rho_n^{(3)}$ for $\Omega_n$, and to
$\rho_n^{(4)},\rho_n^{(5)},\rho_n^{(0)}$ for $\bar\Omega_n$.
}
\label{fig:232_dependency_graph}
\end{figure}

\begin{lemma}    
\label{lem:232_D_obs}
    The graph $D_n$ is acyclic, has unique source $1$, and every vertex can reach $T_n$.
\end{lemma}
\begin{proof}
All edges go from a smaller to a larger index, so $D_n$ is acyclic.
The remainder of the statement follows by an easy induction on $n$ from the recursive construction.
\end{proof}

\paragraph{All improving sequences.}
For $\varepsilon\in\{0,1\}$, define
\[
\Xi_m^\varepsilon=
\begin{cases}
\Omega_m,&\varepsilon=0,\\
\bar\Omega_m,&\varepsilon=1.
\end{cases}
\]
We use the following statement, which is formulated so that it applies both to the entire construction and to a lower-layer recursive block.

\begin{lemma}
\label{lem:activation}
Fix $0 \le m\le n$ and $\varepsilon\in\{0,1\}$.
Suppose $F_0,\dots,F_m$ are in phase $\varepsilon$, and among the variables in these layers only $v_{m,1}$ is improving.
Suppose further that the values of all neighbors of $F_0\cup\cdots\cup F_m$ outside these layers remain fixed until $v_{m,8}$ flips.
If $m\ge1$, also suppose that the external neighbor of $v_{m,8}$ (i.e., $u_{m+1,0}$ when $m\ne n$, and $v'_0$ when $m=n$) has value $\varepsilon$ and remains fixed until $v_{m,8}$ flips.
Then all maximal improving flip sequences restricted to $F_0,\dots,F_m$ are exactly the strings obtained by labeling the linear extensions of $D_m$ with the entries of $\Xi_m^\varepsilon$.
After all $|\Xi_m^\varepsilon| = T_m$ flips have occurred, the layers $F_0,\dots,F_m$ are in phase $1-\varepsilon$.
\end{lemma}
\begin{proof}
    We prove the statement by induction on $m$.
    For $m=0$, $v_{0,1}$ is the only improving variable by assumption.
    After it flips, \cref{lem:local}(e) makes $v_{0,8}$ improving, and no variable of $F_0$ can improve again.
    Thus the unique improving sequence is $(v_{0,1},v_{0,8})$, which is the unique linear extension of $D_0$.
    The phase changes by \cref{obs:phase-parity}.
    
    Now let $m\ge1$ and assume the statement for $m-1$.
    We follow the three recursive blocks in the definition of $\Xi_m^\varepsilon$.
    At three points in the proof, two groups of flips can proceed independently.
    We call these groups \emph{branches}: their flips may be interleaved, subject to the dependencies within each branch, but both branches must finish before the next flip outside them can occur.
    The flips within each branch must still respect the order specified below.
    
    \subparagraph{First branching stage.}
    Initially, only $v_{m,1}$ is improving.
    By \cref{lem:local}(a),(c), the flips $v_{m,1},\;v_{m,2},\;\rho_m^{(1+3\varepsilon)}$ are forced (i.e., one must happen right after another as there is no other improving flip).
    Note that no variable in $F_0\cup\cdots\cup F_{m-1}$ has flipped during this flip sequence.
    Moreover, among the variables in these lower layers, only $v_{m-1,1}$ is adjacent to a variable whose value has changed: its preference is controlled by the outputs of the gadget for $v_{m-1,1}$ in $F_m$.
    The external neighbor $u_{m,0}$ of $v_{m-1,8}$ has not changed, so every other variable in the lower layers retains its previous preference and hence remains non-improving.
    By \cref{lem:local}(a),
    \[
    v_{m-1,1}\leftarrow
    v_{m,2}\oplus v_{m,4}\oplus v_{m,6}
    =(1-\varepsilon)\oplus\varepsilon\oplus\varepsilon
    =1-\varepsilon.
    \]
    Since the lower layers are still in phase $\varepsilon$, we have $v_{m-1,1}=\varepsilon$.
    Thus $v_{m-1,1}$ is improving, and it is the only improving variable in $F_0\cup\cdots\cup F_{m-1}$.
    The external neighbor of $v_{m-1,8}$ is $u_{m,0}$, which still has value $\varepsilon$ and does not flip before $v_{m-1,8}$.
    Hence, the inductive hypothesis applies and yields exactly a linear extension of $D_{m-1}$ labeled by $\Xi_{m-1}^{\varepsilon}$.
    No variable in $F_m$ becomes improving before the final flip of $v_{m-1,8}$.
    Indeed, the only vertices of $F_m$ that are adjacent to the lower layers are $u_{m,0}$, $g_{m,3}$, and $g_{m,4}$.
    By \cref{lem:local}(a), the last two variables have their preferred values independent of the lower layers.
    Hence, the only possible new improving move in $F_m$ caused by flips in the lower layers is $u_{m,0}$, and the change required at that vertex cannot occur before $v_{m-1,8}$ flips.
    
    After the first recursive block, \cref{lem:local}(b) forces $u_{m,0},\;u_{m,1},\;u_{m,2},\;u_{m,3}$.
    At this point exactly $q_m$ and $v_{m,3}$ are improving: $q_m$ always prefers the value for $u_{m,3}$, while for $v_{m,3}$, since its two neighbors $v_{m,2}$ and $v_{m,4}$ disagree, $v_{m,3} \leftarrow u_{m,3}$.
    At this point, the sequence splits into two branches.
    The flip of $q_m$ affects only the preferences of $u_{m,0}$ and $v_{m,8}$, and in the present configuration it makes neither of them improving.
    On the other hand, starting with $v_{m,3}$ forces $v_{m,3},\;v_{m,4},\;\rho_m^{(2+3\varepsilon)}$.
    At that point the lower layers are in phase $1-\varepsilon$, the external neighbor $u_{m,0}$ of $v_{m-1,8}$ also has value $1-\varepsilon$, and the same XOR argument as above makes $v_{m-1,1}$ the only improving variable in the lower layers.
    Hence, the inductive hypothesis gives a linear extension of $D_{m-1}$ labeled by $\Xi_{m-1}^{1-\varepsilon}$.
    The variable $q_m$ is adjacent to none of the variables in this branch, so its flip may be interleaved arbitrarily with it.
    When the second recursive block is complete, the next occurrence of $u_{m,0}$ is improving if and only if $q_m$ has already flipped; if not, $q_m$ is the only improving variable and must flip first.
    Thus, the possible interleavings are exactly those prescribed by the subgraph of $D_m$ on $\{T_{m-1}+12, \dots, 2T_{m-1}+20\}$ (i.e., the vertices on the path between the first occurrence of $q$ and the second occurrence of $u_0$ in \cref{fig:232_dependency_graph}).
    
    \subparagraph{Second branching stage.}
    After the second occurrence of $u_{m,0}$, \cref{lem:local}(b) forces $u_{m,1}$ and then $u_{m,2}$.
    Now exactly $u_{m,3}$ and $v_{m,5}$ are improving: $u_{m,3}$, because $u_{m,3} \leftarrow u_{m,2}$ by \cref{lem:local}(b); and $v_{m,5}$, because its two neighbors $v_{m,4}$ and $v_{m,6}$ now disagree and hence $v_{m,5} \leftarrow \bar{u}_{m,2}$ by \cref{lem:local}(d).
    Following a similar argument as before, we obtain two resulting branches that start with $(u_{m,3},q_m)$
    and $(v_{m,5},v_{m,6},\rho_m^{(3+3\varepsilon)})$.
    After the flips $v_{m,5}, v_{m,6}, \rho_m^{(3+3\varepsilon)}$ have occurred, the lower layers are in phase $\varepsilon$ and the external neighbor $u_{m,0}$ of $v_{m-1,8}$ also has value $\varepsilon$.
    By the same argument as in the first recursive block, $v_{m-1,1}$ is then the only improving variable in $F_0\cup\cdots\cup F_{m-1}$.
    Hence, the inductive hypothesis applies and continues this branch with a linear extension of $D_{m-1}$ labeled by $\Xi_{m-1}^{\varepsilon}$.
    The two branches are independent and hence may be interleaved arbitrarily, subject only to the displayed order within each branch and to the internal order of the recursive block.
    After the third recursive block, the next occurrence of $u_{m,0}$ is improving exactly when the branch $(u_{m,3},q_m)$ has been completed; otherwise its unfinished suffix is forced first.
    Hence, both branches must be completed before the next flip of $u_{m,0}$, as encoded in $D_m$.
    
    \subparagraph{Final branching stage.}
    After the third occurrence of $u_{m,0}$, the flip of $u_{m,1}$ is forced.
    Then exactly $u_{m,2}$ and $v_{m,7}$ are improving.
    The resulting independent branches are $(u_{m,2},u_{m,3},q_m)$ and $(v_{m,7})$.
    They may again be interleaved arbitrarily, respecting the order in the first branch.
    By \cref{lem:local}(f), $v_{m,8}$ becomes improving exactly after both branches have completed, and it is then the only improving variable.
    After $v_{m,8}$ flips, no variable in $F_0\cup\cdots\cup F_m$ is improving.
    
    The forced portions of the sequence give the consecutive top-layer dependencies in $D_m$, while the three applications of the inductive hypothesis give its three copies of $D_{m-1}$.
    Finally, the dependencies between branches account for exactly the six additional edges listed in the definition of $D_m$.
    Hence, the possible improving flip sequences are precisely the strings obtained by labeling the linear extensions of $D_m$ with $\Xi_m^\varepsilon$.
    Finally, \cref{obs:phase-parity} shows that after all $T_m$ flips, the layers $F_0,\dots,F_m$ are in phase $1-\varepsilon$.
\end{proof}

\begin{lemma}
\label{lem:232-exp}
    The map $(s_1, \dots, s_{T_n}) \longmapsto (\omega_{s_1}, \dots, \omega_{s_{T_n}})$ is a bijection between all linear extensions of $D_n$ and all maximal improving flip sequences of $\mathcal{I}_n$ starting from the initial solution $\sigma_0$.
\end{lemma}
\begin{proof}
    At $\sigma_0$, by \cref{lem:local}, it is easy to check that only $v_{n,1}$ is improving.
    Moreover, the external second input of $v_{n,8}$ is $v'_0=0$ and, by \cref{lem:local}(e), it remains $0$ forever.
    Hence, the hypotheses of \cref{lem:activation} hold with $m=n$ and $\varepsilon=0$.
    It follows that all maximal improving flip sequences restricted to $F_0, \dots, F_n$ are exactly the linear extensions of $D_n$ labeled by the entries of $\Omega_n$.

    By \cref{lem:local}(e), $v'_1$ and $v'_2$ follow each other, while $v'_0$ strictly prefers $0$.
    Hence, they remain $0$ forever.
    As a result, the maximal improving flip sequences restricted to $F_0, \dots, F_n$ are exactly the maximal improving flip sequences for $\mathcal{I}_n$.
    The lemma then follows.
\end{proof}

\begin{corollary}
\label{cor:232-all-exp}
    All maximal improving flip sequences of $\mathcal{I}_n$ from $\sigma_0$ have length $T_n = \frac{39\cdot3^n-35}{2}$ and end at the same local optimum, namely the assignment in which every layer $F_i$ is in phase $1$ and $v'_0=v'_1=v'_2=0$.
\end{corollary}

Consequently, we obtain the desired all-exp property.

\thmtwothreetwo*
\begin{proof}
    The graph for $\mathcal{I}_n$ has maximum degree three by construction.
    Hence, by \cref{lem:mca-form}, the constructed instance is a $(2,3,2)$-MCA instance.
    Note that the instance has $O(n)$ variables and each weight can be encoded by $O(n)$ bits.
    By \cref{cor:232-all-exp}, every maximal improving sequence has length $\Theta(3^n)$.
    The theorem then follows.
\end{proof}

\section*{AI methodology}
We obtained our three main results, \cref{thm:2opt,thm:322,thm:232}, without the use of AI tools. We used ChatGPT 5.6 Sol to generate Python code for optimizing the weights and sizes of the constructions in our proofs in Sections~\ref{sec:322} and~\ref{sec:2opt}; this yielded slightly simpler constructions but did not affect the statements of the results. 
We also used Claude Fable 5.1 to optimize the construction in Section~\ref{sec:232}.
All proofs were carried out by the authors, who take full responsibility for their correctness.

\bibliographystyle{alphaurl}
\bibliography{refs}

@article{Ben1992,
  author = {Jon Jouis Bentley},
  title = {Fast Algorithms for Geometric Traveling Salesman Problems},
  journal = {ORSA Journal on Computing},
  volume = {4},
  number = {4},
  pages = {387--411},
  year = {1992},
  doi = {10.1287/ijoc.4.4.387}
}

@article{CKT1999,
  author = {Chandra, Barun and Karloff, Howard and Tovey, Craig},
  title = {New Results on the Old $k$-opt Algorithm for the Traveling Salesman Problem},
  journal = {SIAM Journal on Computing},
  volume = {28},
  number = {6},
  pages = {1998--2029},
  year = {1999},
  doi = {10.1137/S0097539793251244}
}

@article{Cro1958,
  author = {G. A. Croes},
  title = {A Method for Solving Traveling-Salesman Problems},
  journal = {Operations Research},
  volume = {6},
  number = {6},
  pages = {791--812},
  year = {1958},
  doi = {10.1287/opre.6.6.791}
}

@article{DM2013,
  author  = {Dominic Dumrauf and Burkhard Monien},
  title   = {On the {$\cal PLS$}-complexity of maximum constraint assignment},
  journal = {Theoretical Computer Science},
  volume  = {469},
  pages   = {24--52},
  year    = {2013},
  doi     = {10.1016/j.tcs.2012.10.044}
}

@techreport{HH2023,
  author = {Hoang, Hung P. and Hougardy, Stefan},
  title = {On the {PLS}-completeness of {TSP}/$k$-{O}pt},
  institution = {Research Institute for Discrete Mathematics, University of Bonn},
  year = {2023},
  number = {No: 231275}
}

@inproceedings{HHH2024,
  author={Sophia Heimann and Hung P. Hoang and Stefan Hougardy},
     TITLE = {The $k$-{O}pt algorithm for the {T}raveling {S}alesman {P}roblem
              has exponential running time for {$k\geq5$}},
 BOOKTITLE = {51st {I}nternational {C}olloquium on {A}utomata, {L}anguages,
              and {P}rogramming (ICALP 2024)},
    series = {Leibniz International Proceedings in Informatics (LIPIcs)},
    VOLUME = {297},
     PAGES = {84:1--84:18},
      YEAR = {2024},
       DOI = {10.4230/lipics.icalp.2024.84}
}

@article{HHH2026b,
  author={Sophia Heimann and Hung P. Hoang and Stefan Hougardy},
  title = {Exponential Local Search in Bounded-Occurrence {MAX-CSP}}, 
  journal = {arXiv:2610.16898v1 [cs.DS]},
  year = {2026},
  doi = {10.48550/arXiv.2610.16898}
}

@inproceedings{HHH2026,
  author = {Sophia Heimann and Hung P. Hoang and Stefan Hougardy},
  title = {A Near-Complete Resolution of the Exponential-Time Complexity of $k$-opt for the Traveling Salesman Problem},
  booktitle = {Proceedings of the 2026 Annual ACM-SIAM Symposium on Discrete Algorithms (SODA 2026)},
  pages = {5861--5885},
  year = {2026},
  doi = {10.1137/1.9781611978971.208},
}

@incollection {JM1997,
    AUTHOR = {Johnson, David S. and McGeoch, Lyle A.},
     TITLE = {The traveling salesman problem: a case study},
 BOOKTITLE = {Local Search in Combinatorial Optimization},
    EDITOR = {Emile Aarts and Jan Karel Lenstra},
     PAGES = {215--310},
   chapter = {8},
 PUBLISHER = {Wiley},
      YEAR = {1997}
}

@incollection{JM2002,
  title = {Experimental analysis of heuristics for the {STSP}},
  author = {David S. Johnson  and Lyle A. McGeoch},
  booktitle = {The traveling salesman problem and its variations},
  pages = {369--443},
  chapter = {9},
  year = {2002},
  editor = {Gregory Gutin and Abraham P. Punnen},
  publisher = {Kluwer Academic Publishers}
}

@article{JPY1988,
  title = {How easy is local search?},
  author = {Johnson, David S. and Papadimitriou, Christos H. and Yannakakis, Mihalis},
  journal = {Journal of Computer and System Sciences},
  volume = {37},
  number = {1},
  pages = {79--100},
  year = {1988},
  publisher = {Academic Press},
  doi = {10.1016/0022-0000(88)90046-3}
}

@phdthesis{Kaz2020,
    author = {Artem Kaznatcheev},
    title = {Algorithmic Biology of Evolution and Ecology},
    school = {Oriel College, University of Oxford},
    year = {2020}
}

@article{KCJ2020,
  author  = {Artem Kaznatcheev and David A. Cohen and Peter G. Jeavons},
  title   = {Representing Fitness Landscapes by Valued Constraints
             to Understand the Complexity of Local Search},
  journal = {Journal of Artificial Intelligence Research},
  volume  = {69},
  pages   = {1077--1102},
  year    = {2020},
  doi     = {10.1613/jair.1.12156}
}

@inproceedings{Kre1989,
  title = {Structure in locally optimal solutions},
  author = {Mark W. Krentel},
  booktitle = {30th Annual Symposium on Foundations of Computer Science (FOCS 1989)},
  pages = {216--221},
  year = {1989},
  organization = {IEEE Computer Society}
}

@inproceedings{MR4967954,
    AUTHOR = {Kaznatcheev, Artem and Vazquez Alferez, Sofia},
     TITLE = {Greed is slow on sparse graphs of oriented valued constraints},
 BOOKTITLE = {31st {I}nternational {C}onference on {P}rinciples and
              {P}ractice of {C}onstraint {P}rogramming (CP 2025)},
    series = {Leibniz International Proceedings in Informatics (LIPIcs)},
    VOLUME = {340},
     pages =	{18:1--18:13},
    editor = {de la Banda, Maria Garcia},
      YEAR = {2025},
       DOI = {10.4230/lipics.cp.2025.18},
}

@unpublished{Lue1975,
    author = {George Lueker},
    title = {Notes on two-changes},
    year = {1975},
    note = {Unpublished manuscript, Princeton University}
}

@article {MS2023,
  author = {Michel, Lukas and Scott, Alex},
  title = {Superpolynomial smoothed complexity of 3-{FLIP} in {L}ocal {M}ax-{C}ut}, 
  JOURNAL = {arXiv:2310.19594v3 [cs.DS]},
  year = {2023},
  eprint = {2310.19594}
}

@InProceedings {MT2010,
    AUTHOR = {Monien, Burkhard and Tscheuschner, Tobias},
    editor = {Calamoneri, Tiziana and Diaz, Josep},
     TITLE = {On the power of nodes of degree four in the local max-cut problem},
 BOOKTITLE = {Algorithms and Complexity},
    SERIES = {Lecture Notes in Comput. Sci.},
    VOLUME = {6078},
     PAGES = {264--275},
 PUBLISHER = {Springer Berlin Heidelberg},
      YEAR = {2010},
       DOI = {10.1007/978-3-642-13073-1\_24},
}

@article {MR2571908,
    AUTHOR = {Scott, Alexander D. and Sorkin, Gregory B.},
     TITLE = {Polynomial constraint satisfaction problems, graph bisection,
              and the {I}sing partition function},
   JOURNAL = {ACM Trans. Algorithms},
  FJOURNAL = {ACM Transactions on Algorithms},
    VOLUME = {5},
      YEAR = {2009},
    NUMBER = {4},
     PAGES = {Art. 45, 27},
       DOI = {10.1145/1597036.1597049}
}

@article{SY1991,
  author={Sch{\"a}ffer, Alejandro A. and Yannakakis, Mihalis},
  title={Simple local search problems that are hard to solve},
  journal={SIAM Journal on Computing},
  volume={20},
  number={1},
  pages={56--87},
  year={1991},
  publisher={SIAM},
  doi={10.1137/0220004}
}

@article{HZZ2020,
  author = {Stefan Hougardy and Fabian Zaiser and Xianghui Zhong},
  title = {The approximation ratio of the 2-{O}pt {H}euristic for the metric {T}raveling {S}alesman {P}roblem},
  journal = {Operations Research Letters},
  volume = {48},
  number = {4},
  pages = {401--404},
  year = {2020},
  doi = {10.1016/j.orl.2020.05.007}
}

@inproceedings{Zho2020,
  title={On the Approximation Ratio of the $k$-Opt and {L}in-{K}ernighan Algorithm for Metric and Graph {TSP}},
  author={Zhong, Xianghui},
  booktitle={28th Annual European Symposium on Algorithms (ESA 2020)},
  pages={83:1--83:13},
  year={2020},
  editor={Fabrizio Grandoni and Grzegorz Herman and Peter Sanders},
  doi={10.4230/LIPIcs.ESA.2020.83},
  organization={Schloss Dagstuhl--Leibniz-Zentrum f{\"u}r Informatik}
}
\end{document}